\documentclass[onecolumn,aps,nofootinbib,superscriptaddress,tightenlines,12pt,pra]{revtex4-1}
\usepackage{amsmath,amssymb}
\usepackage{amsthm}
\usepackage{bm}
\usepackage{amsfonts}
\usepackage{color}
\usepackage{mathtools}
\usepackage{enumerate}
\usepackage{overpic}
\usepackage{bbm}

\allowdisplaybreaks

\usepackage[ colorlinks = true,
             linkcolor = blue,
             urlcolor  = blue,
             citecolor = red,
             anchorcolor = green,
]{hyperref}

\newtheorem{lemma}{Lemma}
\newtheorem{prop}[lemma]{Proposition}
\newtheorem{theorem}[lemma]{Theorem}
\newtheorem{definition}[lemma]{Definition}
\newtheorem{corollary}[lemma]{Corollary}
\newtheorem{remark}[lemma]{Remark}

\DeclareMathOperator{\tr}{Tr}
\DeclareMathOperator{\supp}{supp}

\newcommand{\cI}{\mathcal{I}}
\newcommand{\cR}{\mathcal{R}}

\newcommand{\cS}{\mathcal{S}}
\newcommand{\cL}{\mathcal{L}}

\newcommand{\cP}{\mathcal{P}}

\newcommand{\cN}{\mathcal{N}}
\newcommand{\cA}{\mathcal{A}}
\newcommand{\cM}{\mathcal{M}}
\newcommand{\cV}{\mathcal{V}}
\newcommand{\cQ}{\mathcal{Q}}
\newcommand{\cU}{\mathcal{U}}
\newcommand{\eps}{\varepsilon}

\begin{document}

\title{Amortized Channel Divergence for\\Asymptotic Quantum Channel Discrimination}

\author{Mark M.~Wilde}
\affiliation{Hearne Institute for Theoretical Physics, Department of Physics and Astronomy, and Center for Computation and Technology, Louisiana State University, Baton Rouge, Louisiana 70803, USA}

\author{Mario Berta}
\affiliation{Department of Computing, Imperial College London, United Kingdom}

\author{Christoph Hirche}
\affiliation{F\'{\i}sica Te\`{o}rica: Informaci\'{o} i Fen\`{o}mens Qu\`{a}ntics, Departament de F\'{i}sica, Universitat Aut\`{o}noma de Barcelona}

\author{Eneet Kaur}
\affiliation{Hearne Institute for Theoretical Physics, Department of Physics and Astronomy, and Center for Computation and Technology, Louisiana State University, Baton Rouge, Louisiana 70803, USA}

\begin{abstract}
It is well known that for the discrimination of classical and quantum channels in the finite, non-asymptotic regime, adaptive strategies can give an advantage over non-adaptive strategies. However, Hayashi [IEEE Trans.~Inf.~Theory 55(8), 3807 (2009)] showed that in the asymptotic regime, the exponential error rate for the discrimination of classical channels is not improved in the adaptive setting. We extend this result in several ways. First, we establish the strong Stein's lemma for classical-quantum channels by showing that asymptotically the exponential error rate for classical-quantum channel discrimination is not improved by adaptive strategies. Second, we recover many other classes of channels for which adaptive strategies do not lead to an asymptotic advantage. Third, we give various converse bounds on the power of adaptive protocols for general asymptotic quantum channel discrimination. Intriguingly, it remains open whether adaptive protocols can improve the exponential error rate for quantum channel discrimination in the asymmetric Stein setting. Our proofs are based on the concept of amortized distinguishability of quantum channels, which we analyse using data-processing inequalities.
\end{abstract}

\maketitle

\tableofcontents


\section{Introduction}

A fundamental task in quantum statistics is to distinguish between two (or multiple) non-orthogonal quantum states. After considerable efforts, the resource trade-off is by now well understood in the information-theoretic limit of asymptotically many copies and quantified by quantum Stein's lemma~\cite{HP91,ON00}, the quantum Chernoff bound~\cite{NS09,ACMBMAV07}, as well as refinements thereof~\cite{Nagaoka06,Audenaert2008,Mosonyi2015}.

As a natural extension of quantum state discrimination, we study here the task of distinguishing between two quantum channels, in the information-theoretic limit of asymptotically many repetitions. Whereas the mathematical properties of states and channels are strongly intertwined, channel discrimination is qualitatively different from state discrimination for a variety of reasons. Most importantly, when distinguishing between two quantum channels one can employ adaptive protocols that make use of a quantum memory~\cite{CDP08a}. The physical scenario in which such adaptive protocols apply consists of a discriminator being given ``black-box'' access to $n$ uses of a channel $\mathcal{N}$ or $\mathcal{M}$, and there is no physical constraint on the kind of operations that he is allowed to perform. In particular, the discriminator is allowed to prepare a quantum state with a quantum memory register that is arbitrarily large, perform adaptive quantum channels with arbitrarily large input and output quantum memories between every call to $\mathcal{N}$ or $\mathcal{M}$, and finally perform an arbitrary quantum measurement on the final state. See Figure~\ref{fig:adaptive-prot} for a graphical depiction.

\begin{figure}[hb]
\begin{center}
\includegraphics[
width=4.3399in
]
{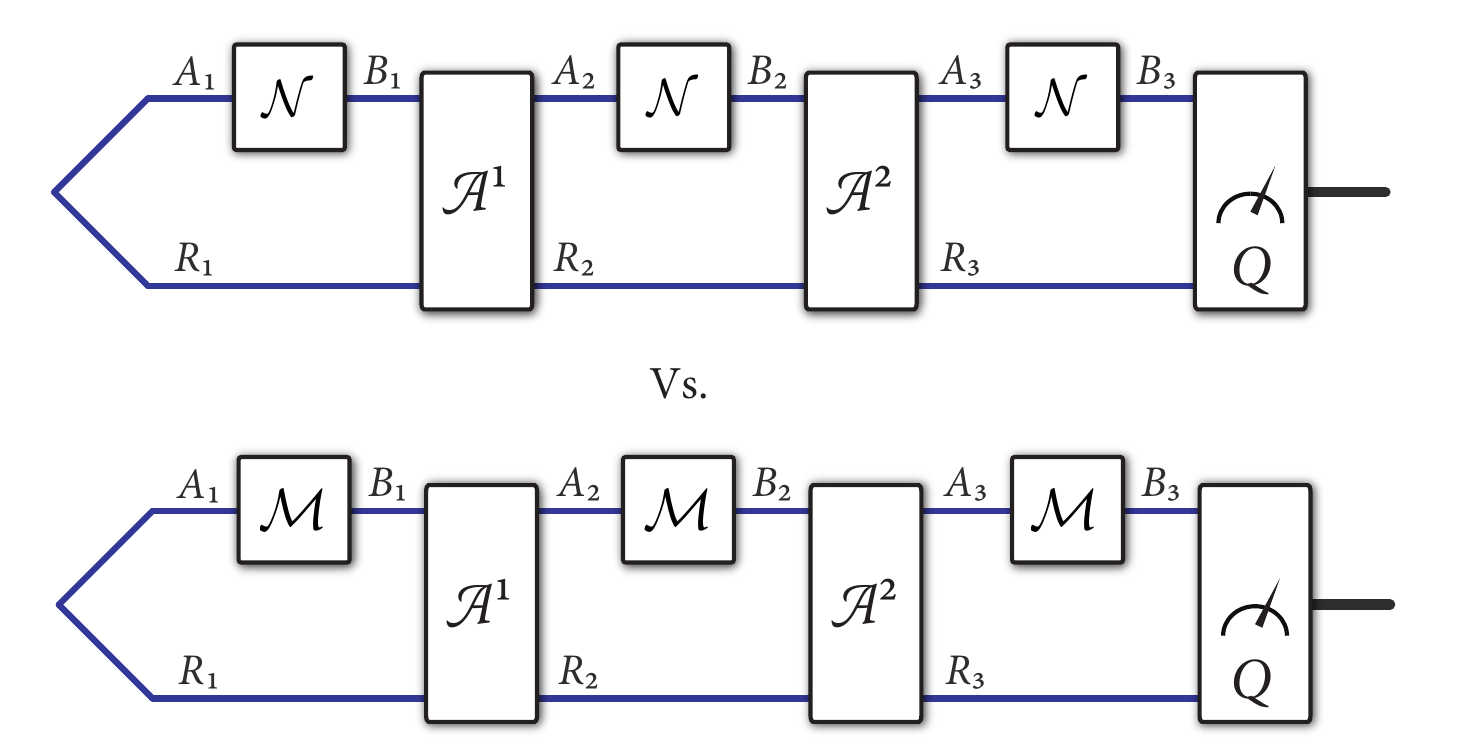}
\end{center}
\caption{A general protocol for channel discrimination when the channel $\mathcal{N}$ or $\mathcal{M}$ is called three times.}
\label{fig:adaptive-prot}
\end{figure}

For the finite, non-asymptotic regime, such protocols are then also known to give an advantage over non-adaptive protocols, the latter of which are restricted to picking a fixed input state and then executing standard state discrimination for the channel outputs.
For an in-depth discussion of this phenomenon, we refer to the latest works~\cite{Duan09,Harrow10} and references therein (also see the recent \cite{Puzzuoli2017}, \cite{NPPZ18}, \cite{PPKK18}). In fact, the advantage of adaptive protocols in this regime already manifests itself for the discrimination of classical channels~\cite[Section 5]{Harrow10}. Somewhat surprisingly, however, Hayashi showed that this advantage disappears for classical channel discrimination in the information-theoretic limit of a large number of repetitions~\cite{Hayashi09}. In particular, the optimal exponential error rate for the discrimination of classical channels in the sense of Stein and Chernoff is achieved by just picking a large number of copies of the best possible product state input and then performing state discrimination for the product output states.

In contrast, in the quantum case, asymptotic channel discrimination has been studied much less systematically than the aforementioned finite, non-asymptotic regime. Notable exceptions include~\cite{Cooney2016} involving replacer channels and \cite{PhysRevLett.118.100502,TW2016} about jointly teleportation-simulable channels. Moreover, references~\cite{Yu17,Pirandola18} feature bounds for general quantum channels, but the exact quantitative performance of these bounds remains rather unclear in the asymptotic setting. We would also like to point to the very related quantum strategies framework of \cite{GW07,G09,G12}, as well as the quantum tester framework of \cite{CDP08b,CDP08a}.

In this paper, we extend some of the seminal classical results~\cite{Hayashi09} to the quantum setting by providing a framework for deriving upper bounds on the power of adaptive protocols for asymptotic quantum channel discrimination. In particular, in order to quantify the largest distinguishability that can be realized between two quantum channels, we introduce the concept of amortized channel divergence. This then allows us to give converse bounds for adaptive channel discrimination protocols in the asymmetric hypothesis testing setting in the sense of Stein, as well as in the symmetric hypothesis testing setting in the sense of Chernoff. Now, whenever the amortized channel divergences collapse to the standard channel divergences~\cite{PhysRevA.97.012332}, we immediately get single-letter converse bounds on the power of adaptive protocols for channel discrimination. Most importantly, we arrive at the characterization of the strong Stein's lemma for classical-quantum channel discrimination. Namely, as a full extension of the corresponding classical result~\cite[Corollary 1]{Hayashi09}, we have that picking many copies of the best possible product-state input and then applying quantum Stein's lemma for the product output states is asymptotically optimal. Other examples with tight characterizations include unitary and isometry channels~\cite{Duan07,Duan09}, projective measurements~\cite{Duan06}, replacer channels~\cite{Cooney2016}, as well as environment-parametrized channels that are environment seizable, as given here in Definition~\ref{def:seizable} (the latter including the channels considered in \cite{PhysRevLett.118.100502,TW2016}).

Intriguingly, we have to leave open the question of whether adaptive protocols improve the exponential error rate for quantum channel discrimination in the asymmetric Stein setting. Even though we provide many classes of channels for which adaptive protocols do not give an advantage in the asymptotic limit, we suspect that in general such a gap exists. We emphasise that this might already occur for entanglement breaking channels or even quantum-classical channels (measurements). Moreover, this would also be consistent with the known advantage of adaptive protocols in the symmetric Chernoff setting~\cite{Duan09,Harrow10,Duan16}. From a learning perspective and following Hayashi's comments for the classical case~\cite[Section 1]{Hayashi09}, this leaves open the possibility that quantum memory is asymptotically helpful for designing active learning protocols for inferring about unknown parameters of quantum systems.

Our paper is structured as follows. In Section~\ref{sec:notation}, we introduce our notation, and in Section~\ref{sec:settings}, we give the precise information-theoretic settings for asymptotic quantum channel discrimination. As our main technical tool, we then introduce amortized channel divergences and analyse their mathematical properties in Section~\ref{sec:technical-tools}. Based on this framework, we proceed to present various converse bounds on the power of adaptive protocols for quantum channel discrimination in Section~\ref{sec:general-bounds}. This is followed by our main result in Section~\ref{sec:cq}, the strong Stein's lemma for classical-quantum channel discrimination. Section~\ref{sec:examples} discusses various other examples for which tight characterisations are available. We end with Section~\ref{sec:conclusion}, where we conclude and discuss open questions.


\section{Notation}\label{sec:notation}

Here we introduce our notation and give the relevant definitions needed later.


\subsection{Setup}

Throughout, quantum systems are denoted by $A$, $B$, and $C$ and have finite dimensions $|A|$, $|B|$, and $|C|$, respectively. Linear operators acting on system $A$ are denoted by $L_A\in\cL(A)$ and positive semi-definite operators by $P_A\in\cP(A)$. Quantum states of system $A$ are denoted by $\rho_A\in\cS(A)$ and pure quantum states by $\Psi_A\in\cV(A)$. 
A maximally entangled state $\Phi_{RA}$ of Schmidt rank $d$ is given by
\begin{equation}
 \Phi_{RA}\coloneqq \frac{1}{d}
 \sum_{i,j=1}^d |i\rangle \langle j |_R
 \otimes |i\rangle \langle j |_A,
\end{equation}
where $\{|i\rangle_R\}_i$
and $\{|i\rangle_A\}_i$ are orthonormal bases.
Quantum channels are completely positive and trace-preserving maps from $\cL(A)$ to $\cL(B)$ and denoted by $\cN_{A\to B}\in\cQ(A\to B)$. The Choi state of a quantum channel $\cN_{A\to B}$ is a standard concept in quantum information and is defined as $\cN_{A\to B}( \Phi_{RA})$. Classical systems are denoted by $X$, $Y$, and $Z$ and have finite dimensions $|X|$, $|Y|$, and $|Z|$, respectively. For $p\geq1$ the Schatten norms are defined for $L_A\in\cL(A)$ as
\begin{align}
\|L_A\|_p\coloneqq\Big(\tr\big[|L_A|^p\big]\Big)^{1/p}.
\end{align}

In this work, we also consider superchannels \cite{CDP08}, which are linear maps that take as input a quantum channel and output a quantum channel. Such superchannels have previously been considered in various contexts in quantum information theory \cite{LM15,WFD17,CG18,G18}. To define them, let $\mathcal{L}(A\rightarrow B)$ denote the set of all linear maps from $\mathcal{L}(A)$ to $\mathcal{L}(B)$. Similarly let $\mathcal{L}(C\rightarrow
D)$ denote the set of all linear maps from $\mathcal{L}(C)$ to $\mathcal{L}(D)$.
Let $\Theta:\mathcal{L}(A\rightarrow B)\rightarrow\mathcal{L}(C\rightarrow D)$
denote a linear supermap, taking $\mathcal{L}(A\rightarrow B)$ to $\mathcal{L}(C\rightarrow D)$. A quantum channel is a particular kind of linear map, and any linear supermap $\Theta$ that takes as input an arbitrary quantum channel
$\Psi_{A\rightarrow B}\in\mathcal{Q}(A\rightarrow B)$ and is required to output a quantum
channel $\Phi_{C\rightarrow D}\in\mathcal{Q}(C\rightarrow D)$ should
preserve the properties of complete positivity and trace preservation. Any
such transformation that does so is called a superchannel. In \cite{CDP08}, it
was proven that any superchannel $\Theta:\mathcal{L}(A\rightarrow
B)\rightarrow\mathcal{L}(C\rightarrow D)$ can be physically realized as follows. If
\begin{equation}
\Phi_{C\rightarrow D} = \Theta\lbrack\Psi_{A\rightarrow B}]
\end{equation}
for an arbitrary
input channel $\Psi_{A\rightarrow B}\in\mathcal{Q}(A\rightarrow B)$ and some output channel $\Phi_{C\rightarrow D} \in \mathcal{Q}(C\rightarrow D)$, then the physical realization of the superchannel 
$\Theta$ is as follows:
\begin{equation}\label{eqn:superchannel}
\Phi_{C\rightarrow D}=\Omega_{BE\rightarrow D}\circ\left(  \Psi_{A\rightarrow B}\otimes\mathcal{I}_{E}\right)  \circ\Lambda_{C\rightarrow AE},
\end{equation}
where $\Lambda_{C\rightarrow AE}:\mathcal{L}(C)\rightarrow \mathcal{L}(AE)$ is a
pre-processing channel, system $E$ corresponds to some memory or environment
system, and $\Omega_{BE\rightarrow D}:\mathcal{L}(BE)\rightarrow \mathcal{L}(D)$ is a
post-processing channel. 


\subsection{Quantum entropies}

The quantum relative entropy for $\rho,\sigma\in\cS(A)$ is defined as \cite{Ume62}
\begin{align}
D(\rho\|\sigma)\coloneqq\begin{cases} \tr\big[\rho\left(\log\rho-\log\sigma\right)\big] \qquad & \supp(\rho)\subseteq\supp(\sigma)\\ +\infty & \text{otherwise,}\end{cases}
\end{align}
where in the above and throughout the paper we employ the convention that all logarithms are evaluated using base two.
The Petz-R\'enyi divergences are defined for $\rho,\sigma\in\cS(A)$ and $\alpha\in(0,1)\cup(1,\infty)$ as \cite{P85,P86}
\begin{align}
D_\alpha(\rho\|\sigma)\coloneqq\frac{1}{\alpha-1}\log\tr\left[\rho^\alpha\sigma^{1-\alpha}\right],
\end{align}
whenever either $\alpha\in(0,1)$ and $\rho$ is not orthogonal to $\sigma$ in Hilbert-Schmidt inner product or $ \alpha>1$ and $\supp(\rho)\subseteq\supp(\sigma)$. Otherwise, we set $D_\alpha(\rho\|\sigma)\coloneqq +\infty$.
In the above and throughout the paper, we employ the convention that inverses are to be understood as generalized inverses.
For $\alpha\in\{0,1\}$, we define the Petz-R\'enyi divergence in the limit as
\begin{align}
D_0(\rho\|\sigma)& \coloneqq\lim_{\alpha\to0}D_\alpha(\rho\|\sigma)=-\log
\tr\left[ \Pi_{\rho} \sigma\right],
\\
D_1(\rho\|\sigma)& \coloneqq\lim_{\alpha\to1}D_\alpha(\rho\|\sigma)=D(\rho\|\sigma),
\end{align}
where $\Pi_{\rho}$ denotes the projection onto the support of $\rho$.
Another quantity of interest related to the Petz-R\'enyi divergences is the Chernoff divergence~\cite{ACMBMAV07,Audenaert2008,NS09}
\begin{align}
C(\rho\|\sigma) & \coloneqq-\inf_{0\leq\alpha\leq1}\log\tr\left[\rho^\alpha\sigma^{1-\alpha}\right]\\
& = \sup_{0\leq\alpha\leq1} (1-\alpha) D_\alpha(\rho \| \sigma). \label{eq:Chernoff-div-expr-renyi}
\end{align}
The sandwiched R\'enyi divergences are defined for $\rho,\sigma\in\cS(A)$ and $\alpha\in(0,1)\cup(1,\infty)$  as~\cite{muller2013quantum,WWY14}
\begin{align}
\widetilde{D}_\alpha(\rho\|\sigma)\coloneqq\frac{1}{\alpha-1}\log\tr\left[\left(\sigma^{\frac{1-\alpha}{2\alpha}}\rho\sigma^{\frac{1-\alpha}{2\alpha}}\right)^\alpha\right]
\end{align}
whenever either $ \alpha\in(0,1)$ and $\rho$ is not orthogonal to $\sigma$ in Hilbert-Schmidt inner product or $\alpha>1$ and $\supp(\rho)\subseteq\supp(\sigma)$. Otherwise we set $\widetilde{D}_\alpha(\rho\|\sigma)\coloneqq\infty$. For $\alpha=1$, we define the sandwiched R\'enyi relative entropy in the limit as \cite{muller2013quantum,WWY14}
\begin{align}
\widetilde{D}_1(\rho\|\sigma)\coloneqq\lim_{\alpha\to1}\widetilde{D}_\alpha(\rho\|\sigma)=D(\rho\|\sigma).
\end{align}
We have that
\begin{align}
\widetilde{D}_{1/2}(\rho\|\sigma)=-\log F(\rho,\sigma),
\end{align}
with Uhlmann's fidelity defined as $F(\rho,\sigma)\coloneqq\|\sqrt{\rho}\sqrt{\sigma}\|_1^2$ \cite{U76}.
In the limit $\alpha\to\infty$, the sandwiched R\'enyi relative entropy converges to the max-relative entropy~\cite{Datta09,Jain02}
\begin{align}
D_{\max}(\rho\|\sigma) \coloneqq\widetilde{D}_\infty(\rho\|\sigma)& \coloneqq\lim_{\alpha\to\infty}\widetilde{D}_\alpha(\rho\|\sigma)
\\
& =\log \left\|\sigma^{-1/2}\rho\sigma^{-1/2}\right\|_\infty \\
& =\inf\left\{\lambda:\rho\leq2^{\lambda}\cdot\sigma\right\},
\end{align}
as shown in \cite{muller2013quantum}. The log-Euclidean R\'enyi divergence is defined for positive definite density operators $\rho,\sigma\in \cS(A)$ as
\cite{Mosonyi2017} (see also \cite{AN00,HP93,ON00,Nagaoka06,AD15})
\begin{equation}
D^\flat_{\alpha}(\rho\| \sigma) \coloneqq
\frac{1}{\alpha - 1}
\log \tr[\exp(\alpha \ln \rho + (1-\alpha) \ln \sigma) ],
\end{equation}
and for general density operators as the following limit: $\lim_{\eps \to 0} D^\flat_{\alpha}(\rho + \eps I\| \sigma + \eps I)$. An explicit expression for the limiting value above is available in \cite[Lemma~3.1]{Mosonyi2017}. In the limit $\alpha\to 1$, the log-Euclidean R\'enyi divergence converges to the quantum relative entropy \cite[Lemma~3.4]{Mosonyi2017}:
\begin{equation}
\lim_{\alpha \to 1}D^\flat_{\alpha}(\rho\| \sigma) = D(\rho\Vert\sigma).
\end{equation}
In analogy to the Chernoff divergence representation in \eqref{eq:Chernoff-div-expr-renyi} we also define the log-Euclidean Chernoff distance as
\begin{equation}\label{eq:Chernoff-div-expr-flat-renyi}
C^\flat(\rho\|\sigma)  \coloneqq\sup_{0\leq\alpha\leq1} (1-\alpha) D^\flat_\alpha(\rho \| \sigma).
\end{equation}
The log-Euclidean R\'enyi divergence comes up in our work due to the following ``divergence sphere optimization,'' holding for not too large $r > 0$ and states $\rho$ and $\sigma$ \cite[Remark 1]{Nagaoka06}:
\begin{equation}\label{eq:div-sphere-rep}
\inf_{\tau : D(\tau \Vert \sigma) \leq r}D(\tau \Vert \rho) = \sup_{\alpha\in (0,1)} \left\{ \frac{\alpha-1}{\alpha} \left[r - D^\flat_\alpha(\rho \Vert \sigma) \right]\right\}.
\end{equation}
The following inequalities relate the three aforementioned quantum R\'enyi divergences evaluated on quantum states $\rho$ and $\sigma$~\cite{WWY14,IRS17,Mosonyi2017}:
\begin{align}
\alpha D_{\alpha}(\rho\| \sigma) 
 & \leq \widetilde{D}_{\alpha}(\rho\| \sigma)
 \leq D_{\alpha}(\rho\| \sigma) 
\leq D^\flat_{\alpha}(\rho\| \sigma)
&   \text{ for } \alpha & \in(0,1) , \\ 
D^\flat_{\alpha}(\rho\| \sigma)
& \leq \widetilde{D}_{\alpha}(\rho\| \sigma)
\leq D_{\alpha}(\rho\| \sigma) 
 &  \text{ for } \alpha & \in(1,\infty).
\end{align}
All of the above quantum R\'enyi divergences reduce to the corresponding classical versions by embedding probability distributions into diagonal, commuting quantum states. 


\section{Settings for asymptotic channel discrimination}\label{sec:settings}

In this section, we describe the information-theoretic settings for asymptotic quantum channel discrimination that we study. We emphasise that this is in contrast to most of the previous work that has focused on the finite, non-asymptotic regime.


\subsection{Protocol for quantum channel discrimination}\label{sec:discrimination}

The problem of quantum channel discrimination is made mathematically precise by the following hypothesis testing problems for quantum channels. Given two quantum channels $\cN_{A\to B}$ and $\cM_{A\to B}$ acting on an input system $A$ and an output system $B$, a general adaptive strategy for  discriminating them is as follows.

We allow the preparation of an arbitrary
input state $\rho_{R_{1}A_{1}}=\tau_{R_{1}A_{1}}$, where $R_{1}$ is an
ancillary register. The $i$th use of a channel accepts the register $A_{i}$ as
input and produces the register $B_{i}$ as output. After each invocation of
the channel $\mathcal{N}_{A\to B}$ or $\mathcal{M}_{A\to B}$, an (adaptive) channel $\mathcal{A}_{R_{i}B_{i}\rightarrow
R_{i+1}A_{i+1}}^{(i)}$ is applied to the registers $R_{i}$ and $B_{i}$,
yielding a quantum state $\rho_{R_{i+1}A_{i+1}}$ or $\tau_{R_{i+1}A_{i+1}}$ in
registers $R_{i+1}A_{i+1}$, depending on whether the channel is equal to
${\mathcal{N}_{A\to B}}$ or ${\mathcal{M}_{A\to B}}$. That is,
\begin{align}
\rho_{R_{i+1}A_{i+1}} &  \coloneqq \mathcal{A}_{R_{i}B_{i}\rightarrow
R_{i+1}A_{i+1}}^{\left(  i\right)  }(\rho_{R_{i}B_{i}}), &  &  \rho
_{R_{i}B_{i}}\coloneqq \mathcal{N}_{A_{i}\rightarrow B_{i}}(\rho_{R_{i}A_{i}
})\label{eq:rho-adaptive},\\
\tau_{R_{i+1}A_{i+1}} &  \coloneqq \mathcal{A}_{R_{i}B_{i}\rightarrow
R_{i+1}A_{i+1}}^{\left(  i\right)  }(\tau_{R_{i}B_{i}}),&  & 
\tau
_{R_{i}B_{i}}\coloneqq \mathcal{M}_{A_{i}\rightarrow B_{i}}(\tau_{R_{i}A_{i}
})\label{eq:tau-adaptive},
\end{align}
for every $1\leq i<n$ on the left-hand side, and for every $1\leq i\leq n$ on
the right-hand side. Finally, a quantum measurement $\{Q_{R_{n}B_{n}}
,I_{R_{n}B_{n}}-Q_{R_{n}B_{n}}\}$ is performed on the systems $R_{n}B_{n}$ to
decide which channel was applied. The outcome $Q$ corresponds to a final decision that the channel is $\cN$, while the outcome $I-Q$ corresponds to a final decision that the channel is $\cM$. We define the final decision probabilities as
\begin{align}
p &  \coloneqq \operatorname{Tr}\big[Q_{R_{n}B_{n}}\rho_{R_{n}B_{n}}\big], \label{eq:prot-p}\\
q &  \coloneqq \operatorname{Tr}\big[Q_{R_{n}B_{n}}\tau_{R_{n}B_{n}}\big]. \label{eq:prot-q}
\end{align}
Figure~\ref{fig:adaptive-prot} depicts such a general protocol for channel discrimination when the channel $\mathcal{N}$ or $\mathcal{M}$ is called three times.

In what follows, we use the simplifying notation $\{ Q, \cA\}$ to identify a particular strategy using channels $\{\cA^{(i)}_{R_i B_i \to R_{i+1} A_{i+1}}\}_i$ and a final measurement $\{ Q_{R_n B_n}, I_{R_n B_n} - Q_{R_n B_n} \}$. For simplicity, this shorthand also includes the preparation of the initial state $\rho_{R_1 A_1} = \tau_{R_1 A_1} $, which can be understood as arising from the action of an initial channel $\cA^{(0)}_{R_0 B_0 \to R_{1} A_{1}}$ for which the input systems $R_0$ and $B_0$ are trivial. This naturally gives rise to the two possible error probabilities:
\begin{align}
\label{eq:ChannelErrorsI+II}
\alpha_n(\{ Q, \cA\})&\coloneqq\tr\big[
(I_{R_n B_n}-Q_{R_n B_n})\rho_{R_n B_n}\big]\;
&\text{type~I error probability,}
\\
\beta_n(\{ Q, \cA\})&\coloneqq \tr\big[Q_{R_n B_n} \tau_{R_n B_n} \big]\;
&\text{type~II error probability.}&
\end{align}
In what follows, we discuss the behaviour of the type~I and type~II error probabilities in various asymmetric and symmetric settings.

In the above specification of quantum channel discrimination, the physical setup corresponding to it is that the discriminator has ``black box'' access to $n$ uses of the channel $\cN$ or~$\cM$, meaning that the channel is some device in the laboratory of the discriminator, he has physical access to both the input and output systems of the channel, and he is allowed to apply arbitrary procedures to distinguish them. As such, the above method of discriminating the channels is the most natural and general in this setting. Other physical constraints motivate different models of channel discrimination protocols, and in fact, there could be a large number of physically plausible channel discrimination strategies to consider, depending on the physical constraints of the discriminator(s). For example, if the channels being compared have input and output systems that are in different physical locations, as would be the case for a long-haul fiber optic cable, then it might not be feasible  to carry out such a general channel discrimination protocol as described above (two parties in distant laboratories would be needed), and it would be meaningful to consider a different channel discrimination protocol. However, the channel discrimination described above is the most general, and if there is a limitation established for the distinguishability of two channels in this model, then the same limitation applies to any other channel discrimination model that could be considered.

\begin{remark}\label{rmk:block-coding}
Another kind of channel discrimination strategy often considered in the literature is a parallel discrimination strategy, in which a state $\gamma_{RA^n}$ is prepared, either the tensor-power channel $(\cN_{A \to B})^{\otimes n}$ or $(\cM_{A \to B})^{\otimes n}$ is applied, and then a joint measurement is performed on the systems~$RB^n$. As noted in \cite{CDP08a}, a parallel channel discrimination strategy of the channels $\cN$ and $\cM$ is a special case of an adaptive channel discrimination strategy as detailed above. Indeed, the first state $\rho_{R_1 A_1}$ in an adaptive protocol could be $\gamma_{RA^n}$ with the system $R_1$ of $\rho_{R_1A_1}$ identified with the systems $RA_2 \cdots A_n$ of $\gamma_{RA^n}$, and then the role of the first adaptive channel would be simply to swap in system $A_2$ of $\gamma_{RA^n}$ for the second channel call, the second adaptive channel would swap in system $A_3$ of $\gamma_{RA^n}$ for the third channel call, etc. As such, parallel channel discrimination is not the most general approach to consider, and as stated previously, any limitation placed on the distinguishability of the channels from an adaptive discrimination strategy serves as a limitation when using a parallel discrimination strategy.
\end{remark}

To the best of our knowledge, adaptive quantum channel discrimination protocols were first studied by Chiribella {\it et al.}~\cite{CDP08a}, whereas the particular information-theoretic quantities that we introduce in the following Sections \ref{sec:asym-stein}--\ref{sec:syn-chernoff} go back to Hayashi~\cite{Hayashi09} for the classical case and to Cooney {\it et al.}~\cite{Cooney2016} for the quantum case.


\subsection{Asymmetric setting -- Stein}\label{sec:asym-stein}

For asymmetric hypothesis testing, we minimize the type~II error probability, under the constraint that the type~I error probability does not exceed a constant $\eps\in(0,1)$. We are then interested in characterizing the non-asymptotic quantity
\begin{align}
\zeta_n(\eps,\cN,\cM)\coloneqq\sup_{\{Q,\cA\}}\left\{-\frac{1}{n}\log\beta_n(\{ Q, \cA\})\middle|\alpha_n(\{Q,\cA\})\leq\eps\right\},
\end{align}
as well as the asymptotic quantities
\begin{equation}
\underline{\zeta}(\eps,\cN,\cM)\coloneqq\liminf_{n \to \infty} \zeta_n(\eps,\cN,\cM), \qquad \overline{\zeta}(\eps,\cN,\cM)\coloneqq\limsup_{n \to \infty} \zeta_n(\eps,\cN,\cM).
\end{equation}


\subsection{Strong converse exponent -- Han-Kobayashi}\label{sec:han-kob}

The strong converse exponent is a refinement of the asymmetric hypothesis testing quantity discussed above. For $r>0$, we are interested in characterizing the non-asymptotic quantity
\begin{align}
H_n(r,\cN,\cM)\coloneqq\inf_{\{Q,\cA\}}\left\{-\frac{1}{n}\log(1-\alpha_n(\{ Q, \cA\}))\middle|\beta_n(\{Q,\cA\})\leq2^{-rn}\right\},
\end{align}
as well as the asymptotic quantities
\begin{equation}
\underline{H}(r,\cN,\cM)\coloneqq\liminf_{n \to \infty} H_n(r,\cN,\cM), \qquad \overline{H}(r,\cN,\cM)\coloneqq\limsup_{n \to \infty} H_n(r,\cN,\cM).
\end{equation}
The interpretation is that the type~II error probability is constrained to tend to zero exponentially fast at a rate $r > 0$, but then if $r$ is too large, the type~I error probability will necessarily tend to one exponentially fast, and we are interested in the exact rate of exponential convergence. Note that this strong converse exponent is only non-trivial if $r$ is sufficiently large.


\subsection{Error exponent -- Hoeffding}

The error exponent is another refinement of asymmetric hypothesis testing, in the sense that the type~II error probability is constrained to decrease exponentially with exponent $r>0$. We are then interested in characterizing the  error exponent of the type~I error probability under this constraint. That is, we are interested in characterizing the non-asymptotic quantity
\begin{align}
B_n(r,\cN,\cM)\coloneqq\sup_{\{Q,\cA\}}\left\{-\frac{1}{n}\log\alpha_n(\{Q,\cA\})\middle|\beta_n(\{Q,\cA\})\leq2^{-rn}\right\}.
\end{align}
as well as the asymptotic quantities
\begin{equation}
\underline{B}(r,\cN,\cM)\coloneqq\liminf_{n \to \infty} B_n(r,\cN,\cM), \qquad \overline{B}(r,\cN,\cM)\coloneqq\limsup_{n \to \infty} B_n(r,\cN,\cM).
\end{equation}
Note that this error exponent is non-trivial only if $r$ is not too large.


\subsection{Symmetric setting -- Chernoff}\label{sec:syn-chernoff}

Here we are interested in minimizing the total error probability of guessing incorrectly, that is, symmetric hypothesis testing, which is sometimes also described as the Bayesian setting of hypothesis testing. Given an {\it a priori} probability $p\in(0,1)$ that the first channel $\mathcal{N}$ is selected, the non-asymptotic symmetric error exponent is defined as\footnote{The quantity underlying the non-asymptotic symmetric error exponent was previously studied in \cite{CDP08a,G12} and shown to be related to the norm defined therein (see \cite{GW07,CDP08b,G09} for related work).}
\begin{align}
\xi_n(p,\cN,\cM)\coloneqq\sup_{\{Q,\cA\}}-\frac{1}{n}\log\Big(p\cdot \alpha_n(\{Q,\cA\})+(1-p)\beta_n(\{Q,\cA\})\Big).
\end{align}
Given that the expression above involves an optimization over all final measurements abbreviated by $Q$, we can employ the well known result relating optimal error probability to trace distance \cite{H69,H73,Hel76} to conclude that
\begin{equation}
\xi_n(p,\cN,\cM)=\sup_{\{\cA\}}-\frac{1}{n}\log\left(\frac{1}{2}\left(1- \left\Vert p \rho_{R_n B_n }- (1-p)\tau_{R_n B_n}\right\Vert_1\right)\right),
\end{equation}
where $\rho_{R_n B_n }$ and $\tau_{R_n B_n }$ are defined in \eqref{eq:rho-adaptive} and \eqref{eq:tau-adaptive}, respectively.
We are then interested in the asymptotic symmetric error exponent
\begin{align}
\underline{\xi}(\cN,\cM)\coloneqq\liminf_{n \to \infty} \xi_n(p,\cN,\cM) & = \liminf_{n \to \infty}  \xi_n(1/2,\cN,\cM), \\
\overline{\xi}(\cN,\cM)\coloneqq\limsup_{n \to \infty} \xi_n(p,\cN,\cM) & = \limsup_{n \to \infty} \xi_n(1/2,\cN,\cM),
\end{align}
with the equalities following, e.g., from \cite[Theorem 12]{Yu17}. That is, choosing \textit{a priori} probabilities  different from $1/2$ does not affect the asymptotic symmetric error exponent. We have the following relation between the asymptotic Hoeffding error exponent and the asymptotic symmetric error exponent:
\begin{align}\label{eq:Hoeff-to-Chern}
\overline{\xi}(\cN,\cM)=\sup\Big\{r\Big|\overline{B}(r,\cN,\cM)\geq r\Big\}.
\end{align}


\subsection{Energy-constrained channel discrimination}\label{sec:energy-cons-ch-disc}

The protocols for quantum channel discrimination could be energy constrained as well. This is an especially important consideration when discriminating bosonic Gaussian channels~\cite{S17}, for which the theory could become trivial without such an energy constraint imposed. For example, if the task is to discriminate  two pure-loss bosonic Gaussian channels of different transmissivities and there is no energy constraint, then these channels can be perfectly discriminated with a single call: one would send in a coherent state of arbitrarily large photon number, and then states output from the two different channels are orthogonal in the limit of infinite photon number (see, e.g., \cite[Section 2]{Winter17}).

To develop the formalism of energy-constrained channel discrimination, let $H_A$ be a Hamiltonian acting on the channel input Hilbert space, and  we take $H_A$ to be a positive semi-definite operator throughout for simplicity. Then, for the channel discrimination protocol described in Section~\ref{sec:discrimination} to be energy constrained, we demand that the average energy of the reduced states at all of the channel inputs satisfy
\begin{equation}
\frac{1}{n}\sum_{i=1}^{n}\operatorname{Tr}[H_A\rho_{A_{i}}]\leq E,
\label{eq:energy-constraint}
\end{equation}
where $E\in\lbrack0,\infty)$. It then follows that an unconstrained protocol corresponds to choosing $H_A=I_A$ and $E=1$, so that the corresponding ``energy constraint'' in \eqref{eq:energy-constraint} is automatically satisfied for all quantum states.

The resulting quantities of interest then depend on the Hamiltonian $H_A$ and energy constraint $E$, and we write $\{Q, \cA, H, E\}$ to denote the strategy employed. We write the type I and II error probabilities as
$\alpha_n(\{Q,\cA, H, E\})$ and $\beta_n(\{Q,\cA, H, E\}) $, respectively. The resulting optimized quantities of interest from the previous sections are then defined in the same way, but additionally depend on the Hamiltonian $H_A$ and energy constraint $E$. We denote them by 
\begin{align}
& \zeta_{n}(\eps,\cN, \cM, H,E) & \text{ (Stein)}, \\
&H_n(r,\cN, \cM, H, E) & \text{ (Han-Kobayashi)}, \\
& B_n(r,\cN, \cM, H, E) & \text{ (Hoeffding)}, \\
& \xi_n(p,\cN, \cM, H, E) & \text{ (Chernoff)}.
\end{align}


\section{Amortized distinguishability of quantum channels}\label{sec:technical-tools}

In order to analyse the hypothesis testing problems for quantum channels as discussed in Section~\ref{sec:discrimination}, we now introduce the concept of the amortized distinguishability of quantum channels. This allows us to reduce questions about the operational problems of hypothesis testing to mathematical questions about quantum channels, states, and distinguishability measures of them. In the following, we also detail many properties of the amortized distinguishability of quantum channels, which are of independent interest.


\subsection{Generalized divergences}

We say that a function $\mathbf{D}:\cS(A)\times\cS(A)\to\mathbb{R}\cup\{+\infty\}$ is a generalized divergence \cite{PV10,SW12} if for arbitrary Hilbert spaces $\mathcal{H}_A$ and $\mathcal{H}_B$, arbitrary states $\rho_A,\sigma_A\in\cS(A)$, and an arbitrary channel $\cN_{A\to B}\in\cQ(A\to B)$, the following data-processing inequality holds
\begin{align}\label{eq:data-processing}
\mathbf{D}(\rho_{A}\|\sigma_A)\geq\mathbf{D}(\cN_{A\to B}(\rho_A)\|\cN_{A\to B}(\sigma_A)).
\end{align}
From this inequality, we find in particular that for all states $\rho_A,\sigma_A\in\cS(A)$, $\omega_R\in\cS(R)$, the following identity holds \cite{WWY14}
\begin{align}
\mathbf{D}(\rho_A\otimes\omega_R\|\sigma_A\otimes\omega_R)=\mathbf{D}(\rho_A\|\sigma_A),
\end{align}
and that for an arbitrary isometric channel $\cU_{A\to B}\in\cQ(A\to B)$, we have that \cite{WWY14}
\begin{align}
\mathbf{D}(\cU_{A\to B}(\rho_A)\|\cU_{A\to B}(\sigma_A))=\mathbf{D}(\rho_A\|\sigma_A).
\end{align}

We call a generalized divergence \textit{faithful} if the inequality $\mathbf{D}(\rho_A\|\rho_A)\leq0$ holds for an arbitrary state $\rho_A\in\cS(A)$, and \textit{strongly faithful} if for arbitrary states $\rho_A,\sigma_A\in\cS(A)$ we have $\mathbf{D}(\rho_A\|\sigma_A)=0$ if and only if $\rho_A=\sigma_A$. Moreover, a generalized divergence is sub-additive with respect to tensor-product states if for all $\rho_A,\sigma_A\in\cS(A)$ and all $\omega_B,\tau_B\in\cS(B)$ we have
\begin{align}
\mathbf{D}(\rho_A\otimes\omega_B\|\sigma_A\otimes\tau_B)\leq\mathbf{D}(\rho_A\|\sigma_A)+\mathbf{D}(\omega_B\|\tau_B).
\end{align}
Examples of interest are in particular the quantum relative entropy, the Petz-R\'enyi divergences, the sandwiched R\'enyi divergences, or the Chernoff distance\,---\,as defined in Section~\ref{sec:notation}.

As discussed in \cite{PhysRevA.97.012332,SWAT17}, a generalized  divergence possesses the direct-sum property  on classical-quantum states if the following equality holds:
\begin{equation}
\mathbf{D}\!\left(\sum_{x}p_{X}(x)|x\rangle\langle x|_{X}\otimes\rho^{x}\middle\Vert\sum_{x}p_{X}(x)|x\rangle\langle x|_{X}\otimes\sigma^{x}\right)  =\sum_{x}p_{X}(x)\mathbf{D}(\rho^{x}\Vert\sigma^{x}),\label{eq:d-sum-prop}
\end{equation}
where $p_{X}$ is a probability distribution, $\{|x\rangle\}_{x}$ is an orthonormal basis, and $\{\rho^{x}\}_{x}$ and $\{\sigma^{x}\}_{x}$ are sets of states. We note that this property holds for trace distance, quantum relative entropy, and the Petz-R\'enyi and sandwiched R\'enyi quasi-entropies $\operatorname{sgn}(\alpha-1)\tr\left[\rho^\alpha\sigma^{1-\alpha}\right]$ and $\operatorname{sgn}(\alpha-1)\tr[(\sigma^{\frac{1-\alpha}{2\alpha}}\rho\sigma^{\frac{1-\alpha}{2\alpha}})^\alpha]$, respectively. A generalized divergence is jointly convex if
\begin{equation}
\mathbf{D}\!\left(  \sum_{x}p_{X}(x)\rho^{x}\middle\Vert\sum_{x}p_{X}(x)\sigma^{x}\right)  \leq \sum_{x}p_{X}(x)\mathbf{D}(\rho^{x}\Vert\sigma^{x}).
\end{equation}
Any generalized divergence is jointly convex if it satisfies the direct-sum property, a fact that follows by applying the defining property in \eqref{eq:data-processing} and data processing under partial trace.

Based on generalized divergences, one can define a generalized channel divergence as a measure for the distinguishability of two quantum channels \cite{PhysRevA.97.012332}. The idea behind the following measure of channel distinguishability is to allow for an arbitrary input state to be used to distinguish the channels:
\begin{definition}[Generalized channel divergence \cite{PhysRevA.97.012332}]
Let $\bf D$ be a generalized divergence and $\cN_{A\to B},\cM_{A\to B}\in\cQ(A\to B)$. The generalized channel divergence of $\cN_{A\to B}$ and $\cM_{A\to B}$ is defined as
\begin{align}\label{eq:channel-divergence}
\mathbf{D}(\cN\|\cM)\coloneqq\sup_{\rho\in\cS(RA)} \mathbf{D}( \cN_{A\to B}(\rho_{RA})\|\cM_{A\to B}(\rho_{RA})),
\end{align}
where the supremum is with respect to bipartite states $\rho_{RA}$, as well as the dimension of the reference system $R$.
\end{definition}

Even though the generalized channel divergence is defined to have an optimization over all bipartite states with unbounded reference system $R$, it immediately follows from the axioms on generalized divergences together with purification and the Schmidt decomposition that without loss of generality we can restrict the supremum to pure states $\Psi_{RA}\in\cS(RA)$ and choose system~$R$ isomorphic to system $A$. Hence, if the channel input system is finite-dimensional, then the optimization problem in \eqref{eq:channel-divergence} becomes bounded. Particular instances of generalized channel divergences include the diamond norm of the difference of $\mathcal{N}_{A\rightarrow B}$ and $\mathcal{M}_{A\rightarrow B}$ \cite{K97}, as well as the R\'enyi channel divergence from~\cite{Cooney2016}. 	


\subsection{Amortized channel divergence}

We now define the amortized channel divergence as a measure of the distinguishability of two quantum channels. The idea behind this measure, in contrast to the generalized channel divergence recalled above, is to consider two \textit{different} states $\rho_{RA}$ and $\sigma_{RA}$ that can be input to the channels $\cN_{A\to B}$ and $\cM_{A\to B}$, in order to explore the largest distinguishability that can be realized between the channels. However, from a resource-theoretic perspective, these initial states themselves could have some distinguishability, and so it is sensible to subtract off the initial distinguishability of the states $\rho_{RA}$ and $\sigma_{RA}$ from the final distinguishability of the channel output states $\cN_{A\to B}(\rho_{RA})$ and $\cM_{A\to B}(\sigma_{RA})$. This procedure leads to the amortized channel divergence:
\begin{definition}[Amortized channel divergence]
Let $\bf D$ be a generalized divergence, and let $\cN_{A\to B},\cM_{A\to B}\in\cQ(A\to B)$. We define the amortized channel divergence as \begin{align}\label{eq:gen_channel-divergence}
\mathbf{D}^{\mathcal{A}}(\cN\|\cM)\coloneqq \sup_{\rho_{RA},\sigma_{RA}\in\cS(RA)} \left[\mathbf{D}( \cN_{A\to B}(\rho_{RA}) \| \cM_{A\to B}(\sigma_{RA}))-\mathbf{D}(\rho_{RA}\|\sigma_{RA})\right].
\end{align}
\end{definition}

Note that in general the supremum cannot be restricted to pure states only, and moreover, there is \textit{a priori} no dimension bound on the system $R$. Hence, the optimization problem in~\eqref{eq:gen_channel-divergence} can in general be unbounded.

We note here that the idea behind amortized channel divergence is inspired by related ideas from entanglement theory, in which one quantifies the entanglement of a quantum channel by the largest difference  in entanglement between the output and input states  of the channel \cite{BHLS03,LHL03,KW17a}. Several properties of a channel's amortized entanglement were shown in \cite{KW17a}, and in the following sections, we establish several important properties of the amortized channel divergence, the most notable one being a data-processing inequality, i.e., that it is monotone under the action of a superchannel. Some of these properties are related to those recently considered in \cite{Yuan2018} for the quantum relative entropy of channels defined in~\cite{Cooney2016,PhysRevA.97.012332}. Moreover, very recently a special case of the amortized channel divergence was proposed in \cite{CE18}, and we discuss this more in Remark~\ref{rem:fid-div}.


\subsection{Properties of amortized channel divergence}

The generalized channel divergence is never larger than its amortized version:

\begin{prop}[Distinguishability does not decrease under amortization]\label{Lem:AmortizedIneq}
Let $\bf D$ be a faithful generalized divergence and $\cN_{A\to B},\cM_{A\to B}\in\cQ(A\to B)$. Then we have that
\begin{align}
\mathbf{D}^{\mathcal{A}}(\cN\|\cM)& \geq \mathbf{D}(\cN\|\cM) \label{eq:amort-larger}\\
& \geq 0.
\end{align}
\end{prop}

The proof is immediate, following because we can choose $\sigma_{RA}$ in the optimization of $\mathbf{D}^{\mathcal{A}}$ to be equal to $\rho_{RA}$ and then apply the faithfulness assumption. As we will see, what is fundamental to the problem of channel discrimination is to find instances of divergences and quantum channels for which we have the opposite inequality holding also 
\begin{align}\label{eq:amort-collapse}
\mathbf{D}^{\mathcal{A}}(\cN\|\cM)\overset{?}{\leq}\mathbf{D}(\cN\|\cM).
\end{align}
If this inequality holds, we say that there is an ``amortization collapse,'' due to the fact that, when combined with the inequality in \eqref{eq:amort-larger}, we would have the equality $\mathbf{D}^{\mathcal{A}}(\cN\|\cM)=\mathbf{D}(\cN\|\cM)$, understood as the amortized channel divergence collapsing to the generalized channel divergence.

Additionally, amortized channel divergences are faithful whenever the underlying generalized divergence is faithful, as the following lemma states. 

\begin{prop}[Faithfulness]
If a generalized divergence is strongly faithful on states, then its associated amortized channel divergence is strongly faithful for channels, meaning that $\mathbf{D}^{\mathcal{A}}(\cN\|\cM) = 0$ if and only if $\cN = \cM$.
\end{prop}

\begin{proof}
Suppose that the channels are identical: $\mathcal{N}_{A\to B}=\mathcal{M}_{A\to B}$. Then we have that
\begin{equation}
\mathbf{D}^{\mathcal{A}}(\mathcal{N}_{A\to B}\Vert\mathcal{M}_{A\to B})=\mathbf{D}^{\mathcal{A}}(\mathcal{N}_{A\to B}
\Vert\mathcal{N}_{A\to B})=0.
\end{equation}
This follows because
\begin{align}
\mathbf{D}^{\mathcal{A}}(\mathcal{N}_{A\to B}\Vert\mathcal{M}_{A\to B})  & =\sup_{\rho,\sigma\in\cS(RA)}\left[\mathbf{D}(\mathcal{N}_{A\to B}(\rho_{RA})\Vert\mathcal{M}_{A\to B}(\sigma_{RA}))-\mathbf{D}(\rho_{RA}\Vert\sigma_{RA})\right]\\
& =\sup_{\rho,\sigma\in\cS(RA)}\left[\mathbf{D}(\mathcal{N}_{A\to B}(\rho_{RA})\Vert\mathcal{N}_{A\to B}(\sigma_{RA}))-\mathbf{D}(\rho_{RA}
\Vert\sigma_{RA})\right]\\
& \leq0,
\end{align}
which follows from the data-processing inequality. Equality is achieved by picking $\rho_{RA}=\sigma_{RA}$ and invoking strong faithfulness of the underlying measure. Now suppose that $\mathbf{D}^{\mathcal{A}}(\mathcal{N}_{A\to B}\Vert\mathcal{M}_{A\to B})=0$. Then, we have by definition that
\begin{equation}
\sup_{\rho,\sigma\in\cS(RA)}\left[\mathbf{D}(\mathcal{N}_{A\to B}(\rho_{RA})\Vert\mathcal{M}_{A\to B}(\sigma_{RA}
))-\mathbf{D}(\rho_{RA}\Vert\sigma_{RA})\right]=0.
\end{equation}
Since we have that
\begin{equation}
\mathbf{D}^{\mathcal{A}}(\mathcal{N}_{A\to B}\Vert\mathcal{M}_{A\to B})\geq \mathbf{D}(\mathcal{N}_{A\to B}\Vert
\mathcal{M}_{A\to B})\geq0,
\end{equation}
this means that
\begin{equation}
\mathbf{D}(\mathcal{N}_{A\to B}\Vert\mathcal{M}_{A\to B})=\sup_{\rho\in\cS(RA)}\mathbf{D}(\mathcal{N}_{A\to B}(\rho_{RA}
)\Vert\mathcal{M}_{A\to B}(\rho_{RA}))=0.
\end{equation}
We could then pick $\rho_{RA}$ equal to the maximally entangled state, and from faithfulness of the underlying measure, deduce that the Choi states are equal. But if this is the case, then the channels are equal.
\end{proof}

The data-processing inequality is the statement that a distinguishability measure for quantum states should not increase under the action of the same channel on these states. It is one of the most fundamental principles of information theory, and this is the reason why the notion of generalized divergence is a useful abstraction. As an extension of this concept, here we prove a data-processing inequality for the amortized channel divergence, which establishes that it does not increase under the action of the same superchannel on the underlying channels. The generalized channel divergence of \cite{PhysRevA.97.012332} satisfies this property (as established in \cite{G18}), and we show here that the amortized channel divergence satisfies this property as well.

\begin{prop}[Data processing]
Let $\mathcal{N}_{A\rightarrow B},\mathcal{M}_{A\rightarrow B}\in
\mathcal{Q}(A\rightarrow B)$. Let $\Theta$ be a superchannel as described in \eqref{eqn:superchannel}. Then the following inequality holds
\begin{equation}
\mathbf{D}^{\mathcal{A}}(\mathcal{N}_{A \to B}\Vert\mathcal{M}_{A \to B})\geq\mathbf{D}^{\mathcal{A}}(\Theta\left(  \mathcal{N}_{A\rightarrow B}\right)\Vert\Theta\left(  \mathcal{M}_{A\rightarrow B}\right)  ).
\end{equation}
\end{prop}

\begin{proof}
Set $\mathcal{F}_{C\rightarrow D}\coloneqq\Theta\left(  \mathcal{N}_{A\rightarrow
B}\right)  $ and $\mathcal{G}_{C\rightarrow D}\coloneqq\Theta\left(  \mathcal{M}_{A\rightarrow B}\right)  $ as the respective channels that are output from the superchannel $\Theta$. Let $\rho_{RC}$ and $\sigma_{RC}$ be arbitrary
input states for $\mathcal{F}_{C\rightarrow D}$ and $\mathcal{G}_{C\rightarrow
D}$, respectively. Set $\eta_{RAE}\coloneqq\Lambda_{C\rightarrow AE}(\rho_{RC})$ and
$\zeta_{RAE}\coloneqq\Lambda_{C\rightarrow AE}(\sigma_{RC})$, where $\Lambda
_{C\rightarrow AE}$ is the pre-processing quantum channel from \eqref{eqn:superchannel}. Then
\begin{align}
&  \mathbf{D}(\mathcal{F}_{C\rightarrow D}(\rho_{RC})\Vert\mathcal{G}_{C\rightarrow D}(\sigma_{RC}))-\mathbf{D}(\rho_{RC}\Vert\sigma_{RC})\notag \\
&  \leq\mathbf{D}(\mathcal{F}_{C\rightarrow D}(\rho_{RC})\Vert\mathcal{G}_{C\rightarrow D}(\sigma_{RC}))-\mathbf{D}(\Lambda_{C\rightarrow AE}(\rho
_{RC})\Vert\Lambda_{C\rightarrow AE}(\sigma_{RC}))\\
&  =\mathbf{D}((\Omega_{BE\rightarrow D}\circ\mathcal{N}_{A\rightarrow B}\circ\Lambda_{C\rightarrow AE})(\rho_{RC})\Vert(\Omega_{BE\rightarrow D}\circ\mathcal{M}_{A\rightarrow B}\circ\Lambda_{C\rightarrow AE})(\sigma
_{RC}))\nonumber\\
&  \qquad-\mathbf{D}(\eta_{RAE}\Vert\zeta_{RAE})\\
&  =\mathbf{D}((\Omega_{BE\rightarrow D}\circ\mathcal{N}_{A\rightarrow
B})(\eta_{RAE})\Vert(\Omega_{BE\rightarrow D}\circ\mathcal{M}_{A\rightarrow
B})(\zeta_{RAE}))-\mathbf{D}(\eta_{RAE}\Vert\zeta_{RAE})\\
&  \leq\mathbf{D}(\mathcal{N}_{A\rightarrow B}(\eta_{RAE})\Vert\mathcal{M}_{A\rightarrow B}(\zeta_{RAE}))-\mathbf{D}(\eta_{RAE}\Vert\zeta_{RAE})\\
&  \leq\mathbf{D}^{\mathcal{A}}(\mathcal{N}\Vert\mathcal{M}).
\end{align}
The first inequality follows from data processing with the pre-processing channel $\Lambda_{C\rightarrow AE}$. The next two equalities follow from definitions. The second-to-last inequality follows from data processing with the post-processing channel $\Omega_{BE\rightarrow D}$. The final inequality follows because the states $\eta_{RAE}$ and $\zeta_{RAE}$ are particular states, but the amortized channel divergence involves an optimization over all such input states. Since the chain of inequalities holds for arbitrary input states $\rho_{RC}$ and $\sigma_{RC}$, we conclude the inequality in the statement of the proposition by taking a supremum over all such states.
\end{proof}

Joint convexity is a natural property that a measure of channel distinguishability should obey. The statement is that channel distinguishability should not increase under a mixing of the channels under consideration.

\begin{prop}
[Joint convexity]Let $\mathcal{N}_{A\rightarrow B}^{x},\mathcal{M}_{A\rightarrow B}^{x}\in\mathcal{Q}(A\rightarrow B)$ for all $x\in\mathcal{X}$, and let $p_{X}(x)$ be a probability distribution. Then if the underlying
generalized divergence obeys the direct-sum property in \eqref{eq:d-sum-prop}, the amortized
channel divergence is jointly convex, in the sense that
\begin{equation}
\sum_{x}p_{X}(x)\mathbf{D}^{\mathcal{A}}(\mathcal{N}^{x}\Vert\mathcal{M}^{x})\geq\mathbf{D}^{\mathcal{A}}(\overline{\mathcal{N}}\Vert\overline{\mathcal{M}}),
\end{equation}
where $\overline{\mathcal{N}}\coloneqq\sum_{x}p_{X}(x)\mathcal{N}^{x}$ and
$\overline{\mathcal{M}}\coloneqq\sum_{x}p_{X}(x)\mathcal{M}^{x}$.
\end{prop}

\begin{proof}
Let $\rho_{RA}$ and $\sigma_{RA}$ be arbitrary states. Then, we have that
\begin{align}
& \mathbf{D}\left(\overline{\mathcal{N}}_{A\rightarrow B}(\rho_{RA})\middle\|\overline{\mathcal{M}}_{A\rightarrow B}(\sigma_{RA}))-\mathbf{D}(\rho_{RA}\Vert\sigma_{RA}\right)\\
& \leq\mathbf{D}\!\left(  \sum_{x}p_{X}(x)|x\rangle\langle x|_{X}\otimes\mathcal{N}_{A\rightarrow B}^{x}(\rho_{RA})\middle\Vert\sum_{x}p_{X}(x)|x\rangle\langle x|_{X}\otimes\mathcal{M}_{A\rightarrow B}^{x}(\sigma_{RA})\right)  -\mathbf{D}(\rho_{RA}\Vert\sigma_{RA})\\
& =\sum_{x}p_{X}(x)\mathbf{D}\!\left(  \mathcal{N}_{A\rightarrow B}^{x}(\rho_{RA})\Vert\mathcal{M}_{A\rightarrow B}^{x}(\sigma_{RA})\right)-\mathbf{D}(\rho_{RA}\Vert\sigma_{RA})\\
& =\sum_{x}p_{X}(x)\left[  \mathbf{D}\!\left(  \mathcal{N}_{A\rightarrow B}^{x}(\rho_{RA})\Vert\mathcal{M}_{A\rightarrow B}^{x}(\sigma_{RA})\right)-\mathbf{D}(\rho_{RA}\Vert\sigma_{RA})\right]\\
& \leq\sum_{x}p_{X}(x)\mathbf{D}^{\mathcal{A}}(\mathcal{N}^{x}\Vert\mathcal{M}^{x}).
\end{align}
The first inequality follows from data processing. The first equality follows
from the direct-sum property. The final inequality follows from optimizing.
Since the inequality holds for an arbitrary choice of states $\rho_{RA}$ and
$\sigma_{RA}$, we conclude the statement of the proposition.
\end{proof}

The following stability property is a direct consequence of the definition of amortized channel divergence.

\begin{prop}[Stability]
Let $\cN_{A\to B},\cM_{A\to B}\in\cQ(A\to B)$. Then, we have
\begin{equation}
\mathbf{D}^{\mathcal{A}}(\cI_R\otimes \mathcal{N}\Vert \cI_R\otimes \mathcal{M}) = \mathbf{D}^{\mathcal{A}}(\mathcal{N}\Vert\mathcal{M}),
\end{equation}
where $\cI_R$ denotes the identity channel on a quantum system of arbitrary size.
\end{prop}

Two channels $\cN_{A\to B}$ and $\cM_{A\to B}$ are jointly covariant with respect to a group $\mathcal{G}$ if for all $g \in \mathcal{G}$, there exist unitary channels $\mathcal{U}_{A}^{g}$ and $\mathcal{V}_{B}^{g}$ such that \cite{TW2016,DW17}
\begin{equation}
\cN_{A\to B} \circ \mathcal{U}_{A}^{g} = \mathcal{V}_{B}^{g} \circ \cN_{A\to B}, \qquad \cM_{A\to B} \circ \mathcal{U}_{A}^{g} = \mathcal{V}_{B}^{g} \circ \cM_{A\to B}.
\end{equation}
For channels that are jointly covariant with respect to a group, we find that it suffices to optimize over $\rho_{RA}$ and $\sigma_{RA}$ whose
reduced states on $A$ satisfy the symmetry. As such, the following lemma represents a counterpart to \cite[Proposition~II.4]{PhysRevA.97.012332}, which established a related statement for the generalized channel divergence.

\begin{lemma}[Symmetries]
Let $\cN_{A\to B},\cM_{A\to B}\in\cQ(A\to B)$ be jointly covariant with respect to a group $G$, as defined above. It then suffices to optimize $\mathbf{D}^{\mathcal{A}}(\mathcal{N}\Vert\mathcal{M})$ over states $\overline{\rho}_{RA}$ and $\overline{\sigma}_{RA}$ such that
\begin{align}
\overline{\rho}_{A}   =\frac{1}{\left\vert G\right\vert }\sum_{g}
\mathcal{U}_{A}^{g}(\overline{\rho}_{A}),\qquad
\overline{\sigma}_{A}   =\frac{1}{\left\vert G\right\vert }\sum
_{g}\mathcal{U}_{A}^{g}(\overline{\sigma}_{A}).
\end{align}
\end{lemma}

\begin{proof}
Each step in what follows is a consequence of data processing. Consider that
\begin{multline}
\mathbf{D}(\mathcal{N}_{A\rightarrow B}(\rho_{RA})\Vert\mathcal{M}
_{A\rightarrow B}(\sigma_{RA}))-\mathbf{D}(\rho_{RA}\Vert\sigma_{RA})
\leq
\mathbf{D}(\mathcal{N}_{A\rightarrow B}(\rho_{RA})\Vert\mathcal{M}
_{A\rightarrow B}(\sigma_{RA}))\\
-\mathbf{D}\!\left(  \frac{1}{\left\vert G\right\vert }\sum_{g}|g\rangle\langle
g|_{G}\otimes\mathcal{U}_{A}^{g}(\rho_{RA})\middle\Vert\frac{1}{\left\vert
G\right\vert }\sum_{g}|g\rangle\langle g|_{G}\otimes\mathcal{U}_{A}^{g}
(\sigma_{RA})\right).
\end{multline}
Let us focus on the first term:
\begin{align}
&  \mathbf{D}(\mathcal{N}_{A\rightarrow B}(\rho_{RA})\Vert\mathcal{M}
_{A\rightarrow B}(\sigma_{RA}))\nonumber\\
&  =\mathbf{D}\!\left(  \frac{1}{\left\vert G\right\vert }\sum_{g}
|g\rangle\langle g|_{G}\otimes\mathcal{N}_{A\rightarrow B}(\rho_{RA}
)\middle\Vert\frac{1}{\left\vert G\right\vert }\sum_{g}|g\rangle\langle g|_{G}
\otimes\mathcal{M}_{A\rightarrow B}(\sigma_{RA})\right)  \\
&  =\mathbf{D}\!\left(  \frac{1}{\left\vert G\right\vert }\sum_{g}
|g\rangle\langle g|_{G}\otimes(\mathcal{V}_{B}^{g}\circ\mathcal{N}
_{A\rightarrow B})(\rho_{RA})\middle\Vert\frac{1}{\left\vert G\right\vert }\sum
_{g}|g\rangle\langle g|_{G}\otimes(\mathcal{V}_{B}^{g}\circ\mathcal{N}
_{A\rightarrow B})(\sigma_{RA})\right)  \\
&  =\mathbf{D}\!\left(  \frac{1}{\left\vert G\right\vert }\sum_{g}
|g\rangle\langle g|_{G}\otimes\mathcal{N}_{A\rightarrow B}(\mathcal{U}_{A}
^{g}(\rho_{RA}))\middle\Vert\frac{1}{\left\vert G\right\vert }\sum_{g}|g\rangle
\langle g|_{G}\otimes\mathcal{N}_{A\rightarrow B}(\mathcal{U}_{A}^{g}
(\sigma_{RA}))\right).
\end{align}
Then, we find that
\begin{multline}
\mathbf{D}(\mathcal{N}_{A\rightarrow B}(\rho_{RA})\Vert\mathcal{M}
_{A\rightarrow B}(\sigma_{RA}))-\mathbf{D}(\rho_{RA}\Vert\sigma_{RA})\\
\leq\mathbf{D}(\mathcal{N}_{A\rightarrow B}(\overline{\rho}_{RA}
)\Vert\mathcal{M}_{A\rightarrow B}(\overline{\sigma}_{RA}))-\mathbf{D}
(\overline{\rho}_{RA}\Vert\overline{\sigma}_{RA}),
\end{multline}
where $\overline{\rho}_{RA}$ and $\overline{\sigma}_{RA}$ are states such that
\begin{align}
\overline{\rho}_{A}   =\frac{1}{\left\vert G\right\vert }\sum_{g}
\mathcal{U}_{A}^{g}(\overline{\rho}_{A}), \qquad 
\overline{\sigma}_{A}   =\frac{1}{\left\vert G\right\vert }\sum
_{g}\mathcal{U}_{A}^{g}(\overline{\sigma}_{A}).
\end{align}
This concludes the proof.
\end{proof}


\subsection{Amortization collapse for max-relative entropy}

One instance of a generalized divergence for which the inequality in \eqref{eq:amort-collapse} holds for any two quantum channels is the max-channel divergence, defined from the max-relative entropy. That is, Proposition~\ref{lem:D_max} states that the max-channel divergence does not increase under  amortization. This  result thus complements some developments in the theory of quantum communication, in which related amortization collapses occurred for the max-relative entropy~\cite{CH17,BW17}.

\begin{prop}\label{lem:D_max}
Let $\cN_{A\to B},\cM_{A\to B}\in\cQ(A\to B)$. Then, for $\alpha\geq1$ we have that
\begin{align}
\widetilde{D}^{\mathcal{A}}_\alpha(\cN\|\cM)\leq D_{\max}(\cN\|\cM),
\end{align}
and this implies in particular that
\begin{align}\label{Eq:maxAmortized}
D^{\mathcal{A}}_{\max}(\cN\|\cM)=D_{\max}(\cN\|\cM).
\end{align}
\end{prop}

Fundamental to the proof of the amortization collapse in Proposition~\ref{lem:D_max} is the following lemma.

\begin{lemma}[Data-processed triangle inequality \cite{CH17}]\label{DataTriangle}
Let $\cN_{A\to B}\in\cQ(A\to B)$, $\rho_A,\sigma_A\in\cS(A)$, and $\omega_B\in\cS(B)$. Then, for $\alpha\geq 1$ we have that
\begin{align}
\widetilde{D}_\alpha(\cN_{A\to B}(\rho_A)\|\omega_B)\leq\widetilde{D}_\alpha(\rho_A\|\sigma_A)+D_{\max}(\cN_{A\to B}(\sigma_A)\|\omega_B).
\end{align}
\end{lemma}

\begin{proof}[Proof of Proposition~\ref{lem:D_max}]
To see the inequality, simply note that 
\begin{align}
&\!\!\!\widetilde D_\alpha(\cN(\rho_{RA})\|\cM(\sigma_{RA}))-\widetilde D_\alpha(\rho_{RA}\|\sigma_{RA})\notag \\
&\leq \widetilde D_\alpha(\rho_{RA}\|\sigma_{RA})+D_{\max}(\cN(\sigma_{RA})\|\cM(\sigma_{RA}))-\widetilde D_\alpha(\rho_{RA}\|\sigma_{RA})\\
&= D_{\max}(\cN(\sigma_{RA})\|\cM(\sigma_{RA}))\\
&\leq D_{\max}(\cN\|\cM),
\end{align}
where the first inequality follows from Lemma~\ref{DataTriangle} and the second by taking the optimization over states. The equality in \eqref{Eq:maxAmortized} then simply follows from 
\begin{align}
D_{\max}^{\mathcal{A}}(\cN\|\cM)\geq D_{\max}(\cN\|\cM),
\end{align}
which in turn follows from Proposition~\ref{Lem:AmortizedIneq}.
\end{proof}

The max-relative entropy could, due to the supremum over input states, potentially be hard to compute in general. In the following lemma, we show that the given quantity can always be expressed in a simple form, as the max-relative entropy of the Choi states of the channels, which is also a semi-definite program (SDP) and therefore efficiently computable.

\begin{lemma}\label{lemma:dmax}
Let $\cN_{A\to B},\cM_{A\to B}\in\cQ(A\to B)$. Then, we have that
\begin{align}\label{eq:maxEquality}
D_{\max}(\cN\|\cM) = D_{\max}(\mathcal{N}_{A\rightarrow B}(\Phi_{RA})\Vert\mathcal{M}_{A\rightarrow B}(\Phi_{RA})),
\end{align}
which is an SDP. 
\end{lemma}

\begin{proof}
Recall that the max-channel divergence is given by
\begin{align}
\sup_{\psi_{RA}}D_{\max}(\mathcal{N}_{A\rightarrow B}(\psi_{RA})\Vert\mathcal{M}_{A\rightarrow B}(\psi_{RA}))=\sup_{\psi_{RA}}\inf\Big\{\lambda:\mathcal{N}_{A\rightarrow B}(\psi_{RA})\leq2^{\lambda}\cdot\mathcal{M}_{A\rightarrow B}(\psi_{RA})\Big\}.
\end{align}
Consider that
\begin{equation}\label{eq:triv-ineq-dmax-chan}
D_{\max}(\mathcal{N}\Vert\mathcal{M})\geq D_{\max}(\mathcal{N}_{A\rightarrow B}(\Phi_{RA})\Vert\mathcal{M}_{A\rightarrow B}(\Phi_{RA})),
\end{equation}
by definition, given that the left-hand side involves an optimization, but the
right-hand side does not. Next, let $\lambda$ be such that
\begin{equation}
\mathcal{N}_{A\rightarrow B}(\Phi_{RA})\leq2^{\lambda}\mathcal{M}
_{A\rightarrow B}(\Phi_{RA}),
\end{equation}
where $\Phi_{RA}$ denotes a maximally entangled state.
By scaling by a dimension factor, the above is equivalent to
\begin{equation}\label{eq:op-ineq-D_max-chan}
\mathcal{N}_{A\rightarrow B}(\Gamma_{RA})\leq2^{\lambda}\mathcal{M}_{A\rightarrow B}(\Gamma_{RA}),
\end{equation}
where $\Gamma_{RA} \coloneqq |A| \Phi_{RA}$ denotes the  version of the maximally entangled state~$\Phi_{RA}$ that is not normalized. Then, due to the fact that any pure state $\psi_{RA}=X_{R}\Gamma_{RA}
X_{R}^{\dag}$ for some operator $X_{R}$ such that $\operatorname{Tr}
[X_{R}^{\dag}X_{R}]=1$, we then conclude that the following operator
inequality is satisfied
\begin{align}
X_{R}\mathcal{N}_{A\rightarrow B}(\Gamma_{RA})X_{R}^{\dag}  & \leq2^{\lambda
}X_{R}\mathcal{M}_{A\rightarrow B}(\Gamma_{RA})X_{R}^{\dag}\\
\Leftrightarrow\mathcal{N}_{A\rightarrow B}(X_{R}\Gamma_{RA}X_{R}^{\dag})  &
\leq2^{\lambda}\mathcal{M}_{A\rightarrow B}(X_{R}\Gamma_{RA}X_{R}^{\dag})\\
\Leftrightarrow\mathcal{N}_{A\rightarrow B}(\psi_{RA})  & \leq2^{\lambda
}\mathcal{M}_{A\rightarrow B}(\psi_{RA}).
\end{align}
Thus, we could potentially find a smaller value of $\lambda$ for which
$\mathcal{N}_{A\rightarrow B}(\psi_{RA})\leq2^{\lambda}\mathcal{M}
_{A\rightarrow B}(\psi_{RA})$ is satisfied, implying that
\begin{equation}
\inf\Big\{\mu:\mathcal{N}_{A\rightarrow B}(\psi_{RA})\leq2^{\mu}\mathcal{M}_{A\rightarrow B}(\psi_{RA})\Big\}  \leq\lambda.
\end{equation}
But since the argument holds for all choices of $\lambda$ satisfying
\eqref{eq:op-ineq-D_max-chan}, we conclude that
\begin{equation}
\inf\Big\{\mu:\mathcal{N}_{A\rightarrow B}(\psi_{RA})\leq2^{\mu}\mathcal{M}_{A\rightarrow B}(\psi_{RA})\Big\}  \leq\inf\Big\{\lambda:\mathcal{N}_{A\rightarrow B}(\Phi_{RA})\leq2^{\lambda}\mathcal{M}_{A\rightarrow B}(\Phi_{RA})\Big\},
\end{equation}
which is equivalent to
\begin{equation}
D_{\max}(\mathcal{N}_{A\rightarrow B}(\psi_{RA})\Vert\mathcal{M}_{A\rightarrow B}(\psi_{RA}))\leq D_{\max}(\mathcal{N}_{A\rightarrow B}(\Phi_{RA})\Vert\mathcal{M}_{A\rightarrow B}(\Phi_{RA})).
\end{equation}
Now, we have proven that the inequality above holds for an arbitrary choice of
the state $\psi_{RA}$, and so we conclude that
\begin{multline}
D_{\max}(\mathcal{N}_{A\rightarrow B}(\Phi_{RA})\Vert\mathcal{M}_{A\rightarrow
B}(\Phi_{RA}))\label{eq:nontriv-ineq-dmax-chan}\\
\geq\sup_{\psi}D_{\max}(\mathcal{N}_{A\rightarrow B}(\psi_{RA}
)\Vert\mathcal{M}_{A\rightarrow B}(\psi_{RA}))=D_{\max}(\mathcal{N}
\Vert\mathcal{M}).
\end{multline}
Combining \eqref{eq:triv-ineq-dmax-chan} and \eqref{eq:nontriv-ineq-dmax-chan}
gives the statement of the lemma.
To see that this is an SDP, we write
\begin{equation}
D_{\max}(\cN\|\cM)= \log\inf\{\lambda : 
\mathcal{N}_{A\rightarrow B}(\Gamma_{RA})
\leq\lambda\cdot 
\mathcal{M}_{A\rightarrow B}(\Gamma_{RA})
\}. \label{eq:SDP-dmax}
\end{equation}
This concludes the proof.
\end{proof}

\begin{remark}
Note that the max-channel divergence is a special case of the sandwiched R\'enyi channel divergences proposed in \cite{Cooney2016}, as well as the generalized channel divergences from \cite{PhysRevA.97.012332}. Recently, the  following channel divergence was proposed
in the context of the resource theory of coherence \cite[Definition 19]{GFWRSCW18}:
\begin{equation}
D'_{\max}(\cN \| \cM) \coloneqq \inf\Big\{\lambda : 2^{\lambda} \cM - \cN \textrm{ is CP }\Big\}.
\end{equation}
Since it suffices to check whether a map is CP by evaluating it on the maximally entangled state, we find, as a consequence of \eqref{eq:SDP-dmax} and Lemma~\ref{lemma:dmax}, that $D'_{\max}(\cN \| \cM) = D_{\max}(\cN \| \cM)$.
\end{remark}

As a consequence of Proposition~\ref{lem:D_max} and Lemma~\ref{lemma:dmax}, we find that
\begin{align}
D_{\max}^{\mathcal{A}}(\cN\|\cM)=D_{\max}(\mathcal{N}_{A\rightarrow B}(\Phi_{RA})\Vert\mathcal{M}_{A\rightarrow B}(\Phi_{RA})),
\end{align}
and in Section~\ref{sec:general-bounds}, we discuss how this directly translates into strong converse bounds for quantum channel discrimination.
Analogous to the above amortization collapse for max-relative entropy, we prove in Appendix~\ref{app:hilbert-alpha-div} that an amortization collapse occurs for an amortized channel divergence based on the Hilbert $\alpha$-divergence from \cite{BG17}. Since the trace distance is a special case of a Hilbert $\alpha$-divergence, as shown in \cite[Theorem 1]{BG17}, it follows that the diamond norm of the difference of two quantum channels does not increase under amortization. We also discuss how certain channel metrics based on quantum fidelity do not increase under amortization. For other divergences including the quantum relative entropy, the Petz-R\'enyi divergences, the sandwiched divergences, or the Chernoff distance, we are in general not able to prove that there is a collapse of the corresponding amortized channel divergence. However, when we evaluate these amortized divergences for various special channels, we are able to prove that an amortization collapse occurs. We discuss this in Sections~\ref{sec:cq}-\ref{sec:examples}, along with the direct implications for the operational settings of quantum channel discrimination.


\subsection{Meta-converse for quantum channel discrimination\\via amortized channel divergence}

The following Lemma~\ref{ThmAmortizedBound} functions as a meta-converse for quantum channel discrimination (similarly to \cite{polyanskiy10}). By this, we mean that we can recover particular converse statements for quantum channel discrimination by plugging in different choices of a generalized divergence into Lemma~\ref{ThmAmortizedBound}. 
Consider a general channel discrimination protocol as introduced in Section~\ref{sec:discrimination}, with final decision probabilities $p$ and $q$, as given in \eqref{eq:prot-p} and \eqref{eq:prot-q}, respectively. Conceptually, the statement of the lemma is that the distinguishability of the final decision probabilities $p$ and $q$ at the end of a channel discrimination protocol, in which the channels are called $n$ times, is limited by $n$ times the amortized channel divergence of the two channels.

\begin{lemma}[Meta-converse]\label{ThmAmortizedBound}
Let $\cN_{A\to B},\cM_{A\to B}\in\cQ(A\to B)$. Then, we have for any protocol for quantum channel discrimination as introduced in Section~\ref{sec:discrimination} and any faithful generalized divergence that
\begin{align}
\mathbf{D}(p\Vert q)\leq n\cdot\mathbf{D}^{\cA}(\mathcal{N}\Vert\mathcal{M}),
\end{align}
where
\begin{align}
\mathbf{D}(p\Vert q)\coloneqq\mathbf{D}(\zeta(p)\Vert\zeta(q))\quad\text{with}\quad\zeta(p)\coloneqq p|0\rangle\langle0|+(1-p)|1\rangle\langle1|.
\end{align}
\end{lemma}

\begin{proof}
Let $\{\cQ, \cA\}$ denote an arbitrary protocol for discrimination of the channels $\cN$ and $\cM$, as discussed in Section~\ref{sec:discrimination}, and let $p$ and $q$ denote the final decision probabilities.
Consider that
\begin{align}
\mathbf{D}(p\Vert q)&\leq\mathbf{D}(\rho_{R_{n}B_{n}}\Vert\tau_{R_{n}B_{n}})\\
&\leq\mathbf{D}(\rho_{R_{n}B_{n}}\Vert\tau_{R_{n}B_{n}})-\mathbf{D}(\rho_{R_{1}A_{1}}\Vert\tau_{R_{1}A_{1}})\\
&=\mathbf{D}(\rho_{R_{n}B_{n}}\Vert\tau_{R_{n}B_{n}})-\mathbf{D}(\rho_{R_{1}A_{1}}\Vert\tau_{R_{1}A_{1}})+\sum_{i=2}^{n}\Big(\mathbf{D}(\rho_{R_{i}A_{i}}\Vert\tau_{R_{i}A_{i}})-\mathbf{D}(\rho_{R_{i}A_{i}}\Vert\tau_{R_{i}A_{i}})\Big)\\
&=\mathbf{D}(\rho_{R_{n}B_{n}}\Vert\tau_{R_{n}B_{n}})-\mathbf{D}(\rho_{R_{1}A_{1}}\Vert\tau_{R_{1}A_{1}})\nonumber\\
&\quad+\sum_{i=2}^{n}\left(\mathbf{D}(\mathcal{A}_{R_{i-1}B_{i-1}\rightarrow R_{i}A_{i}}^{\left(  i-1\right)  }(\rho_{R_{i-1}B_{i-1}})\Vert\mathcal{A}_{R_{i-1}B_{i-1}\rightarrow R_{i}A_{i}}^{\left(i-1\right)}(\tau_{R_{i-1}B_{i-1}}))-\mathbf{D}(\rho_{R_{i}A_{i}}\Vert\tau_{R_{i}A_{i}})\right).\label{eq:bridge-step}
\end{align}
The first inequality follows from data processing under the final measurement $\{Q_{R_{n}B_{n}},I_{R_{n}B_{n}}-Q_{R_{n}B_{n}}\}$. The second inequality follows from the assumption of faithfulness and the fact that the initial states are equal: $\rho_{R_1 A_1} = \tau_{R_1 A_1}$. The next two equalities are straightforward. Continuing, we have that
\begin{align}
\text{Eq.~\eqref{eq:bridge-step}}&\leq\mathbf{D}(\rho_{R_{n}B_{n}}\Vert
\tau_{R_{n}B_{n}})-\mathbf{D}(\rho_{R_{1}A_{1}}\Vert\tau_{R_{1}A_{1}})\notag\\
& \qquad +\sum_{i=1}^{n-1}\mathbf{D}(\rho_{R_{i}B_{i}}\Vert\tau_{R_{i}B_{i}})-\sum_{i=2}^{n}\mathbf{D}(\rho_{R_{i}A_{i}}\Vert\tau_{R_{i}A_{i}})\\
& =\sum_{i=1}^{n}\Big(\mathbf{D}(\rho_{R_{i}B_{i}}\Vert\tau_{R_{i}B_{i}})-\mathbf{D}(\rho_{R_{i}A_{i}}\Vert\tau_{R_{i}A_{i}})\Big)\\
&=\sum_{i=1}^{n}\Big(\mathbf{D}(\mathcal{N}_{A\rightarrow B}(\rho_{R_{i}A_{i}})\Vert\mathcal{M}_{A\rightarrow B}(\tau_{R_{i}A_{i}}))-\mathbf{D}(\rho_{R_{i}A_{i}}\Vert\tau_{R_{i}A_{i}})\Big)
\label{eq:here-for-d-sum}
\\
&\leq n\cdot\sup_{\rho,\sigma\in\cS(RA)}\left[\mathbf{D}(\mathcal{N}_{A\rightarrow B}(\rho_{RA})\Vert\mathcal{M}_{A\rightarrow B}(\sigma_{RA}))-\mathbf{D}(\rho_{RA}\Vert\sigma_{RA})\right]\\
&=n\cdot\mathbf{D}^{\mathcal{A}}(\mathcal{N}\Vert\mathcal{M}).
\end{align}
The first inequality follows from data processing with respect to the channel $\mathcal{A}_{R_{i-1}B_{i-1}\rightarrow R_{i}A_{i}}^{\left(i-1\right)}$. The next two equalities are rewritings. The final inequality follows by optimization and the last by definition.
\end{proof}


\subsection{Energy-constrained channel divergences and meta-converse for energy-constrained channel discrimination}

Given that we consider energy-constrained protocols for channel discrimination as described in Section~\ref{sec:energy-cons-ch-disc}, it is natural to consider energy-constrained channel divergences. This was done in \cite{SWAT17} for the generalized channel divergence, where an energy-constrained generalized channel divergence was defined for energy $E\in\lbrack0,\infty)$ and a Hamiltonian $H_A$ as follows:
\begin{equation}
\mathbf{D}_{H,E}(\mathcal{N}\Vert\mathcal{M})=\sup_{\psi_{RA}: \operatorname{Tr}\{H_A\psi_A\}\leq E}\mathbf{D}(\mathcal{N}_{A\rightarrow B}(\psi_{RA})\Vert\mathcal{M}_{A\rightarrow B}(\psi_{RA})).
\end{equation}
Note again that it suffices to optimize over pure states with system $R$ isomorphic to system~$A$. Special cases of the energy-constrained generalized channel divergence from \cite{SWAT17} include the energy-constrained diamond norm of the difference of two channels \cite{Sh17,Winter17}, as well as the energy-constrained Bures distance \cite{S16}.

Here, we define the amortized, energy-constrained channel divergence as
\begin{equation}
\mathbf{D}^{\cA}_{H,E}(\mathcal{N}\Vert\mathcal{M})=\sup_{\substack{\rho_{RA},\sigma_{RA}\in\cS(RA),\\
\operatorname{Tr}\{H_A\rho_A\}, \operatorname{Tr}\{H_A\sigma_A\}\leq E}} \left[\mathbf{D}( \cN_{A\to B}(\rho_{RA}) \| \cM_{A\to B}(\sigma_{RA}))-\mathbf{D}(\rho_{RA}\|\sigma_{RA})\right].
\end{equation}	
These quantities possess many of the properties of the unconstrained divergences, as detailed in the previous sections, but we do not list them here for the sake of brevity. We close this section by providing a generalization of the meta-converse in Lemma~\ref{ThmAmortizedBound} to the case of energy-constrained channel discrimination protocols.

\begin{lemma}
With the same notation as in Lemma~\ref{ThmAmortizedBound}, 
if the channel discrimination protocol has an average energy constraint as in \eqref{eq:energy-constraint}, with Hamiltonian $H$ and energy $E\in[0,\infty)$, and the faithful generalized divergence obeys the direct-sum property in \eqref{eq:d-sum-prop}, then the following bound holds:
\begin{align}\label{eq:en-constr-bnd-gen-div}
\mathbf{D}(p\Vert q)\leq n\cdot\mathbf{D}^{\cA}_{H,E}(\mathcal{N}\Vert\mathcal{M}).
\end{align}
\end{lemma}

\begin{proof}
The analysis proceeds in the same way as in the proof of Lemma~\ref{ThmAmortizedBound}, but at \eqref{eq:here-for-d-sum}, we can exploit the assumed direct-sum property of the generalized divergence and find that
\begin{align}
& \sum_{i=1}^{n}\Big(\mathbf{D}(\mathcal{N}_{A\rightarrow B}(\rho_{R_{i}A_{i}})\Vert\mathcal{M}_{A\rightarrow B}(\tau_{R_{i}A_{i}}))-\mathbf{D}(\rho_{R_{i}A_{i}}\Vert\tau_{R_{i}A_{i}})\Big)
\notag \\
& = n \cdot \Big(\mathbf{D}(\mathcal{N}_{A\rightarrow B}(\rho_{URA})\Vert\mathcal{M}_{A\rightarrow B}(\tau_{URA}))-\mathbf{D}(\rho_{URA}\Vert\tau_{URA})\Big)\\
& \leq n \cdot \mathbf{D}^{\cA}_{H,E}(\mathcal{N}\Vert\mathcal{M}),
\end{align}
where the states $\rho_{URA}$ and $\tau_{URA}$ are defined as
\begin{align}
\rho_{URA} \coloneqq
\frac{1}{n}
\sum_{i=1}^n \vert i \rangle \langle i \vert_U \otimes \rho_{R_{i}A_{i}}, \qquad
\tau_{URA} \coloneqq
\frac{1}{n}
\sum_{i=1}^n \vert i \rangle \langle i \vert_U \otimes \tau_{R_{i}A_{i}}, 
\end{align}
with the $R$ system as large as it needs to be to accommodate the largest of the $R_i$ systems. The last inequality follows because the reduced states $\tr_{UR}[\rho_{URA}]$ and $\tr_{UR}[\tau_{URA}]$ each satisfy the average energy constraint in \eqref{eq:energy-constraint} by assumption, and then we can optimize over all such states.
\end{proof}


\section{Converse bounds for quantum channel discrimination}\label{sec:general-bounds}

The amortization results can now be translated into general bounds on quantum channel discrimination. We start here by reviewing the work~\cite{Duan09}, which discusses conditions for when $n$ copies of two channels become perfectly distinguishable. Interestingly, these conditions are single-letter and efficiently checkable. They can be stated in terms of Kraus decompositions $\{N_i\}_i$ and $\{M_j\}_j$ of $\cN$ and $\cM$, respectively. (Note that the criterion is independent of the Kraus decompositions chosen.) Namely, $n$ copies of $\cN_{A\to B}$ and $\cM_{A\to B}$ are perfectly distinguishable for some finite $n$ if and only if
\begin{itemize}
\item $\exists\;\psi_{RA}\in\cV(RA)$ with $|R|=|A|$ such that $\mathrm{supp}(\cN_{A\to B}(\psi_{RA}))\cap\mathrm{supp}(\cM_{A\to B}(\psi_{RA}))=\emptyset$
\item $I\notin\mathrm{span}\{N_i^\dagger M_j\}$.
\end{itemize}
We note that, under these conditions, all the asymptotic quantities introduced in Section~\ref{sec:discrimination} become trivial. Hence, it remains to find bounds on quantum channel discrimination for the case when at least one of the above conditions does not hold. In the remainder of this section, we give general bounds and postpone specific classes of channels for which we get tight single-letter characterizations to Sections~\ref{sec:cq}-\ref{sec:examples}.


\subsection{Stein bound}

For non-adaptive protocols, when we restrict the input states to be product states\,---\,but still allow for a quantum memory system $R$\,---\,it directly follows from Stein's lemma for quantum state discrimination~\cite{HP91,ON00} that the optimal asymptotic error exponent for $\eps\in(0,1)$ is given by the quantum relative entropy divergence $D(\cN\|\cM)$, as observed in~\cite{Cooney2016}. This obviously also gives an achievability bound for the adaptive setting. In the following, we are interested in converse bounds for the adaptive setting, and in later sections, we discuss when our general converse bounds match the achievability result. It turns out that the amortized quantum relative entropy divergence $D^{\mathcal{A}}(\cN\|\cM)$ provides such a converse bound.

\begin{prop}\label{weakConverse}
Let $\cN_{A\to B},\cM_{A\to B}\in\cQ(A\to B)$. Then, we have for $n\in\mathbb{N}$ and $\eps\in[0,1)$ that
\begin{align}\label{eq:weak-conv-stein-unconstrained}
\zeta_n(\eps,\cN,\cM)\leq \frac{1}{1-\eps}
\left(D^{\mathcal{A}}(\cN\|\cM)+\frac{h_2(\varepsilon)}{n}\right),
\end{align}
where $h_2(\varepsilon) \coloneqq -\varepsilon \log(\varepsilon) - (1-\varepsilon) \log(1-\varepsilon)$ denotes the binary entropy.
If the channel discrimination protocol is energy constrained, as discussed in Section~\ref{sec:energy-cons-ch-disc}, with Hamiltonian~$H_A$ and energy constraint $E\in[0,\infty)$, then the following bound holds:
\begin{align}\label{eq:weak-conv-stein-constrained}
\zeta_n(\eps,\cN,\cM,H,E)\leq \frac{1}{1-\eps}
\left(D^{\mathcal{A}}_{H,E}(\cN\|\cM)+\frac{h_2(\varepsilon)}{n}\right).
\end{align}
\end{prop} 

\begin{proof}
Let $\{\cQ, \cA\}$ denote an arbitrary protocol for discrimination of the channels $\cN$ and~$\cM$, as discussed in Section~\ref{sec:discrimination}, and let $p$ and $q$ denote the final decision probabilities.
As observed previously (e.g., in \cite{KW17}), if the constraint $\alpha_n(\{ Q, \cA\}) \leq\varepsilon$ is not satisfied with equality, then one can modify the measurement operator as $ Q \to \lambda Q$ for some $\lambda \in (0,1)$ such that $\alpha_n(\{ \lambda Q, \cA\}) =\varepsilon$, whereas the type~II error probability only decreases under this modification. Since we are interested in 
the optimized quantity $\zeta_n(\eps,\cN,\cM)$, we always perform this modification to any channel discrimination protocol if necessary. 
Now following the original proof of Stein's lemma for quantum states~\cite{HP91}, we find that
\begin{align}
D(p\| q)  
&= (1-p) \log\frac{1-p}{1-q} + p \log(p/q)\\
& = \alpha_n(\{ Q, \cA\}) \log\frac{\alpha_n(\{ Q, \cA\})}{1 - \beta_n(\{ Q, \cA\}) } + (1-\alpha_n(\{ Q, \cA\}))\log\frac{1-\alpha_n(\{ Q, \cA\})}{\beta_n(\{ Q, \cA\})} \\
& = \varepsilon\log\frac{\varepsilon}{1 - \beta_n(\{ Q, \cA\}) } + (1-\varepsilon)\log\frac{1-\varepsilon}{\beta_n(\{ Q, \cA\})}\\
& = - h_2(\varepsilon) -
\varepsilon\log(1 - \beta_n(\{ Q, \cA\})) -
(1-\varepsilon)\log\beta_n(\{Q,\cA\}) \\
&\geq - h_2(\varepsilon) - (1-\varepsilon)\log\beta_n(\{Q,\cA\}). 
\end{align}
By rearranging the above equation, it follows that
\begin{align}
-\log\beta_n(\{ Q, \cA\}) \leq \frac{1}{1-\varepsilon}\Big(D(p\| q) +h_2(\varepsilon)\Big).
\end{align}
Now we can apply Lemma~\ref{ThmAmortizedBound}, choosing the generalized channel divergence to be the quantum relative entropy. We then find that
\begin{align}
-\log\beta_n(\{ Q, \cA\}) \leq \frac{1}{1-\varepsilon}\Big(n \cdot D^{\mathcal{A}}(\cN\|\cM)) +h_2(\varepsilon)\Big).
\end{align}
This is a universal upper bound, holding for an arbitrary channel discrimination protocol, and so now dividing by $n$ and taking the supremum over $\{Q,\cA\}$ on the left-hand side, we conclude the inequality in \eqref{eq:weak-conv-stein-unconstrained}.

For the energy-constrained case, the proof goes in the same way, but we apply \eqref{eq:en-constr-bnd-gen-div}, choosing again the generalized divergence to be the quantum relative entropy and noting that it obeys the direct-sum property in \eqref{eq:d-sum-prop}. Hence, \eqref{eq:weak-conv-stein-constrained} follows.
\end{proof}

We next obtain a strong converse bound in terms of the max-channel divergence.

\begin{prop}\label{LemmaStrongConverseDmax}
Let $\cN_{A\to B},\cM_{A\to B}\in\cQ(A\to B)$. For $\eps\in[0,1)$, the following bound holds
\begin{align}
\zeta_n(\eps,\cN,\cM)
\leq D_{\max}(\cN\|\cM)
 + \frac{1}{n}\log\!\left(\frac{1}{1-\varepsilon}\right).
\end{align}
\end{prop}

\begin{proof}
Let $\{\cQ, \cA\}$ denote an arbitrary protocol for discrimination of the channels $\cN$ and $\cM$, as discussed in Section~\ref{sec:discrimination}, and let $p$ and $q$ denote the final decision probabilities.
As discussed in the previous proof, we can take $\alpha_n(\{Q,\mathcal{A}\})=\varepsilon$, leading to
\begin{align}
D_{\max}(p\Vert q) & = \log \max\{(1-\varepsilon) / q, \varepsilon / (1-q)\}\geq \log(1-\varepsilon) - \log q.
\end{align}
By applying the meta-converse in Lemma~\ref{ThmAmortizedBound}, as well as the amortization collapse for max-relative entropy from Proposition~\ref{lem:D_max}, we conclude that
\begin{equation}
-\frac{1}{n}\log q \leq D_{\max}(\cN\|\cM)+\frac{1}{n}\log\frac{1}{1-\varepsilon}.
\end{equation}
Since the bound is a uniform bound applying to any channel discrimination protocol, we conclude the statement of the proposition.
\end{proof}

Combining the bounds in Proposition~\ref{weakConverse} and Proposition~\ref{lem:D_max}, we arrive at the following asymptotic statements for the Stein setting.

\begin{corollary}\label{cor:all-stein}
Let $\cN_{A\to B},\cM_{A\to B}\in\cQ(A\to B)$ and $\eps\in(0,1)$. Then, we have that
\begin{align}\label{Eq:AllStein}
D(\cN\|\cM)\leq\underline{\zeta}(\eps,\cN,\cM)\leq\overline{\zeta}(\eps,\cN,\cM)\leq D_{\max}(\cN\|\cM). 
\end{align}
\end{corollary}

Note that the $D^{\mathcal{A}}(\cN\|\cM)$ bound from Proposition~\ref{weakConverse} is \textit{a priori} an unbounded optimization problem. However, as stated in Lemma~\ref{lemma:dmax}, the $D_{\max}(\cN\|\cM)$ bound in Corollary~\ref{cor:all-stein} can be written as a semi-definite program and is thus efficiently computable. Note also that it is in general unclear if the amortized quantity $D^{\mathcal{A}}(\cN\|\cM)$ can be achieved by an adaptive protocol. However, in Sections~\ref{sec:cq}-\ref{sec:examples} we present various examples for which we find that $D^{\mathcal{A}}(\cN\|\cM)=D(\cN\|\cM)$. This will then also allow us to discuss refinements in terms of the error exponent and the strong converse exponent.

\begin{remark}
The result stated in Corollary~\ref{cor:all-stein} allows us to conclude a ``faithfulness'' statement for the Stein setting, similar to that made in \cite{Yu17} for the Chernoff setting. Namely, the asymptotic Type~II error exponent $\overline{\zeta}(\eps,\cN,\cM)$ is finite if and only if the support condition
\begin{equation}\label{eq:sup-cond-faith-stein}
\operatorname{supp}(\cN_{A \to B}(\Phi_{RA})) \subseteq \operatorname{supp}(\cM_{A \to B}(\Phi_{RA}))
\end{equation}
holds. To see this, suppose that the support condition in \eqref{eq:sup-cond-faith-stein} holds. Then $D_{\max}(\cN\|\cM)$ is finite and so is $\overline{\zeta}(\eps,\cN,\cM)$, by the upper bound in \eqref{Eq:AllStein}. Now suppose that the support condition in \eqref{eq:sup-cond-faith-stein} does not hold. Then $D(\cN\|\cM)$ is infinite and so is $\overline{\zeta}(\eps,\cN,\cM)$, by the lower bound in \eqref{Eq:AllStein}.
\end{remark}

\bigskip
\textbf{Example.}
In general, the upper bound in \eqref{Eq:AllStein} in terms of $D_{\max}(\mathcal{N}\|\mathcal{M})$ can be rather different from the lower bound. In the following, we study this difference between the bounds for a physically interesting class of channels. 

The generalized amplitude damping channel with parameters $(\eta,p)$ is a qubit channel, modeling dissipation to the environment at a finite temperature \cite{NC00}. A set of Kraus operators for it is as follows: 
\begin{align}
A_1 & =
\sqrt{p}\begin{bmatrix}
    1       & 0 \\
    0       & \sqrt{\eta}
\end{bmatrix}
, & 
A_2 & =
\sqrt{p}\begin{bmatrix}
    0       & \sqrt{1-\eta} \\
    0       & 0
\end{bmatrix},
\\
A_3 & =
\sqrt{1-p}\begin{bmatrix}
    \sqrt{\eta}      & 0 \\
    0       & 1
\end{bmatrix}
, & 
A_4 & =
\sqrt{1-p}\begin{bmatrix}
    0       & 0 \\
    \sqrt{1-\eta}      & 0
\end{bmatrix},
\end{align}
where $p\in[0,1]$ represents the dissipation to the environment and $\eta \in [0,1]$ is related to how much the input qubit mixes with  the environment qubit. This channel can intuitively be understood as the ``qubit version'' of the bosonic thermal channel \cite{S17}. Indeed, the generalized amplitude damping channel arises by preparing an environment qubit in the state $p\vert 0 \rangle \langle 0 \vert + (1-p) \vert 1 \rangle \langle 1 \vert$, interacting it with the qubit channel input via a beamsplitter of transmissivity $\eta$, and tracing over the environment qubit.  

Let $\mathcal{N}$ denote a generalized amplitude damping channel with parameters $(\eta_1, p_1)$, and 
let~$\mathcal{M}$ denote a generalized amplitude damping channel with parameters $(\eta_2, p_2)$.
In \eqref{Eq:AllStein}, observe that the Stein quantity is bounded from below by $D\left(\mathcal{N}\|\mathcal{M}\right)$ and from above by~$D_{\max}(\mathcal{N}\|\mathcal{M})\, $.  

To evaluate $D\left(\mathcal{N}\|\mathcal{M}\right)$, observe that the generalized amplitude damping channel is covariant with respect to $I$ and $Z$. We can therefore apply \cite[Proposition II.4]{PhysRevA.97.012332}, which states that it suffices to restrict the optimization to input states $\psi_{RA}$ such that their reduced state is of the form $\rho_{A}= z\mathinner{|0\rangle\langle 0|}+(1-z)\mathinner{|1\rangle\langle 1|}$, where $z\in[0,1]$. Let us choose a purification of $\rho_A$ as $\mathinner{|\psi(z)\rangle}_{RA}= \sqrt{z}\mathinner{|00\rangle}_{RA}+\sqrt{1-z}\mathinner{|11\rangle}_{RA}$. Since all purifications of a state are related by a unitary acting on the purifying system, we can invoke the unitary invariance of the relative entropy to find that
\begin{align}
D\!\left(\mathcal{N}\|\mathcal{M}\right) &= \max_{\psi_{RA}}D\!\left( \mathcal{N}_{A\rightarrow B}(\psi_{RA})\|
\mathcal{M}_{A\rightarrow B}(\psi_{RA})\right)\\
 &=\max_{z\in [0,1]}D\!\left( \mathcal{N}_{A\rightarrow B}(\psi(z)_{RA})\|\mathcal{M}_{A\rightarrow B}(\psi(z)_{RA})\right).
\end{align}
The latter quantity is straightforward to calculate numerically.

We also know from Lemma~\ref{lemma:dmax} that $D_{\max}(\mathcal{N}\|\mathcal{M})\,=D_{\max}(\mathcal{N}_{A\rightarrow B}(\Phi_{RA})\Vert\mathcal{M}_{A\rightarrow B}(\Phi_{RA}))$, the latter of which we said previously could be calculated via an SDP.  

If $\eta_1=\eta_2$, the generalized amplitude channels are environment-parametrized channels as defined in \eqref{eq:environ-param-1}--\eqref{eq:environ-param-2}, with $\theta_E^{\mathcal{N}}=p_1\vert 0 \rangle \langle 0 \vert + (1-p_1) \vert 1 \rangle \langle 1 \vert$ and $\theta_E^{\mathcal{M}}=p_2\vert 0 \rangle \langle 0 \vert + (1-p_2) \vert 1 \rangle \langle 1 \vert$. Therefore, by  Proposition~\ref{thm:env-param-simple-bounds-div}, $D(\theta_E^{\mathcal{N}}\|\theta_E^{\mathcal{M}})$ gives an upper bound in the Stein setting, in addition to the bound given by $D_{\max}(\mathcal{N}\|\mathcal{M})$.
We plot the difference of these bounds for particular generalized amplitude damping channels in Figures~\ref{fig:gen_amp_damp} and \ref{fig:gen_amp_damp1}.

\begin{figure}
\centering
\includegraphics[scale=0.5]{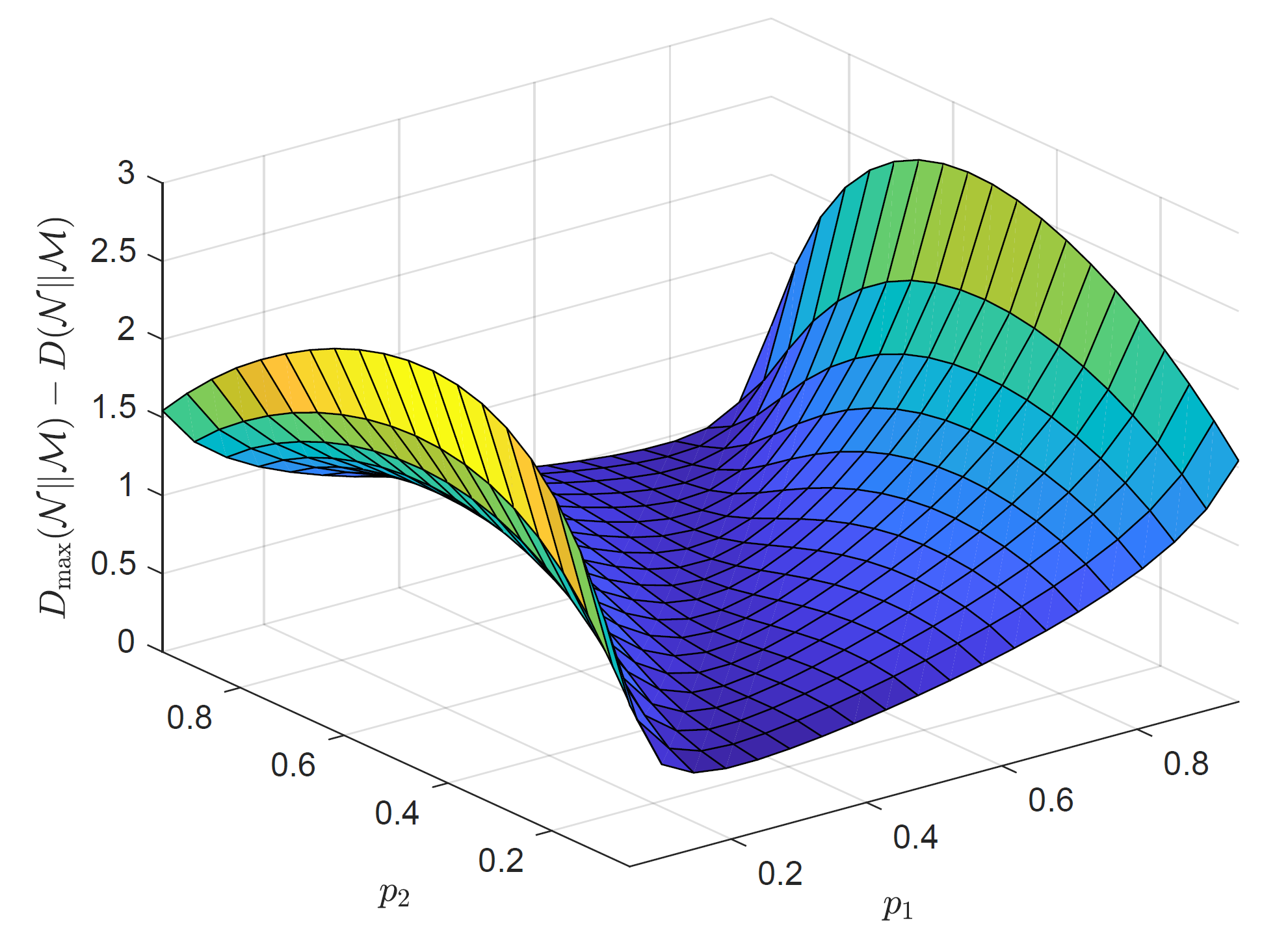}
\caption{This figure displays the difference between the upper and lower bounds in the Stein setting for the generalized amplitude damping channels with parameters $(\eta_1,p_1)$ and $(\eta_2,p_2)$. We vary the parameters $p_1$ and $p_2$ and fix the parameters $\eta_1=0.2$ and $\eta_2=0.3$.}
\label{fig:gen_amp_damp}
\end{figure}

\begin{figure}
\centering
\includegraphics[scale=0.5]{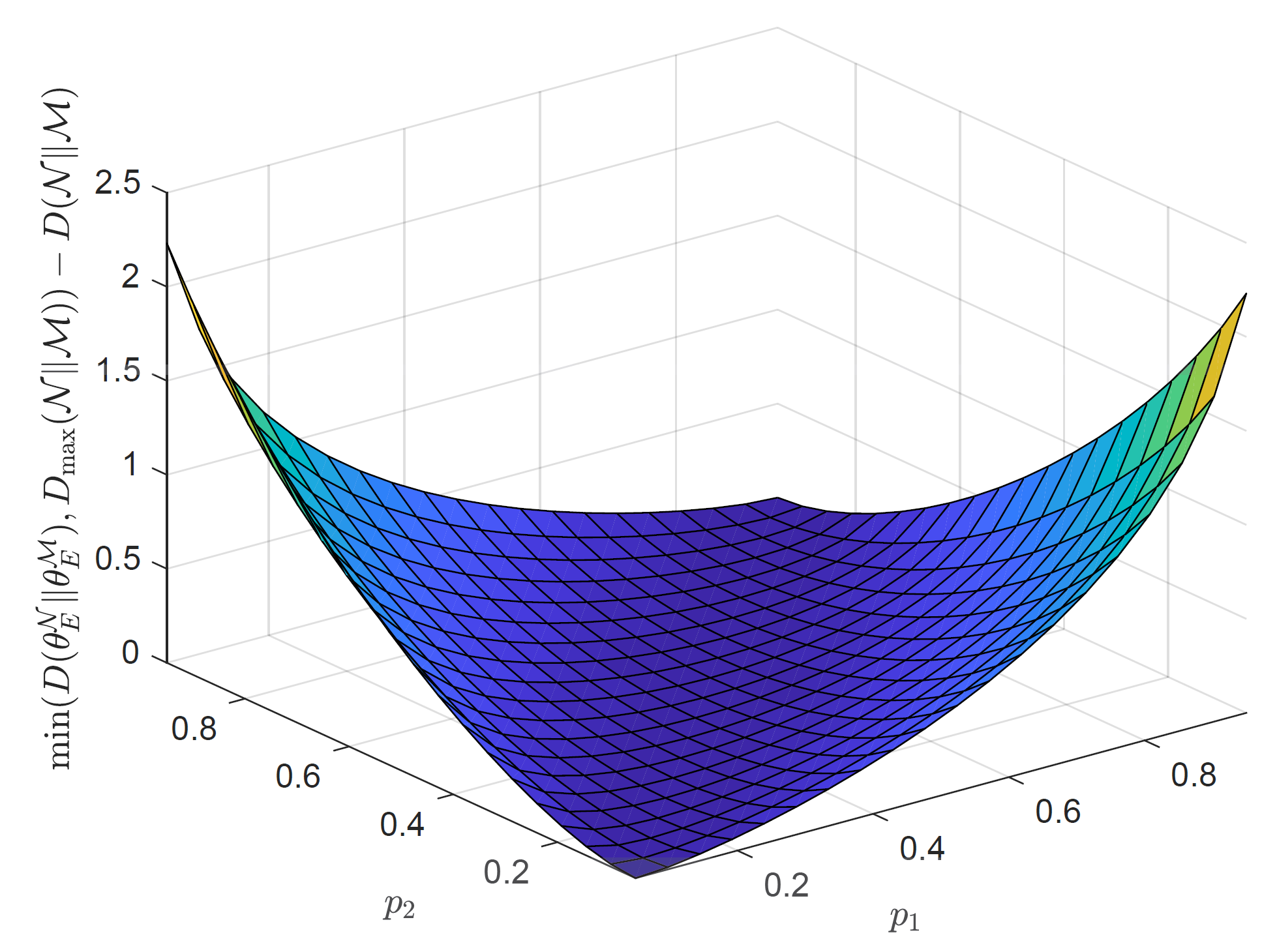}
\caption{This figure displays the difference between the upper and lower bounds in the Stein setting for the generalized amplitude damping channels with parameters $(\eta_1,p_1)$ and $(\eta_2,p_2)$. We vary the parameters $p_1$ and $p_2$ and fix the parameters $\eta_1=0.5$ and $\eta_2=0.5$.}
\label{fig:gen_amp_damp1}
\end{figure}


\subsection{Strong converse exponent} 

A bound on the achievability part of the strong converse exponent can, as before, simply be given by considering a
product-state channel discrimination strategy. Following the result from the state discrimination setting \cite{Mosonyi2015}, the achievable rate is given by a quantity involving the sandwiched R\'enyi divergence. It is not clear whether such a strategy is optimal, and so we also consider the optimality part. In the following theorem, we give a lower bound on the strong converse exponent involving the amortized sandwiched R\'enyi channel divergence. As discussed in later sections, for some channels it can be shown that the amortized sandwiched R\'enyi channel divergence collapses, such that the lower and upper bounds match. 

\begin{prop}\label{thm:converseexponent}
Let $\cN_{A\to B},\cM_{A\to B}\in\cQ(A\to B)$. Then, for $r>0$ we have that
\begin{align}
\overline{H}_n(r,\cN,\cM)
& \geq\sup_{\alpha>1}\frac{\alpha-1}{\alpha}\left(r-\widetilde{D}_\alpha^{\mathcal{A}}(\cN\|\cM)\right)\geq r - D_{\max}(\cN\Vert\cM).
\end{align}
\end{prop}

\begin{proof}
We follow the proof strategy from~\cite[Lemma 4.7]{Mosonyi2015}, combined with involving the amortized  sandwiched R\'enyi divergence.
Let $\{\cQ, \cA\}$ denote an arbitrary protocol for discrimination of the channels $\cN$ and $\cM$, as discussed in Section~\ref{sec:discrimination}, and let $p$ and $q$ denote the final decision probabilities.
By evaluating the sandwiched R\'enyi divergence, we have for $\alpha>1$ that
\begin{align}
\widetilde{D}_\alpha(p\|q)&=\frac{1}{\alpha-1} \log\Big(p^\alpha q^{1-\alpha} + (1-p)^\alpha(1-q)^{1-\alpha}\Big)\label{eq:classical-Renyi-1}\\
&\geq\frac{1}{\alpha-1} \log\Big(p^\alpha q^{1-\alpha}\Big)\\
&=\frac{\alpha}{\alpha-1}\log\Big(1-\alpha_n(\{Q,\mathcal{A}\})\Big)-\log\beta_n(\{Q,\mathcal{A}\})\\
&\geq\frac{\alpha}{\alpha-1}\log\Big(1-\alpha_n(\{Q,\mathcal{A}\})\Big)+nr,\label{eq:classical-Renyi-last}
\end{align}
where the last inequality follows from the constraint $\beta_n(\{Q,\mathcal{A}\})\leq 2^{-nr}$, the latter of which is taken as an assumption in this setting, as discussed in Section~\ref{sec:han-kob}. We then have that 
\begin{align}
-\frac{1}{n}\log\Big(1-\alpha_n(\{Q,\mathcal{A}\})\Big) & \geq\frac{\alpha-1}{\alpha}\left(r-\frac{1}{n}\widetilde{D}_\alpha(p\|q)\right)\\
& \geq\frac{\alpha-1}{\alpha}\left(r-\widetilde{D}_\alpha^{\mathcal{A}}(\cN\|\cM)\right).
\end{align}
The first inequality is a rewriting of \eqref{eq:classical-Renyi-1}--\eqref{eq:classical-Renyi-last}, and the second follows from the meta-converse  from Lemma~\ref{ThmAmortizedBound}, with the divergence chosen to be  the sandwiched R\'enyi divergence. Since the above holds for all $\alpha>1$, we obtain the desired result by taking the supremum over all such $\alpha$. 

The statement for the max-relative entropy follows because
\begin{align}
D_{\max}(p\Vert q) & \geq
\log(p/q) \\
&  = \log(1-\alpha_n(\{Q,\mathcal{A}\})) - \log (\beta_n(\{Q,\mathcal{A}\})) \\
& \geq \log(1-\alpha_n(\{Q,\mathcal{A}\})) + nr,
\end{align}
and then applying the meta-converse in Lemma~\ref{ThmAmortizedBound}, as well as the amortization collapse for max-relative entropy from Proposition~\ref{lem:D_max}. This gives
\begin{align}
-\frac{1}{n}\log(1-\alpha_n(\{Q,\mathcal{A}\})) & \geq r-\frac{1}{n}D_{\max}(p\|q) \\
& \geq r-D_{\max}(\cN\|\cM),
\end{align}
concluding the proof.
\end{proof}


\subsection{Chernoff bound}\label{sec:chernoff-general}

For non-adaptive protocols, when we restrict the input states to be product states\,---\,but still allow for a quantum memory reference system $R$\,---\,it directly follows from the Chernoff bound for quantum state discrimination~\cite{ACMBMAV07,Audenaert2008,NS09} that the symmetric error exponent is given by the Chernoff channel divergence $C(\cN\|\cM)$. This obviously already gives an achievability bound for the adaptive setting as well. However, note that the results from~\cite{Duan09,Harrow10,Yu17} establish that the Chernoff channel divergence $C(\cN\|\cM)$ does not generally quantify the symmetric error exponent for quantum channel discrimination.

In what follows, we are interested in converse bounds for the general adaptive setting. We begin by establishing a bound on the non-asymptotic symmetric error exponent $\xi_{n}(\mathcal{N},\mathcal{M},p)$. Aspects of the proof approach are related to that from \cite[Theorem 9]{Yu17} and \cite[Proposition~1]{CE18}.

\begin{prop}\label{thm:non-asym-chern}
Given quantum channels $\mathcal{N}_{A\rightarrow B}$ and $\mathcal{M}
_{A\rightarrow B}$ and $p\in(0,1)$, the following bound holds for all $n \in \mathbb{N}$:
\begin{equation}
\xi_{n}(\mathcal{N},\mathcal{M},p)\leq\min\Big\{\widetilde{D}_{1/2}^{\mathcal{A}}(\mathcal{N}\Vert\mathcal{M}),D_{\max}(\cN \Vert \cM),D_{\max}(\cM \Vert \cN )\Big\}-\frac{1}{n}\log\left[p\left(1-p\right)\right].
\end{equation}
\end{prop}

\begin{proof}
Invoking Lemma~\ref{lem:FvdG-PSD} from Appendix~\ref{app:gen-fuchs-vdgrf} (see also \cite[Supplementary~Lemma~3]{CKW14}), the following inequality holds for positive
semi-definite $A$ and~$B$:
\begin{equation}
\left\Vert A-B\right\Vert _{1}^{2}+4\left\Vert \sqrt{A}\sqrt{B}\right\Vert
_{1}^{2}\leq\left(  \operatorname{Tr}\{A+B\}\right)  ^{2}.
\end{equation}
For $p\in(0,1)$, and $\rho$ and $\sigma$ density operators, we then find that
\begin{equation}
\left\Vert p\rho-\left(  1-p\right)  \sigma\right\Vert _{1}^{2}+4p\left(
1-p\right)  F(\rho,\sigma)\leq1,
\end{equation}
where $F(\rho,\sigma)\coloneqq\Vert\sqrt{\rho}\sqrt{\sigma}\Vert_{1}^{2}$ is the
quantum fidelity. Note that $4p\left(  1-p\right)  \in\left[  0,1\right]  $
for $p\in\left[  0,1\right]  $.\ Rewriting the above expression, we find that
\begin{equation}
\left\Vert p\rho-\left(  1-p\right)  \sigma\right\Vert _{1}\leq\sqrt
{1-4p\left(  1-p\right)  F(\rho,\sigma)},
\end{equation}
which is equivalent to
\begin{align}
\frac{1}{2}\left(  1-\left\Vert p\rho-\left(  1-p\right)  \sigma\right\Vert
_{1}\right)   &  \geq\frac{1}{2}\left(  1-\sqrt{1-4p\left(  1-p\right)
F(\rho,\sigma)}\right)  \\
&  \geq p\left(  1-p\right)  F(\rho,\sigma),
\end{align}
where we have employed the inequality $\frac{1}{2}\left(  1-\sqrt{1-x}\right)
\geq\frac{x}{4}$, which holds for $x\in\left[  0,1\right]  $. Taking a
negative logarithm, this can be rewritten as
\begin{equation}\label{eq:fidelity-to-trace-gen}
-\log\left(\frac{1}{2}\left(  1-\left\Vert p\rho-\left(  1-p\right)\sigma\right\Vert _{1}\right)  \right)  \leq-\log\left(p\left(  1-p\right)\right)+\widetilde{D}_{1/2}(\rho\Vert\sigma).
\end{equation}
Picking $\rho$ and $\sigma$ to be the final states in an adaptive protocol for
distinguishing $n$ uses of the channels $\mathcal{N}$ and $\mathcal{M}$, and
applying the meta-converse in Lemma~\ref{ThmAmortizedBound}, we then find that
\begin{equation}
\xi_{n}(\mathcal{N},\mathcal{M},p)\leq-\frac{1}{n}\log\left(p\left(
1-p\right)  \right)  +\widetilde{D}_{1/2}^{\mathcal{A}}(\mathcal{N}
\Vert\mathcal{M}).
\end{equation}

To establish the bound in terms of max-channel divergence, we start from \eqref{eq:fidelity-to-trace-gen}, and employ the inequality 
$\widetilde{D}_{1/2}(\rho\Vert\sigma) \leq D_{\max}(\rho\Vert\sigma)$ \cite{muller2013quantum}, along with 
the meta-converse in Lemma~\ref{ThmAmortizedBound} and the amortization collapse for max-relative entropy from Proposition~\ref{lem:D_max}, giving that
\begin{equation}
\xi_{n}(\mathcal{N},\mathcal{M},p)\leq-\frac{1}{n}\log\left(p\left(
1-p\right)  \right)  +D_{\max}(\mathcal{N}
\Vert\mathcal{M}).
\end{equation}
Due to the symmetry $F(\rho,\sigma) = F(\sigma,\rho)$, the same bound with the quantum channels $\cN$ and $\cM$ interchanged holds as well.
\end{proof}

\begin{remark}\label{rem:fid-div}
Recently and in independent work, the following case of Proposition \ref{thm:non-asym-chern} was established in~\cite[Proposition 1]{CE18}:
\begin{equation}
\xi_{n}(\mathcal{N},\mathcal{M},1/2)\leq\frac{2}{n}+\widetilde{D}_{1/2}^{\mathcal{A}}(\mathcal{N}\Vert\mathcal{M}).
\end{equation}
The authors of \cite{CE18} also defined the concept of fidelity divergence of quantum channels, which is equal to $\widetilde{D}_{1/2}^{\mathcal{A}}(\mathcal{N} \Vert\mathcal{M})$ after taking a negative logarithm. They remark that it is sufficient to restrict the optimization over states in $\widetilde{D}_{1/2}^{\mathcal{A}}(\mathcal{N}\Vert\mathcal{M})$ to pure states, but unfortunately this restriction does not imply a limit on the states' dimension. 
\end{remark}

Proposition \ref{thm:non-asym-chern} implies the following bounds on the asymptotic symmetric error exponent.

\begin{corollary}\label{cor:chernoff}
Let $\cN_{A\to B},\cM_{A\to B}\in\cQ(A\to B)$. Then, we have that
\begin{equation}
\overline{\xi}(\cN,\cM)\leq\min\Big\{\widetilde{D}_{1/2}^{\mathcal{A}}(\mathcal{N}\Vert\mathcal{M}),D_{\max}(\cN \Vert \cM), D_{\max}(\cM \Vert \cN )\Big\}.
\end{equation}
\end{corollary}

Note that \cite[Lemma 10]{Yu17} gave a non-explicit upper bound on $\widetilde{D}_{1/2}^{\mathcal{A}}(\mathcal{N}\Vert\mathcal{M})$, that is provably finite in the case that the channels are not perfectly distinguishable by an adaptive channel discrimination protocol. This upper bound establishes that it indeed makes sense to define the symmetric error exponent for general quantum channel discrimination (even though it is not generally given by the Chernoff channel divergence).


\section{Classical-quantum channel discrimination}\label{sec:cq}

In this section, we extend results from the classical setting~\cite{Hayashi09} to classical-quantum channel discrimination. We consider classical-quantum channels that act as
\begin{align}
\mathcal{N}_{X\rightarrow B}(\cdot)  &  =\sum_{x}\langle x|\cdot|x\rangle\nu_{B}^{x}\label{eq:cq1},\\
\mathcal{M}_{X\rightarrow B}(\cdot)  &  =\sum_{x}\langle x|\cdot|x\rangle\mu_{B}^{x},\label{eq:cq2}
\end{align}
where $\{ | x\rangle\}_x$ is an orthonormal basis and $\{\nu^x_B\}_x$ and $\{\mu^x_B\}_x$ are sets of states.

We find in several cases that the optimal classical-quantum channel discrimination protocol is to pick the best possible input and then to apply a tensor-power strategy. This result implies that adaptive strategies, quantum memories, and entangled inputs are of no use in some of the asymptotic settings. Note that this slightly extends the classical setting as well, in the sense that it was previously unclear if quantum memories could be of any help.


\subsection{Stein bound}\label{sec:stein-cq}

We start with the Stein's lemma for classical-quantum channels.

\begin{theorem}\label{thm:cqAmortized}
Let $\cN_{X\to B},\cM_{X\to B}\in\cQ(X\to B)$ be classical-quantum channels, as defined in Eqs.~\eqref{eq:cq1} and \eqref{eq:cq2}. Then, we have that
\begin{align}
\zeta(\cN,\cM)\coloneqq\lim_{\eps\to 0}\lim_{n\to\infty}\zeta_n(\eps,\cN,\cM)=\max_{x}D(\nu_{B}^{x}\Vert\mu_{B}^{x}). \label{eq:cqFull}
\end{align}
\end{theorem}

\begin{proof}
The achievability part follows directly by employing a product-state discrimination strategy. Therefore, it remains to show the converse direction. We know that the amortized quantum relative entropy divergence $D^{\mathcal{A}}(\cN\|\cM)$ provides a weak converse rate (Proposition~\ref{weakConverse}). The missing step is to evaluate that quantity, which is done in the following Lemma~\ref{lem:cqCollapse}. 
\end{proof}

\begin{lemma}
\label{lem:cqCollapse}
Let $\cN_{X\to B},\cM_{X\to B}\in\cQ(X\to B)$ be classical-quantum channels, as defined in Eqs.~\eqref{eq:cq1} and \eqref{eq:cq2}. Then the following amortization collapse occurs for the quantum relative entropy:
\begin{align}
D^{\mathcal{A}}(\mathcal{N}\Vert\mathcal{M})=\max_{x}D(\nu_{B}^{x}\Vert\mu_{B}^{x}).
\end{align}
\end{lemma}
\begin{proof}
We show this as follows. 
The following inequality is a trivial consequence of picking $\rho_{RA}=\sigma_{RA}=|x\rangle\langle x|_{R}\otimes|x\rangle\langle x|_{A}$:
\begin{align}
D^{\mathcal{A}}(\mathcal{N}\Vert\mathcal{M})\geq D(\nu_{B}^{x}\Vert\mu_{B}^{x}).
\end{align}
Since it holds for all $x$, we conclude that
\begin{align}
D^{\mathcal{A}}(\mathcal{N}\Vert\mathcal{M})\geq\max_{x}D(\nu_{B}^{x}\Vert\tau
_{B}^{x}).
\end{align}

To see the other inequality, consider for any states $\rho_{RA}$ and
$\sigma_{RA}$ that
\begin{align}
\mathcal{N}_{A\rightarrow B}(\rho_{RA}) &  =\sum_{x}p(x)\rho_{R}^{x}\otimes\nu_{B}^{x},\\
\mathcal{M}_{A\rightarrow B}(\sigma_{RA}) &  =\sum_{x}q(x)\sigma_{R}^{x}\otimes\mu_{B}^{x},
\end{align}
where $p(x)\rho_{R}^{x}\coloneqq\langle x|_{A}\rho_{RA}|x\rangle_{A}$ and $q(x)\sigma_{R}^{x}\coloneqq\langle x|_{A}\sigma_{RA}|x\rangle_{A}$, with $p(x)$ and $q(x)$ probability distributions and 
$\{\rho^x_R\}_x$ and $\{\sigma^x_R\}_x$ sets of states. Now recall the following property of quantum relative entropy from \cite[Exercise~11.8.8]{W13}:
\begin{equation}
D\left(\sum_{x} p(x) |x\rangle\langle x| \otimes \rho^x_R \middle \| \sum_{x} q(x) |x\rangle\langle x| \otimes \sigma^x_R\right) = D(p\| q) +  \sum_x p(x) D(\rho^x_R \| \sigma^x_R).
\label{eq:rel-ent-direct-sum}
\end{equation}
 Then, we have that
\begin{align}
&D(\mathcal{N}_{A\rightarrow B}(\rho_{RA})\Vert\mathcal{M}_{A\rightarrow B}(\sigma_{RA}))-D(\rho_{RA}\Vert\sigma_{RA})\nonumber\\
&=D\!\left(\sum_{x}p(x)\rho_{R}^{x}\otimes\nu_{B}^{x}\middle\|\sum_{x}q(x)\sigma_{R}^{x}\otimes\mu_{B}^{x}\right)  -D(\rho_{RA}\Vert\sigma_{RA})\\
&\leq D\!\left(\sum_{x}p(x)\rho_{R}^{x}\otimes\nu_{B}^{x}\middle\|\sum_{x}q(x)\sigma_{R}^{x}\otimes\mu_{B}^{x}\right)\notag\\
& \qquad -D\!\left(\sum_{x}p(x)\rho_{R}^{x}\otimes|x\rangle\langle x|_{X}\otimes\nu_{B}^{x}\middle\Vert\sum_{x}q(x)\sigma_{R}^{x}\otimes|x\rangle\langle x|_{X}\otimes\nu_{B}^{x}\right)\\
& \leq D\!\left(\sum_{x}p(x)\rho_{R}^{x}\otimes|x\rangle\langle x|_{X} \otimes\nu_{B}^{x}\middle\Vert\sum_{x}q(x)\sigma_{R}^{x}\otimes|x\rangle\langle x|_{X}\otimes\mu_{B}^{x}\right)\nonumber\\
&\qquad-D\!\left(\sum_{x}p(x)\rho_{R}^{x}\otimes|x\rangle\langle x|_{X}\otimes\nu_{B}^{x}\middle\Vert\sum_{x}q(x)\sigma_{R}^{x}\otimes|x\rangle\langle x|_{X}\otimes\nu_{B}^{x}\right)\\
&= D(p\| q) + \sum_x p(x) D( \rho_{R}^{x}\otimes\nu_{B}^{x} \| \sigma_{R}^{x}\otimes\mu_{B}^{x})
- \left(D(p\| q) + \sum_x p(x) D( \rho_{R}^{x}\otimes\nu_{B}^{x} \| \sigma_{R}^{x}\otimes\nu_{B}^{x})\right)\\
&=  \sum_x p(x) \left[D( \rho_{R}^{x} \| \sigma_{R}^{x}) + 
D( \nu_{B}^{x} \| \mu_{B}^{x})\right]
-  \sum_x p(x) \left[D( \rho_{R}^{x} \| \sigma_{R}^{x})
+D( \nu_{B}^{x} \| \nu_{B}^{x})\right]\\
&=\sum_{x}p(x)D(\nu_{B}^{x}\Vert\mu_{B}^{x})\\
&\leq\max_{x}D(\nu_{B}^{x}\Vert\mu_{B}^{x}).
\end{align}
The first two inequalities follow from data processing: the first from $D(\rho_{RA} \Vert \sigma_{RA}) \geq D(\cN_{A\to B}(\rho_{RA}) \Vert \cN_{A\to B}(\sigma_{RA}))$ and the second from partial trace. The second equality follows from the identity in \eqref{eq:rel-ent-direct-sum}.
Since the above development holds for arbitrary states $\rho_{RA}$ and $\sigma_{RA}$, we find that
\begin{align}
D^{\mathcal{A}}(\mathcal{N}\Vert\mathcal{M})\leq\max_{x}D(\nu_{B}^{x}\Vert\mu_{B}^{x}),
\end{align}
concluding the proof.
\end{proof}


\subsection{Strong converse exponent}

In this section, we prove that for classical-quantum channels the strong converse exponent from the general case in Proposition \ref{thm:converseexponent} matches the exponent achievable with product states. Similar to the proof in the previous section, the collapse of the amortized quantity plays an important role in this proof.

\begin{lemma}\label{lem:RenyiCollapse}
Let $\cN_{X\to B},\cM_{X\to B}\in\cQ(X\to B)$ be classical-quantum channels, as defined in Eqs.~\eqref{eq:cq1} and \eqref{eq:cq2}. Then, we have that
\begin{align}
D_\alpha^{\mathcal{A}}(\mathcal{N}\Vert\mathcal{M})& =\max_{x}D_\alpha(\nu_{B}^{x}\Vert\mu_{B}^{x})\;\text{for $\alpha\in[0,2]$},\quad\text{ as well as,}\\
\widetilde{D}_\alpha^{\mathcal{A}}(\mathcal{N}\Vert\mathcal{M})
& =\max_{x}\widetilde{D}_\alpha(\nu_{B}^{x}\Vert\mu_{B}^{x})\;\text{for $\alpha\geq\frac{1}{2}$}.
\end{align}
\end{lemma}

\begin{proof}
We detail the proof for the Petz-R\'enyi divergences, and note that the proof for the sandwiched R\'enyi divergences is essentially the same. In fact, the key ideas are similar to those used in the proof of Lemma~\ref{lem:cqCollapse}. The following inequality is a trivial consequence of picking $\rho_{RA}=\sigma_{RA}=|x\rangle\langle x|_{R}\otimes|x\rangle\langle x|_{A}$:
\begin{align}
D_\alpha^{\mathcal{A}}(\mathcal{N}\Vert\mathcal{M})\geq D_\alpha(\nu_{B}^{x}\Vert\mu_{B}^{x}).
\end{align}
Since it holds for all $x$, we conclude that
\begin{align}
D_\alpha^{\mathcal{A}}(\mathcal{N}\Vert\mathcal{M})\geq\max_{x}D_\alpha(\nu_{B}^{x}\Vert\mu_{B}^{x}).
\end{align}

To see the other inequality, consider for any states $\rho_{RA}$ and $\sigma_{RA}$ that
\begin{align}
\mathcal{N}_{A\rightarrow B}(\rho_{RA}) &  =\sum_{x}p(x)\rho_{R}^{x}\otimes\nu_{B}^{x}\\
\mathcal{M}_{A\rightarrow B}(\sigma_{RA}) &  =\sum_{x}q(x)\sigma_{R}^{x}\otimes\mu_{B}^{x}.
\end{align}
where $p(x)\rho_{R}^{x}\coloneqq\langle x|_{A}\rho_{RA}|x\rangle_{A}$ and $q(x)\sigma_{R}^{x}\coloneqq\langle x|_{A}\sigma_{RA}|x\rangle_{A}$. Then, we have that
\begin{align}
&D_\alpha(\mathcal{N}_{A\rightarrow B}(\rho_{RA})\Vert\mathcal{M}_{A\rightarrow B}(\sigma_{RA}))-D_\alpha(\rho_{RA}\Vert\sigma_{RA})\nonumber\\
&=D_\alpha\!\left(\sum_{x}p(x)\rho_{R}^{x}\otimes\nu_{B}^{x}\middle\Vert\sum_{x}q(x)\sigma_{R}^{x}\otimes\mu_{B}^{x}\right)-D_\alpha(\rho_{RA}\Vert\sigma_{RA})\\
&\leq D_\alpha\!\left(\sum_{x}p(x)\rho_{R}^{x}\otimes\nu_{B}^{x}\middle\Vert\sum_{x}q(x)\sigma_{R}^{x}\otimes\mu_{B}^{x}\right)\notag \\
& \qquad -D_\alpha\!\left(\sum_{x}p(x)\rho_{R}^{x}\otimes|x\rangle\langle x|_{X}\otimes\nu_{B}^{x}\middle\Vert\sum_{x}q(x)\sigma_{R}^{x}\otimes|x\rangle\langle x|_{X}\otimes\nu_{B}^{x}\right)\\
&\leq D_\alpha\!\left(\sum_{x}p(x)\rho_{R}^{x}\otimes|x\rangle\langle x|_{X}\otimes\nu_{B}^{x}\middle\Vert\sum_{x}q(x)\sigma_{R}^{x}\otimes|x\rangle\langle x|_{X}\otimes\mu_{B}^{x}\right)\nonumber\\
&\qquad-D_\alpha\!\left(\sum_{x}p(x)\rho_{R}^{x}\otimes|x\rangle\langle x|_{X}\otimes\nu_{B}^{x}\middle\Vert\sum_{x}q(x)\sigma_{R}^{x}\otimes|x\rangle\langle x|_{X}\otimes\nu_{B}^{x}\right)\\
&=\frac{1}{\alpha-1}\log\operatorname{Tr}\!\left[\left(\sum_{x}p(x)\rho_{R}^{x}\otimes|x\rangle\langle x|_{X}\otimes\nu_{B}^{x}\right)^\alpha\left(\sum_{x}q(x)\sigma_{R}^{x}\otimes|x\rangle\langle x|_{X}\otimes\mu_{B}^{x}\right)^{1-\alpha}\right]\nonumber\\
&\qquad-\frac{1}{\alpha-1}\log\operatorname{Tr}\!\left[\left(\sum_{x}p(x)\rho_{R}^{x}\otimes|x\rangle\langle x|_{X}\otimes\nu_{B}^{x}\right)^\alpha\left(\sum_{x}q(x)\sigma_{R}^{x}\otimes|x\rangle\langle x|_{X}\otimes\nu_{B}^{x}\right)^{1-\alpha}\right] \\
&=\frac{1}{\alpha-1}\log\frac{\sum_{x}p(x)^\alpha q(x)^{1-\alpha}\operatorname{Tr}\Big[\left(\rho_{R}^{x}\right)^\alpha\left(\sigma_{R}^{x}\right)^{1-\alpha}\Big]\operatorname{Tr}\Big[\left(\nu_{B}^{x}\right)^\alpha\left(\mu_{B}^{x}\right)^{1-\alpha}\Big]}{\sum_{x}p(x)^\alpha q(x)^{1-\alpha}\operatorname{Tr}\Big[\left(\rho_{R}^{x}\right)^\alpha\left(\sigma_{R}^{x}\right)^{1-\alpha}\Big]\operatorname{Tr}\Big[\left(\nu_{B}^{x}\right)^\alpha\left(\nu_{B}^{x}\right)^{1-\alpha}\Big]}\\
&=\frac{1}{\alpha-1}\log\frac{\sum_{x}p(x)^\alpha q(x)^{1-\alpha}\operatorname{Tr}\Big[\left(\rho_{R}^{x}\right)^\alpha\left(\sigma_{R}^{x}\right)^{1-\alpha}\Big]\operatorname{Tr}\Big[\left(\nu_{B}^{x}\right)^\alpha\left(\mu_{B}^{x}\right)^{1-\alpha}\Big]}{\sum_{x}p(x)^\alpha q(x)^{1-\alpha}\operatorname{Tr}\Big[\left(\rho_{R}^{x}\right)^\alpha\left(\sigma_{R}^{x}\right)^{1-\alpha}\Big]}.
\end{align}
Defining the probability distribution
\begin{align}
r(x)=\frac{p(x)^\alpha q(x)^{1-\alpha}\operatorname{Tr}\Big[\left(\rho_{R}^{x}\right)^\alpha\left(\sigma_{R}^{x}\right)^{1-\alpha}\Big]}{\sum_{x}p(x)^\alpha q(x)^{1-\alpha}\operatorname{Tr}\Big[\left(\rho_{R}^{x}\right)^\alpha\left(\sigma_{R}^{x}\right)^{1-\alpha}\Big]},
\end{align}
we see that the above is equal to
\begin{align}
\frac{1}{\alpha-1}\log\sum_{x}r(x)\operatorname{Tr}\Big[\left(\nu_{B}^{x}\right)^\alpha\left(\mu_{B}^{x}\right)^{1-\alpha}\Big]&\leq\max_{x}\frac{1}{\alpha-1}\log\operatorname{Tr}\Big[\left(\nu_{B}^{x}\right)^\alpha\left(\mu_{B}^{x}\right)^{1-\alpha}\Big]\\
&=\max_{x}D_\alpha(\nu_{B}^{x}\Vert\mu_{B}^{x}).
\end{align}
Since the above development holds for arbitrary states $\rho_{RA}$ and $\sigma_{RA}$, we find that
\begin{align}
D_\alpha^{\mathcal{A}}(\mathcal{N}\Vert\mathcal{M})\leq\max_{x}D_\alpha(\nu_{B}^{x}\Vert\mu_{B}^{x}),
\end{align}
concluding the proof.
\end{proof}

We are now ready to state the strong converse exponent for the discrimination of classical-quantum channels. 

\begin{theorem}\label{thm:cq-str-conv-exp}
Let $\cN_{X\to B},\cM_{X\to B}\in\cQ(X\to B)$ be classical-quantum channels, as defined in \eqref{eq:cq1} and \eqref{eq:cq2}. Then, for $r>0$ we have that
\begin{align}
H(r,\cN,\cM)\coloneqq\lim_{n\to\infty}H_n(r,\cN,\cM)=\sup_{\alpha>1}\frac{\alpha-1}{\alpha}\left(r-\max_x\widetilde{D}_\alpha(\nu_{B}^{x}\|\mu_{B}^{x})\right).
\end{align}
\end{theorem}

\begin{proof}
For the optimality part, i.e.,
\begin{equation}\label{eq:lwr-bnd-sc-exp}
H(r,\cN,\cM) \geq \sup_{\alpha>1}\frac{\alpha-1}{\alpha}\left(r-\max_x\widetilde{D}_\alpha(\nu_{B}^{x}\|\mu_{B}^{x})\right),
\end{equation}
we combine Proposition \ref{thm:converseexponent} with Lemma~\ref{lem:RenyiCollapse}, and the result follows immediately. It is helpful to rewrite this lower bound by allowing for an optimization over probability distributions on the input letters $x$. With this in mind, consider that
\begin{equation}\label{eq:letter-to-dist}
\max_{x}\widetilde{D}_\alpha(\nu_{B}^{x}\Vert\mu_{B}^{x})=\max_{p_{X}}\widetilde{D}_\alpha(\nu_{XB}\Vert\mu_{XB}),
\end{equation}
where the second maximum is with respect to a probability distribution $p_{X}$ and
\begin{equation}
\nu_{XB}   \coloneqq\sum_{x}p_{X}(x)|x\rangle\langle x|_{X}\otimes\nu_{B}^{x},\qquad
\mu_{XB}   \coloneqq\sum_{x}p_{X}(x)|x\rangle\langle x|_{X}\otimes\mu_{B}^{x}.
\end{equation}
To see the equality in \eqref{eq:letter-to-dist}, let $x^{\ast}$ be the
optimal choice for $\max_{x}\widetilde{D}_\alpha(\nu_{B}^{x}\Vert\mu_{B}^{x})$.
Then the distribution $p_{X}(x)=\delta_{x^{\ast},x}$ is a particular choice
for the optimization on the right, so that
\begin{equation}
\max_{x}\widetilde{D}_\alpha(\nu_{B}^{x}\Vert\mu_{B}^{x})\leq\max_{p_{X}}\widetilde{D}_\alpha(\nu_{XB}\Vert\mu_{XB}).
\end{equation}
For the other direction, consider that the sandwiched R\'enyi relative entropy is quasi-jointly-concave \cite{muller2013quantum,WWY14}, so that the following inequality holds for an arbitrary probability distribution $p_X$:
\begin{equation}
\widetilde{D}_\alpha(\nu_{XB}\Vert\mu_{XB})\leq\max_{x}\widetilde{D}_\alpha(|x\rangle\langle x|_{X}\otimes\nu_{B}^{x}\Vert|x\rangle\langle x|_{X}\otimes\mu_{B}^{x})=\max_{x}\widetilde{D}_\alpha(\nu_{B}^{x}\Vert\mu_{B}^{x}).
\end{equation}
So this means that we can rewrite the lower bound in \eqref{eq:lwr-bnd-sc-exp} as
\begin{equation}\label{eq:lwr-bnd-sc-exp-rewrite}
H(r,\cN,\cM)\geq\sup_{\alpha>1}\min_{p_X}\frac{\alpha-1}{\alpha}\left(r-\widetilde{D}_\alpha(\nu_{XB}\|\mu_{XB})\right).
\end{equation}
The following upper bound on the strong converse exponent is a consequence of implementing a product-state strategy, using the state discrimination result from \cite[Lemma 4.17]{Mosonyi2015}, i.e., inputting one share of the classical state $\sum_x p_X(x) |x\rangle\langle x | \otimes |x\rangle\langle x |$ for every channel use:
\begin{equation}\label{eq:from-state-disc-SC}
H(r,\cN,\cM)\leq\min_{p_X}\sup_{\alpha>1}\frac{\alpha-1}{\alpha}\left(r-\widetilde{D}_\alpha(\nu_{XB}\|\mu_{XB})\right).
\end{equation}
Thus, in light of the inequalities in 
\eqref{eq:lwr-bnd-sc-exp-rewrite} and \eqref{eq:from-state-disc-SC}, now our aim is to close off the proof by establishing the following equality, which is equivalent to establishing that a minimax exchange is possible:
\begin{equation}\label{eq:minimax-cq}
\sup_{\alpha>1}\min_{p_{X}}\frac{\alpha-1}{\alpha}\left(r-\widetilde{D}_\alpha(\nu_{XB}\Vert\mu_{XB})\right)=\min_{p_{X}}\sup_{\alpha>1}\frac{\alpha-1}{\alpha}\left(r-\widetilde{D}_\alpha(\nu_{XB}\Vert\mu_{XB})\right).
\end{equation}
Strategies for establishing such an equality were given in \cite{Hayashi09,Cooney2016}. Here, we follow 
the proof of \cite[Theorem 2]{Cooney2016}. Let us define the function
\begin{equation}
F(\alpha,p_{X})\coloneqq\left(\alpha-1\right)\widetilde{D}_\alpha(\nu_{XB}\Vert\mu_{XB}).
\end{equation}
Introducing the new variable $u\coloneqq\left(\alpha-1\right)/\alpha$, so that $u\in(0,1)$ for $\alpha>1$, the minimax
statement in \eqref{eq:minimax-cq} is equivalent to the following one:
\begin{equation}\label{eq:minimax-u-world}
\sup_{u\in(0,1)}\min_{p_{X}}f(u,p_{X})=\min_{p_{X}}\sup_{u\in(0,1)}f(u,p_{X}),
\end{equation}
where
\begin{equation}
f(u,p_{X})\coloneqq ur-\widetilde{F}(u,p_{X}),\qquad\widetilde{F}(u,p_{X})\coloneqq(1-u)F\!\left(  \frac{1}{1-u},p_{X}\right).
\end{equation}
By the fact that
\begin{align}
F(\alpha,p_{X})=\log\widetilde{Q}_\alpha(\nu_{XB}\Vert\mu_{XB})=\log\!\left(\sum_{x}p_{X}(x)\widetilde{Q}_\alpha(\nu_{B}^{x}\Vert\mu_{B}^{x})\right),
\end{align}
with $\widetilde{Q}_\alpha(\rho\Vert \sigma)=\tr[(\sigma^{(1-\alpha)/2\alpha}\rho\sigma^{(1-\alpha)/2\alpha})^\alpha]$ the sandwiched R\'enyi relative quasi-entropy, it follows from concavity of the logarithm that the function $p_{X}\mapsto F(\alpha,p_{X})$ is concave, and so then $p_{X}\mapsto f(u,p_{X})$ is convex.
On the other hand, the function $u\mapsto F(u,p_{X})$ is convex by
\cite[Corollary 3.1]{Mosonyi2015}, and by \cite[Lemma 13]{Cooney2016}, it follows that the
function $u\mapsto\widetilde{F}(u,p_{X})$ is also convex. This then means
that the function $u\mapsto f(u,p_{X})$ is concave. From the assumption
that the support condition $\operatorname{supp}(\nu_{B}^{x})\subseteq
\operatorname{supp}(\mu_{B}^{x})$ holds for all $x$, it is clear that the
function $p_{X}\mapsto f(u,p_{X})$ is continuous for all $u\in(0,1)$.
Since the space of probability distributions $p_{X}$ is compact, the
Kneser-Fan minimax theorem \cite{Kneser,Fan} implies \eqref{eq:minimax-u-world}.
\end{proof}

Theorem~\ref{thm:cq-str-conv-exp} above in fact implies the following strong variant of the Stein's lemma for classical-quantum channels.

\begin{corollary}\label{cor:stein-str-conv}
Let $\cN_{X\to B},\cM_{X\to B}\in\cQ(X\to B)$ be classical-quantum channels, as defined in \eqref{eq:cq1} and \eqref{eq:cq2}. Then, for all $\eps\in(0,1)$, we have
\begin{align}
\zeta(\eps,\cN,\cM)\coloneqq\lim_{n\to\infty}\zeta_n(\eps,\cN,\cM)=\max_{x}D(\nu_{B}^{x}\Vert\mu_{B}^{x}). 
\end{align}
\end{corollary}

\begin{proof}
In Theorem~\ref{thm:cq-str-conv-exp}, if $r > \max_x D(\nu_{B}^{x}\|\mu_{B}^{x})$, then by the fact that $\max_x\widetilde{D}_\alpha(\nu_{B}^{x}\|\mu_{B}^{x})$ is monotone increasing with $\alpha$ and the continuity
\begin{align}
\lim_{\alpha\to1} \max_x\widetilde{D}_\alpha(\nu_{B}^{x}\|\mu_{B}^{x}) = \max_x D(\nu_{B}^{x}\|\mu_{B}^{x}),
\end{align}
there exists $\alpha>1$ such that $r > \max_x\widetilde{D}_\alpha(\nu_{B}^{x}\|\mu_{B}^{x})$, and so $H(r,\cN,\cM)>0$, implying that the Type~I error probability tends to one exponentially fast.
\end{proof}

\begin{remark}
In contrast to the weak and strong Stein's lemma (Theorem~\ref{thm:cqAmortized} and Corollary~\ref{cor:stein-str-conv}), we cannot conclude that the strong converse exponent in Theorem \ref{thm:cq-str-conv-exp} is achieved by picking the best possible input element $x$, but we instead have to consider distributions over the input alphabet. This is similar to the classical case, and Hayashi in fact gives an explicit example where considering only one input element $x$ is not sufficient~\cite[Section IV]{Hayashi09}. He then  shows that, in the classical case, it suffices to optimize with respect to probability distributions that are strictly positive on just two elements~\cite[Theorem 3]{Hayashi09}.
\end{remark}


\subsection{Error exponent}

In this section, we give bounds on the Hoeffding error exponent for channel discrimination of classical-quantum channels. First, we provide an upper bound in terms of the log-Euclidean R\'enyi divergence, by using a technique related to that used to establish \cite[Equation (16)]{Hayashi09}. Note that, contrary to the classical case, this development does   not generally lead to a tight characterisation. We suspect that this gap can be closed with an improved proof strategy\,---\,which would also have to be novel for the classical case however. Second, we employ the fact that classical-quantum channels are environment-parametrized \cite{TW2016}, as elaborated upon in Section~\ref{sec:env-param}.

\begin{prop}
Let $\mathcal{N}_{X\rightarrow B},\mathcal{M}_{X\rightarrow B}\in\mathcal{Q}(X\rightarrow B)$ be classical-quantum channels, as defined in \eqref{eq:cq1} and \eqref{eq:cq2}.
Then, for $r>0$ we have
\begin{align}
\sup_{\alpha\in (0,1)}\frac{\alpha-1}{\alpha}\left(r-\max_x D_\alpha(\nu^x_B\| \mu^x_B)\right)&\leq\underline{B}(r,\mathcal{N},\mathcal{M})\\
&\leq\overline{B}(r,\mathcal{N},\mathcal{M})\\
&\leq\sup_{\alpha \in (0,1)}\frac{\alpha-1}{\alpha}\left(r-\max_x D^\flat_\alpha(\nu^x_B\| \mu^x_B)\right).
\end{align}
\end{prop}

\begin{proof}
The lower bound follows from employing a non-adaptive strategy, in which the letter~$x$ optimizing the expression on the left-hand side is sent in to every channel use and then the Hoeffding bound for state discrimination \cite{hayashi2007error} is invoked.

The upper bound follows from reasoning similar to that given for the loose upper bound on
the Hoeffding exponent for the case of quantum states \cite[Exercise 3.15]{H06}, as well as the proof of
the Hoeffding bound in \cite[Eq.~(16)]{Hayashi09}. Fix $\delta>0$. Let
$\mathcal{T}_{X\rightarrow B}$ be a classical-quantum channel
\begin{equation}
\mathcal{T}_{X\rightarrow B}(\cdot)=\sum_{x}\langle x|\cdot|x\rangle\tau
_{B}^{x},
\end{equation}
such that for all $x$
\begin{equation}
\tau_{B}^{x}\coloneqq\underset{\tau_{B}:D(\tau_{B}\Vert\mu_{B}^{x})\leq
r-\delta}{\operatorname{argmin}}D(\tau_{B}\Vert\nu_{B}^{x}).
\end{equation}
By construction, it follows that $r>\max_{x}D(\tau_{B}^{x}\Vert\mu_{B}^{x})$.
Let $\{Q^{(n)},\mathcal{A}^{(n)}\}$ denote a sequence of channel
discrimination strategies for the classical--quantum channels $\mathcal{T}_{X\rightarrow B}$
and $\mathcal{M}_{X\rightarrow B}$, and let us denote the associated Type I
and II\ error probabilities by
\begin{equation}
\alpha_{n}^{\mathcal{T}\Vert\mathcal{M}}(\{Q^{(n)},\mathcal{A}^{(n)}\}),\qquad\beta_{n}^{\mathcal{T}\Vert\mathcal{M}}(\{Q^{(n)},\mathcal{A}^{(n)}\}).
\end{equation}
By Corollary~\ref{cor:stein-str-conv}, the strong converse of Stein's lemma for the classical--quantum channels $\mathcal{T}
_{X\rightarrow B}$ and $\mathcal{M}_{X\rightarrow B}$, if $\left\{
Q^{(n)},\mathcal{A}^{(n)}\right\}  $ is a sequence of channel discrimination
strategies for these channels such that
\begin{equation}
\limsup_{n\rightarrow\infty}-\frac{1}{n}\log\beta_{n}^{\mathcal{T}\Vert\mathcal{M}}(\{Q^{(n)},\mathcal{A}^{(n)}\})=r,
\end{equation}
then necessarily, we have that
\begin{equation}\label{eq:in-between-channel-error}
\limsup_{n\rightarrow\infty}\alpha_{n}^{\mathcal{T}\Vert\mathcal{M}}(\{Q^{(n)},\mathcal{A}^{(n)}\})=1.
\end{equation}
However, this implies that $\{I-Q^{(n)},\mathcal{A}^{(n)}\}$ can be used as a
channel discrimination strategy for the channels $\mathcal{T}_{X\rightarrow
B}$ and $\mathcal{N}_{X\rightarrow B}$, and let us denote the associated Type
I\ and II\ error probabilities by
\begin{equation}
\alpha_{n}^{\mathcal{T}\Vert\mathcal{N}}(\{I-Q^{(n)},\mathcal{A}^{(n)}\}),\qquad\beta_{n}^{\mathcal{T}\Vert\mathcal{N}}(\{I-Q^{(n)},\mathcal{A}^{(n)}\}).
\end{equation}
By applying \eqref{eq:in-between-channel-error}, we conclude that
\begin{equation}
\limsup_{n\rightarrow\infty}\alpha_{n}^{\mathcal{T}\Vert\mathcal{N}}(\{I-Q^{(n)},\mathcal{A}^{(n)}\})=0,
\end{equation}
and by again invoking Corollary~\ref{cor:stein-str-conv}, the strong converse for Stein's lemma for classical--quantum channels, it is necessary
that
\begin{align}
\limsup_{n\rightarrow\infty}-\frac{1}{n}\log\beta_{n}^{\mathcal{T}\Vert\mathcal{N}}(\{I-Q^{(n)},\mathcal{A}^{(n)}\}) &  \leq\max_{x}D(\tau_{B}^{x}\Vert\nu_{B}^{x})\\
&=\max_{x}\min_{\tau_{B}:D(\tau_{B}\Vert\mu_{B}^{x})\leq r-\delta}D(\tau_{B}\Vert\nu_{B}^{x}).
\end{align}
By this line of reasoning, i.e., chaining together the asymptotic limitations coming from Stein's lemma for classical--quantum channels, we conclude that for any sequence of channel discrimination strategies $\{Q^{(n)},\mathcal{A}^{(n)}\}$ for the classical--quantum channels $\mathcal{N}_{X\rightarrow B}$
and $\mathcal{M}_{X\rightarrow B}$ such that
\begin{equation}
\limsup_{n\rightarrow\infty}-\frac{1}{n}\log\beta_{n}^{\mathcal{N}
\Vert\mathcal{M}}(\{Q^{(n)},\mathcal{A}^{(n)}\}) = r,
\end{equation}
it is necessary that
\begin{equation}
\limsup_{n\rightarrow\infty}-\frac{1}{n}\log\alpha_{n}^{\mathcal{N}
\Vert\mathcal{M}}(\{Q^{(n)},\mathcal{A}^{(n)}\}) \leq  \max_{x}
\min_{\tau_{B}:D(\tau_{B}\Vert\nu_{B}^{x})\leq r-\delta}D(\tau_{B}
\Vert\mu_{B}^{x}).
\end{equation}
Thus, we find the following bound holding for an
arbitrary classical--quantum channel $\mathcal{T}_{X\rightarrow B}$ for which $r>\max
_{x}D(\tau_{B}^{x}\Vert\nu_{B}^{x})$:
\begin{equation}
\overline{B}(r,\mathcal{N},\mathcal{M})\leq\max_{x}\min_{\tau_{B}:D(\tau_{B}\Vert\nu_{B}^{x})\leq r-\delta}D(\tau_{B}\Vert\mu_{B}^{x}).
\end{equation}
Since $\delta>0$ is arbitrary, we can employ the facts that the
quantum relative entropy is continuous in its first argument and the
finiteness of the alphabet for the classical-quantum channels to arrive at the following bound:
\begin{equation}
\overline{B}(r,\mathcal{N},\mathcal{M})\leq\max_{x}\min_{\tau_{B}:D(\tau_{B}\Vert\mu_{B}^{x})\leq r}D(\tau_{B}\Vert\nu_{B}^{x}).
\end{equation}
By applying the divergence-sphere optimization from \eqref{eq:div-sphere-rep}, we arrive at the bound in the statement of the proposition.
\end{proof}

Any two classical-quantum channels $\mathcal{N}_{X\rightarrow B}$ and $\mathcal{M}_{X\rightarrow B}$, as defined in \eqref{eq:cq1} and \eqref{eq:cq2}, can be understood as being environment-parametrized (see Section~\ref{sec:env-param}), in the following sense:
\begin{align}
\mathcal{N}_{X\rightarrow B}(\rho) &  =\mathcal{P}\left(  \rho\otimes
\nu_{B_{1}\cdots B_{\left\vert \mathcal{X}\right\vert }}^{\text{all}}\right)
,\\
\mathcal{M}_{X\rightarrow B}(\rho) &  =\mathcal{P}\left(  \rho\otimes
\mu_{B_{1}\cdots B_{\left\vert \mathcal{X}\right\vert }}^{\text{all}}\right)
,\\
\nu_{B_{1}\cdots B_{\left\vert \mathcal{X}\right\vert }}^{\text{all}} &
\coloneqq\bigotimes\limits_{x}\nu_{B_{x}}^{x},\label{eq:nu-all}\\
\mu_{B_{1}\cdots B_{\left\vert \mathcal{X}\right\vert }}^{\text{all}} &
\coloneqq\bigotimes\limits_{x}\mu_{B}^{x},\label{eq:mu-all}
\end{align}
where $\mathcal{P}$ is a common interaction channel that measures the state $\rho$ in the basis
$\{|x\rangle\langle x|\}_{x}$, and if the outcome $x$ is obtained, it outputs
either the state $\nu_{B}^{x}$ or $\mu_{B}^{x}$ by tracing out all systems
except for the $x$th one containing this state.\ The key observation here is
that the channel $\mathcal{P}$ acts in the same way for both $\mathcal{N}
_{X\rightarrow B}$ and $\mathcal{M}_{X\rightarrow B}$, and the only way in
which these channels differ is that the environment states $\nu_{B_{1}\cdots
B_{\left\vert \mathcal{X}\right\vert }}^{\text{all}}$ and $\mu_{B_{1}\cdots
B_{\left\vert \mathcal{X}\right\vert }}^{\text{all}}$ are potentially
different. As such, and as discussed further in Section~\ref{sec:env-param}, any adaptive channel discrimination protocol for
distinguishing $n$ uses of the channel $\mathcal{N}_{X\rightarrow B}$ from the
channel $\mathcal{M}_{X\rightarrow B}$ can be understood as a parallel state
discrimination protocol acting on either
\begin{align}
\text{ the state $\left(\nu_{B_{1}\cdots B_{\left\vert \mathcal{X}\right\vert }}^{\text{all}}\right)^{\otimes n}$ or
the state $\left(\mu_{B_{1}\cdots B_{\left\vert \mathcal{X}\right\vert }}^{\text{all}}\right)^{\otimes n}$.}
\end{align}
As such, any converse bound for state discrimination applies, including that for the Hoeffding bound, leading us to
the following from a direct application of the converse Hoeffding bound for states and the identity $D_\alpha(\nu_{B_{1}\cdots B_{\left\vert \mathcal{X}\right\vert }}^{\text{all}}\| \mu_{B_{1}\cdots B_{\left\vert \mathcal{X}\right\vert }}^{\text{all}}) = \sum_{x\in\mathcal{X}}D_{s}(\nu_{B}^{x}\Vert\mu_{B}^{x})$. Applying Proposition~\ref{thm:env-param-bnds}, we find the following.

\begin{prop}
Let $\mathcal{N}_{X\rightarrow B},\mathcal{M}_{X\rightarrow B}\in
\mathcal{Q}(X\rightarrow B)$ be classical-quantum channels, as defined in \eqref{eq:cq1} and \eqref{eq:cq2}.
Then, for $r>0$ we have that
\begin{equation}
\overline{B}(r,\mathcal{N},\mathcal{M})\leq\sup_{\alpha\in(0,1)}\frac{\alpha-1}{\alpha}\left(r-\sum_{x\in\mathcal{X}}D_\alpha(\nu_{B}^{x}\Vert\mu_{B}^{x})\right).
\end{equation}
\end{prop}

The above bound can be close to the lower bound in terms of $\max_x D_\alpha(\nu_{B}^{x}\Vert\mu_{B}^{x})$ in the case that there are states $\nu^x_B$ and $\mu^x_B$ (corresponding to the same letter $x$) that are highly distinguishable, but all of the other pairs $\nu^{x'}_B$ and $\mu^{x'}_B$ for $x'\neq x$ are not so distinguishable.


\subsection{Chernoff bound}

In this section, we investigate the Chernoff bound for classical-quantum channels. In contrast to the corresponding classical results \cite{Hayashi09}, we see below that our upper bound does not match the product-state lower bound, but the difference turns out to be less than a factor of two. We suspect that this gap can be closed with an improved proof strategy. We can also apply the observation from the previous section to arrive at another upper bound that is tighter in some cases.

\begin{prop}
Let $\cN_{A\to B},\cM_{X\to B}\in\cQ(A\to B)$ be classical-quantum channels, as defined in Eqs.~\eqref{eq:cq1} and \eqref{eq:cq2}. Then, we have
\begin{align}
\max_x C(\nu^x_B\| \mu^x_B)&\leq\underline{\xi}(\cN,\cM)\\
&\leq\overline{\xi}(\cN,\cM)\\
&\leq\min\left\{\max_x\widetilde{D}_{1/2}(\nu_{B}^{x}\|\mu_{B}^{x}), \ C(\nu_{B_{1}\cdots B_{\left\vert \mathcal{X}\right\vert }}^{\operatorname{all}}\|\mu_{B_{1}\cdots B_{\left\vert \mathcal{X}\right\vert }}^{\operatorname{all}}),\ 
\max_x C^\flat_\alpha(\nu^x_B\| \mu^x_B)\right\},
\end{align}
where the states $\nu_{B_{1}\cdots B_{\left\vert \mathcal{X}\right\vert }}^{\operatorname{all}}$ and $\mu_{B_{1}\cdots
B_{\left\vert \mathcal{X}\right\vert }}^{\operatorname{all}}$ are defined in \eqref{eq:nu-all}--\eqref{eq:mu-all} and $C^\flat$ is defined in \eqref{eq:Chernoff-div-expr-flat-renyi}.
\end{prop}

\begin{proof}
The lower bound follows from sending in the letter $x$ that maximizes $\max_x C(\nu^x_B\| \mu^x_B)$ to every channel use and applying the Chernoff bound for quantum state discrimination \cite{NS09,ACMBMAV07}.
The upper bound of $\max_x\widetilde{D}_{1/2}(\nu_{B}^{x}\|\mu_{B}^{x})$ follows directly from Corollary~\ref{cor:chernoff} and the collapse of the amortized sandwiched R\'enyi divergence in Lemma \ref{lem:RenyiCollapse}. The second upper bound follows from the discussion surrounding  \eqref{eq:nu-all}--\eqref{eq:mu-all} and Proposition~\ref{thm:env-param-bnds}. The third upper bound follows from the Hoeffding exponent from the previous section, the equality in \eqref{eq:Hoeff-to-Chern}, and the same analysis as in the proof of \cite[Corollary 2]{Hayashi09}.
\end{proof}

Together with the lower bound achieved by employing a tensor-power input, non-adaptive discrimination strategy, this limits the Chernoff bound for discrimination of classical-quantum channels to
\begin{align}
C(\mathcal{N}\Vert\mathcal{M})\leq \underline{\xi}(\cN,\cM)\leq\overline{\xi}(\cN,\cM)\leq\widetilde{D}_{1/2}(\mathcal{N}\Vert\mathcal{M}).
\end{align}
Note that the upper bound is within a factor of two of the lower bound, due to the following inequalities:
\begin{equation}
\widetilde{D}_{1/2}(\mathcal{N}\Vert\mathcal{M})\leq D_{1/2}(\mathcal{N}\Vert\mathcal{M})\leq2\cdot C(\mathcal{N}\Vert\mathcal{M}).
\end{equation}
This establishes that  perfect distinguishability is possible
for classical-quantum channels if and only if there is at least one letter $x$ such that the channel outputs $\nu^x_B$ and $\mu^x_B$ can be made orthogonal in the Hilbert-Schmidt inner product. Notice that already for entanglement-breaking channels, this is not the case~\cite{Harrow10}.


\section{Other single-letter examples}\label{sec:examples}

In Section~\ref{sec:cq}, we generalized some of the classical results from \cite{Hayashi09} to the classical-quantum case, and we might hope that the same also works in the fully quantum case. One might even conjecture that the various channel divergences without amortization characterize the asymptotic error exponents of interest (as it happens classically). This is known not to hold for the Chernoff setting~\cite{Duan09,Harrow10}, but interestingly for the Stein setting, there are no counterexamples to this conjecture of which we are aware. In the following, we collect some more examples with single-letter characterisations.


\subsection{Zero-error examples}

For testing between two isometries or unitaries, it is known that there exists some finite~$n$ for which they can be perfectly distinguished. Moreover, it turns out that in this case the reference system $R$ can be chosen trivial, and entangled input states are not needed~\cite{A01,Duan07,Duan09}. For testing between two projective measurements, the exact same conclusions can be drawn~\cite{Duan06}. Hence, for those settings, all the asymptotic quantities introduced in Section~\ref{sec:discrimination} become trivial.


\subsection{Environment-parametrized channels}\label{sec:env-param}

In this section, we consider environment-parametrized channels acting as follows \cite{TW2016}:
\begin{align}
\mathcal{N}_{A\rightarrow B}(\rho_{A}) &  =\mathcal{P}_{AE\rightarrow B}
(\rho_{A}\otimes\theta_{E}^{\mathcal{N}}),\label{eq:environ-param-1}\\
\mathcal{M}_{A\rightarrow B}(\rho_{A}) &  =\mathcal{P}_{AE\rightarrow B}
(\rho_{A}\otimes\theta_{E}^{\mathcal{M}}),\label{eq:environ-param-2}
\end{align}
where $\mathcal{P}_{AE\rightarrow B}$ is a fixed channel and $\theta_{E}^{\mathcal{N}}$ and $\theta_{E}^{\mathcal{M}}$ are environment states. This notion is related  to the notion of programmable channels as considered in quantum computation~\cite{NC97,DP05} and to the notion of jointly teleportation-simulable channels as considered in channel discrimination~\cite{PhysRevLett.118.100502}. The next two propositions are related to the developments in~\cite{TW2016}, indicating that the distinguishability of these channels is limited by the distinguishability of the underlying environment states.

\begin{prop}\label{prop:environ-collapse}
Let $\cN_{A\rightarrow B},\cM_{A\rightarrow B}\in\cQ(A\rightarrow B)$ be environment-parametrized channels, as defined in \eqref{eq:environ-param-1}--\eqref{eq:environ-param-2}. Then for any generalized divergence satisfying sub-additivity, the following bound holds
\begin{equation}
\mathbf{D}^{\mathcal{A}}(\mathcal{N}\Vert\mathcal{M})\leq\mathbf{D}(\theta
_{E}^{\mathcal{N}}\Vert\theta_{E}^{\mathcal{M}}).
\end{equation}
\end{prop}

\begin{proof}
Consider fixed states $\rho_{RA}$ and $\sigma_{RA}$. We then have that
\begin{align}
&  \mathbf{D}(\mathcal{N}_{A\rightarrow B}(\rho_{RA})\Vert\mathcal{M}
_{A\rightarrow B}(\sigma_{RA}))-\mathbf{D}(\rho_{RA}\Vert\sigma_{RA}
)\nonumber\\
&  =\mathbf{D}(\mathcal{P}_{AE\rightarrow B}(\rho_{RA}\otimes\theta
_{E}^{\mathcal{N}})\Vert\mathcal{P}_{AE\rightarrow B}(\sigma_{RA}\otimes
\theta_{E}^{\mathcal{M}}))-\mathbf{D}(\rho_{RA}\Vert\sigma_{RA})\\
&  \leq\mathbf{D}(\rho_{RA}\otimes\theta_{E}^{\mathcal{N}}\Vert\sigma
_{RA}\otimes\theta_{E}^{\mathcal{M}})-\mathbf{D}(\rho_{RA}\Vert\sigma_{RA})\\
&  \leq\mathbf{D}(\rho_{RA}\Vert\sigma_{RA})+\mathbf{D}(\theta_{E}
^{\mathcal{N}}\Vert\theta_{E}^{\mathcal{M}})-\mathbf{D}(\rho_{RA}\Vert
\sigma_{RA})\\
&  =\mathbf{D}(\theta_{E}^{\mathcal{N}}\Vert\theta_{E}^{\mathcal{M}}).
\end{align}
Since the bound holds for all states $\rho_{RA}$ and $\sigma_{RA}$, this concludes the proof.
\end{proof}

\begin{prop}\label{thm:env-param-simple-bounds-div}
Let $\cN_{A\rightarrow B},\cM_{A\rightarrow B}\in\cQ(A\rightarrow B)$ be environment-parametrized channels, as defined in \eqref{eq:environ-param-1}--\eqref{eq:environ-param-2}. Then the following bounds hold for all $\varepsilon\in(0,1)$ and $\alpha>1$:
\begin{align}
\zeta_{n}(\varepsilon,\cN,\cM)  & \leq\frac{1}{1-\varepsilon}\left(D(\theta_{E}^{\mathcal{N}}\Vert\theta_{E}^{\mathcal{M}})+\frac{h_{2}(\varepsilon)}{n}\right),\\
\zeta_{n}(\varepsilon,\cN,\cM)  & \leq\widetilde{D}_{\alpha}(\theta_{E}^{\mathcal{N}}\Vert\theta_{E}^{\mathcal{M}})+\frac{\alpha}{n(\alpha-1)}\log\!\left(\frac{1}{1-\varepsilon}\right).
\end{align}
As a consequence, the following asymptotic bound holds for all $\eps \in (0,1)$
\begin{equation}
\overline{\zeta}_{n}(\varepsilon,\cN,\cM)\leq D(\theta_{E}^{\mathcal{N}}\Vert\theta_{E}^{\mathcal{M}}).
\end{equation}
\end{prop}

\begin{proof}
This follows from the meta-converse in Lemma~\ref{ThmAmortizedBound} together with the general bounds in Proposition~\ref{weakConverse} and \cite[Lemma~5]{Cooney2016}, as well as the amortization collapse from Proposition~\ref{prop:environ-collapse}.
\end{proof}

We have taken the perspective that the converse in the above proposition arises from the bound on amortized channel
divergence in Proposition~\ref{prop:environ-collapse}. However, there is a more fundamental reason why the bounds in Proposition \ref{thm:env-param-simple-bounds-div} hold~\cite{DM14}.
\begin{figure}[ptb]
\begin{center}
\includegraphics[width=5.5in]
{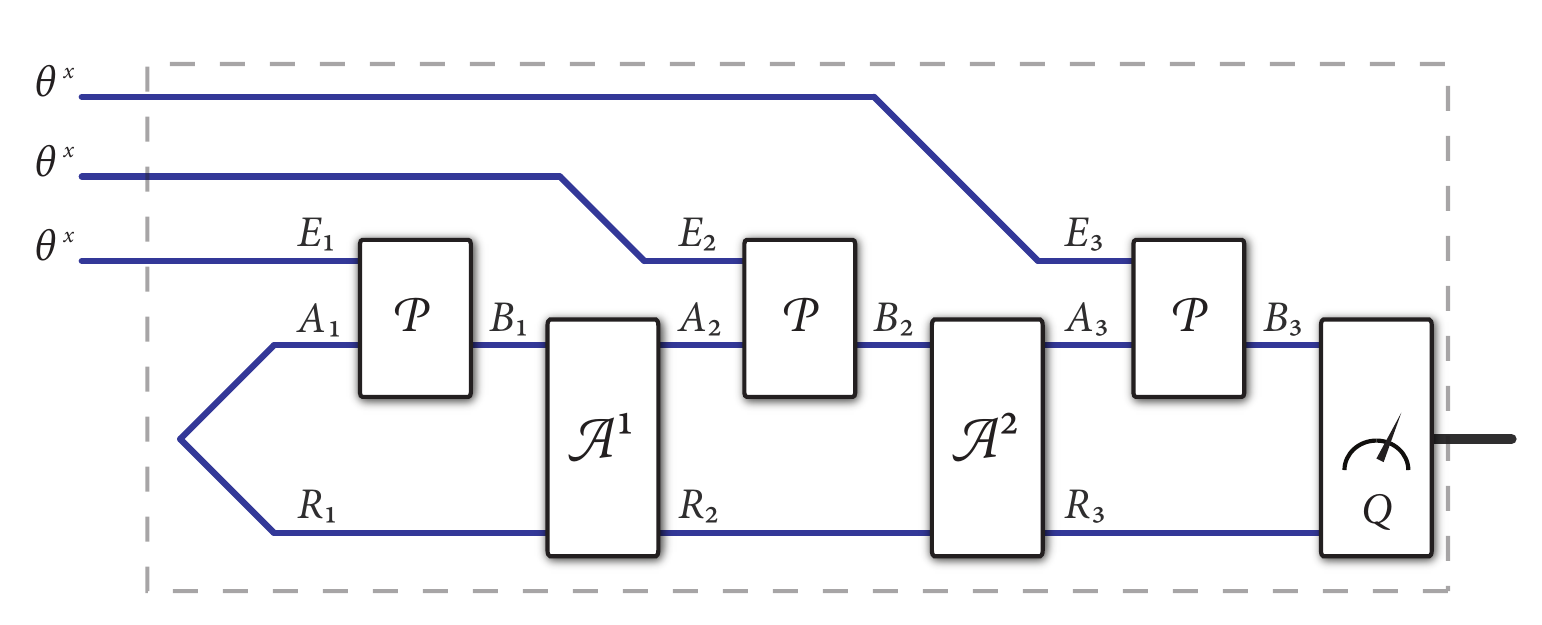}
\end{center}
\caption{The figure depicts how a general protocol for channel discrimination of two environment-parametrized channels  $\mathcal{N}$ and $\mathcal{M}$ can be understood as a quantum state discrimination protocol for the environment states $\theta^0 \equiv \theta_{E}^{\mathcal{N}}$ and $\theta^1 \equiv \theta_{E}^{\mathcal{M}}$, respectively. That is, the operations depicted in the dotted box form a measurement on the state $(\theta^x)^{\otimes n}$ for $n=3$. We emphasise the similarity of this figure with \cite[Figure 2]{DM14}.}
\label{fig:environ-param-disc}
\end{figure}
Due to the structure of environment-parametrized channels $\mathcal{N}$ and $\mathcal{M}$, any $n$-round, adaptive channel discrimination protocol for them, as depicted in Figure~\ref{fig:adaptive-prot}, can be understood as a particular kind of state discrimination protocol for the states $(\theta_{E}^{\mathcal{N}})^{\otimes n}$ and $(\theta_{E}^{\mathcal{M}})^{\otimes n}$. Figure~\ref{fig:environ-param-disc} provides a visual aid to understand this observation. As such, in this case, the type~I and II error probabilities, the fundamental quantities involved in the analysis of hypothesis testing of channels, can be rewritten as follows for any channel discrimination strategy $\left\{Q,\mathcal{A}\right\}$:
\begin{equation}
\alpha_{n}(\left\{  Q,\mathcal{A}\right\}  )=\operatorname{Tr}[(I-\Lambda
_{E^{n}})(\theta_{E}^{\mathcal{N}})^{\otimes n}],\qquad\beta_{n}(\left\{
Q,\mathcal{A}\right\}  )=\operatorname{Tr}[\Lambda_{E^{n}}(\theta
_{E}^{\mathcal{M}})^{\otimes n}],
\end{equation}
where $\{\Lambda_{E^{n}},I_{E^{n}}-\Lambda_{E^{n}}\}$ is a quantum measurement
that depends on the channel discrimination strategy $\left\{  Q,\mathcal{A}
\right\}  $, as well as the interaction channel $\mathcal{P}_{AE\rightarrow
B}$ in the definition in \eqref{eq:environ-param-1}--\eqref{eq:environ-param-2}. Now, since the various
channel discrimination exponents of interest from Section~\ref{sec:discrimination} involve an optimization over all channel discrimination strategies, if we allow for a further optimization over the
interaction channel $\mathcal{P}_{AE\rightarrow B}$, then the bounds loosen,
but we can use them to characterize the various exponents of interest by applying known results from the asymptotic theory of quantum state discrimination \cite{HP91,ON00,ACMBMAV07,NS09,Nagaoka06,hayashi2007error,Mosonyi2015}. We summarize this observation as the following proposition.

\begin{prop}\label{thm:env-param-bnds}
Let $\cN_{A\rightarrow B},\cM_{A\rightarrow B}\in\cQ(A\rightarrow B)$ be
environment-parametrized channels, as defined in
\eqref{eq:environ-param-1}--\eqref{eq:environ-param-2}. Fix $\varepsilon
\in(0,1)$, $r>0$, and $p\in(0,1)$. Then the following bounds hold for all $n \geq 1$:
\begin{align}
\zeta_{n}(\varepsilon,\mathcal{N},\mathcal{M})  & \leq\zeta_{n}(\varepsilon
,\theta_{E}^{\mathcal{N}},\theta_{E}^{\mathcal{M}}),\\
H_n(r,\mathcal{N},\mathcal{M})  & \geq H_{n}(r,\theta
_{E}^{\mathcal{N}},\theta_{E}^{\mathcal{M}}),\\
B_{n}(r,\mathcal{N},\mathcal{M})  & \leq B_{n}(r,\theta_{E}^{\mathcal{N}
},\theta_{E}^{\mathcal{M}}),\\
\xi_{n}(p,\mathcal{N},\mathcal{M})  & \leq\xi_{n}(p,\theta_{E}^{\mathcal{N}
},\theta_{E}^{\mathcal{M}}).
\end{align}
\end{prop}

Now, let us consider an interesting case in which all of the above upper bounds are achieved.

\begin{definition}
[Environment-seizable channels]
\label{def:seizable}
Environment-parametrized channels $\mathcal{N}$
and $\mathcal{M}$ are environment-seizable if there exists a common input
state $\rho_{RA}$\ and post-processing channel $\mathcal{D}_{RB\rightarrow E}$
that can be applied to both environment-parametrized channels $\mathcal{N}$
and $\mathcal{M}$, which allows for seizing the state of the environment
at the output of the channel:
\begin{align}
\mathcal{D}_{RB\rightarrow E}(\mathcal{N}_{A\rightarrow B}(\rho_{RA})) &
=\theta_{E}^{\mathcal{N}},\\
\mathcal{D}_{RB\rightarrow E}(\mathcal{M}_{A\rightarrow B}(\rho_{RA})) &
=\theta_{E}^{\mathcal{M}}.
\end{align}
\end{definition}

As non-trivial examples of environment-seizable channels, let us consider a pair of erasure channels and a pair of  dephasing channels. An important objective from an experimental point of view is to  devise the simplest possible method for seizing the underlying environment state $\theta_{E}^{\mathcal{N}}$ or $\theta_{E}^{\mathcal{M}}$ in the case that the channels are seizable. A quantum erasure channel is defined as \cite{GBP97}
\begin{equation}
\mathcal{E}^{p}(\rho)\coloneqq (1-p)\rho+p|e\rangle\langle e|,
\end{equation}
where $\rho$ is a $d$-dimensional input state, $p\in\left[  0,1\right]  $ is
the erasure probability, and $|e\rangle\langle e|$ is a pure erasure state
orthogonal to any input state, so that the output state has $d+1$ dimensions. To see that any two erasure channels are environment-parametrized, set the initial environment state to be
\begin{equation}
\theta^p_E \coloneqq (1-p)| 0\rangle \langle 0|_E + p | 1\rangle \langle 1|_E.
\end{equation}
Then the common interaction $\mathcal{P}$ consists of adjoining the erasure symbol $|e\rangle \langle e |_{E'}$, applying a controlled-SWAP to the channel input and the register $E'$, controlled on the register $E$, and finally discarding the registers $E$ and $E'$. The simplest way of seizing the environment state $\theta^p_E$ is to input the state  $\vert 0\rangle \langle 0\vert$, leading to $(1-p)| 0\rangle \langle 0|_E + p | e\rangle \langle e|_E$, and then to perform the classical transformation $ \vert 0\rangle \langle 0\vert \to \vert 0\rangle \langle 0\vert$ and $\vert e\rangle \langle e\vert \to \vert 1\rangle \langle 1\vert$.

A $d$-dimensional dephasing channel has the following action:
\begin{equation}
\mathcal{D}^{\mathbf{p}}(\rho)=\sum_{i=0}^{d-1}p_{i}Z^{i}\rho Z^{i\dag},
\end{equation}
where $\mathbf{p}$ is a vector containing the probabilities $p_{i}$ and $Z$
has the following action on the computational basis $Z|x\rangle=e^{2\pi
ix/d}|x\rangle$. To see that any two dephasing channels as above are environment-parametrized, set the initial environment state to be $\theta^{\mathbf{p}}_E \coloneqq \sum_{i=0}^{d-1} p_i \vert i \rangle \langle i \vert_E$.
Then the common interaction consists of applying the controlled unitary $\sum_{i=0}^{d-1} Z^i_A \otimes |i\rangle \langle i |_E$ to the channel input system $A$ and the environment $E$, and then tracing out $E$. 
The simplest way of seizing the environment state is to input the state $\vert \phi\rangle_A$ of maximal coherence, where
\begin{equation}
\vert \phi\rangle_A \coloneqq \frac{1}{\sqrt{d}}\sum_{i=0}^{d-1} \vert i \rangle_A.
\end{equation}
After doing so, the channel output is $\sum_{i=1}^{d}p_{i}Z^{i}\vert \phi \rangle \langle \phi \vert_A Z^{i\dag}$. Applying a Fourier transform to this state gives the underlying environment state $\theta^{\mathbf{p}}_E$.

In the case that two environment-parametrized channels $\mathcal{N}$ and
$\mathcal{M}$ are environment-seizable, the most sensible strategy for channel
discrimination is to first apply the environment seizing procedure highlighted
in Definition~\ref{def:seizable}, in order to seize the environment states $\theta_{E}^{\mathcal{N}}$ or
$\theta_{E}^{\mathcal{M}}$, and then follow with the best state discrimination
protocol for $\theta_{E}^{\mathcal{N}}$ or $\theta_{E}^{\mathcal{M}}$. As a
result of this observation, we conclude the following theorem.

\begin{theorem}\label{thm:seizable-tight}
Let $\cN_{A\rightarrow B},\cM_{A\rightarrow B}\in\cQ(A\rightarrow B)$ be environment-parametrized channels that are also environment seizable. Fix $\varepsilon\in(0,1)$, $r>0$, and $p\in(0,1)$. Then, for all $n\geq 1$, we have that
\begin{align}
\zeta_{n}(\varepsilon,\mathcal{N},\mathcal{M})  & =\zeta_{n}(\varepsilon
,\theta_{E}^{\mathcal{N}},\theta_{E}^{\mathcal{M}}),\\
H_n(r,\mathcal{N},\mathcal{M})  & = H_{n}(r,\theta
_{E}^{\mathcal{N}},\theta_{E}^{\mathcal{M}}),\\
B_{n}(r,\mathcal{N},\mathcal{M})  & = B_{n}(r,\theta_{E}^{\mathcal{N}
},\theta_{E}^{\mathcal{M}}),\\
\xi_{n}(p,\mathcal{N},\mathcal{M})  & =\xi_{n}(p,\theta_{E}^{\mathcal{N}
},\theta_{E}^{\mathcal{M}}).
\end{align}
so that
\begin{align}
\zeta(\varepsilon,\mathcal{N},\mathcal{M})&=D(\theta_{E}^{\mathcal{N}}\Vert\theta_{E}^{\mathcal{M}}),\\
H(r,\mathcal{N},\mathcal{M})&=\sup_{\alpha>1}\frac{\alpha-1}{\alpha}\left(r-\widetilde{D}_\alpha(\theta_{E}^{\mathcal{N}}\|\theta_{E}^{\mathcal{M}})\right),\\
B(r,\mathcal{N},\mathcal{M})&=\sup_{s\in(0,1)}\frac{\alpha-1}{\alpha}\left(r-D_\alpha(\theta_{E}^{\mathcal{N}}\|\theta_{E}^{\mathcal{M}})\right),\\
\xi(\mathcal{N},\mathcal{M})&=C\left(\theta_{E}^{\mathcal{N}}\middle\|\theta_{E}^{\mathcal{M}}\right).
\end{align}
\end{theorem}

Another case of environment-parametrized channels that are also environment seizable are those that are jointly covariant or jointly teleportation-simulable with the Choi state as the associated resource state. Theorem~\ref{thm:seizable-tight} then recovers the previous results about adaptive quantum channel discrimination from~\cite{PhysRevLett.118.100502,TW2016}. Recall that two channels are jointly teleportation-simulable with associated resource states $\omega_{RB^{\prime}}^{\mathcal{N}}$ and $\omega_{RB^{\prime}}^{\mathcal{M}}$ \cite{PL17,TW2016}\ if there exists an LOCC channel $\mathcal{L}_{ARB^{\prime}\rightarrow B}$ such that for all input states $\rho_{A}$
\begin{align}
\mathcal{N}_{A\rightarrow B}(\rho_{A}) &  =\mathcal{L}_{ARB^{\prime
}\rightarrow B}(\rho_{A}\otimes\omega_{RB^{\prime}}^{\mathcal{N}}),\\
\mathcal{M}_{A\rightarrow B}(\rho_{A}) &  =\mathcal{L}_{ARB^{\prime
}\rightarrow B}(\rho_{A}\otimes\omega_{RB^{\prime}}^{\mathcal{M}}).
\end{align}
Jointly teleportation-simulable channels are then a special case of environment-parametrized channels. If the resource states $\omega_{RB^{\prime}}^{\mathcal{N}}$ and $\omega_{RB^{\prime}}^{\mathcal{M}}$ are equal to the respective Choi states of the channels, then these channels are also environment-seizable:\ One would just input the maximally entangled state and then recover the Choi state of the channel. Two channels are jointly covariant if they are jointly covariant with respect to a group $\{U_{A}^{g}\}_{g}$ that forms a one-design:\ $\frac{1}{\left\vert
\mathcal{G}\right\vert }\sum_{g\in\mathcal{G}}U_{A}^{g}(X)U_{A}^{g\dag}=\operatorname{Tr}[X]I/\left\vert A\right\vert $. Following the methods from~\cite[Section 7]{CDP09}, any jointly covariant channels are then jointly teleportation-simulable with Choi states as the associated resource states.

Going further, we can extend the definitions above for environment-parametrized and environment-seizable channels to the case of multiple channels, as was done for environment-parametrized channels in \cite{DW17}. In the case of multiple environment-seizable channels, generalizing the above, channel discrimination problems reduce to state discrimination problems. As such, the recent result from the multiple Chernoff bound \cite{li2016} immediately applies to this setting (but we refrain from stating any details here).


\subsection{Replacer channels and quantum illumination}

Another interesting scenario for channel discrimination occurs when the null hypothesis is an arbitrary channel $\cN$ and the alternative hypothesis is a replacer channel $\cR$, defined as $\cR(X) = \tr[X] \tau$ for some state $\tau$. In \cite{Cooney2016}, the Stein's lemma and strong converse exponent for this setting were identified as follows for $\eps \in (0,1)$ and $r>D(\cN\Vert \cR)$, respectively:
\begin{align}
\zeta(\varepsilon,\mathcal{N},\mathcal{R})&= D(\cN\Vert \cR),\\
H(r,\mathcal{N},\mathcal{R})  & =\sup_{\alpha>1}\frac{\alpha-1}{\alpha}\left(r-\widetilde{D}_\alpha(\cN\Vert \cR) \right).
\label{eq:sc-replacer}
\end{align}
In this section, we prove that the energy-constrained quantum relative entropy of an arbitrary channel and a replacer channel does not increase under amortization. After that, we revisit the setting of \cite{Cooney2016} and establish the weak Stein's lemma for energy-constrained channel discrimination in this setting, demonstrating that adaptive strategies do not help. We note that this result solves an open question stated at the end of \cite{PhysRevA.98.012101}, having to do with the theory of quantum illumination \cite{L08}. After that, we recast one of the results of~\cite{Cooney2016} in terms of amortized channel divergence: that is, we prove that the sandwiched R\'enyi channel divergence of an arbitrary channel and a replacer channel does not increase under amortization.

\begin{lemma}\label{lem:amort-collapse-replacer-en-const}
Let $\cN_{A\to B}\in\cQ(A\to B)$ and let $\cR_{A\to B}\in\cQ(A\to B)$ be a replacer channel $\cR(\cdot) = \tr(\cdot)\,\tau_B$ where $\tau_B \in \cS(B)$. Let $H_A$ be a Hamiltonian and $E \in [0,\infty)$ an energy constraint. Then an amortization collapse occurs for the energy-constrained quantum relative entropy of the channels $\cN$ and $\cR$:
\begin{align}
D^{\mathcal{A}}_{H,E}(\cN\|\cR)=D_{H,E}(\cN\|\cR) \coloneqq\sup_{\psi_{RA}:\tr[H_A\psi_A]\leq E}D(\cN_{A\to B}(\psi_{RA})\|\cR_{A\to B}(\psi_{RA})).
\end{align}
\end{lemma}

\begin{proof}
We first show the $\leq$ direction. 
Let $\rho_{RA}$ and $\sigma_{RA}$ be arbitrary states satisfying the energy constraints $\tr[H_A \rho_A],\tr[H_A \sigma_A] \leq E$. We find that
\begin{align}
& D(\cN_{A\rightarrow B}(\rho_{RA}) \| \sigma_R\otimes \tau_B) - D(\rho_{RA}\|\sigma_{RA}) \notag \\
&\leq D( \cN_{A\rightarrow B}(\rho_{RA}) \| \sigma_R\otimes \tau_B) - D(\rho_{R}\|\sigma_{R}) \\
&=  -S(\cN_{A\rightarrow B}(\rho_{RA})) - \tr[\cN_{A\rightarrow B}(\rho_{RA})\log\sigma_R\otimes \tau_B] - D(\rho_{R}\|\sigma_{R})  \\
&=  -S(\cN_{A\rightarrow B}(\rho_{RA})) - \tr[\cN_{A\rightarrow B}(\rho_{A})\log\tau_B] -\tr[\rho_R\log\sigma_R] - D(\rho_{R}\|\sigma_{R}) \\
&=  -S(\cN_{A\rightarrow B}(\rho_{RA})) - \tr[\cN_{A\rightarrow B}(\rho_{A})\log\tau_B] -\tr[\rho_R\log\rho_R] \\
&=  -S(\cN_{A\rightarrow B}(\rho_{RA})) - \tr[\cN_{A\rightarrow B}(\rho_{RA})\log\rho_R\otimes\tau_B] \\
&=  D( \cN_{A\rightarrow B}(\rho_{RA}) \| \rho_R\otimes\tau_B) \\
& \leq D_{H,E}(\cN\|\cM)
,
\end{align}
where the first inequality follows from the data-processing inequality and the second because the state $\rho_{RA}$ is a particular state satisfying the energy constraint. Since the inequality holds for all states $\rho_{RA}$ and $\sigma_{RA}$ satisfying the energy constraints, we conclude that
\begin{equation}
D^{\mathcal{A}}_{H,E}(\cN\|\cM) \leq 
D_{H,E}(\cN\|\cM).
\end{equation}
Combining with the fact that 
$D^{\mathcal{A}}_{H,E}(\cN\|\cM) \geq 
D_{H,E}(\cN\|\cM)$ for any two channels, which follows from Proposition~\ref{Lem:AmortizedIneq} and the fact that the quantum relative entropy is faithful, we conclude the statement of the lemma.
\end{proof}

As a direct consequence of \eqref{eq:weak-conv-stein-constrained} in Proposition~\ref{weakConverse} and Lemma~\ref{lem:amort-collapse-replacer-en-const}, we conclude the following theorem.

\begin{theorem}\label{thm:stein-energy-constrained}
Let $\cN_{A\to B}\in\cQ(A\to B)$ and let $\cR_{A\to B}\in\cQ(A\to B)$ be a replacer channel $\cR(\cdot) = \tr(\cdot)\,\tau_B$ where $\tau_B \in \cS(B)$. Let $H_A$ be a Hamiltonian and $E \in [0,\infty)$ an energy constraint. 
Then for $n\in\mathbb{N}$ and $\eps\in(0,1]$, the following bound holds
\begin{align}
\zeta_n(\eps,\cN,\cR,H,E)\leq \frac{1}{1-\eps}\left(D_{H,E}(\cN\|\cR)+\frac{h_2(\varepsilon)}{n}\right),
\end{align}
implying that
\begin{equation}
\zeta(\cN,\cR,H,E)\coloneqq\lim_{\varepsilon \to 0} \lim_{n \to \infty}\zeta_n(\eps,\cN,\cR,H,E)= 
D_{H,E}(\cN\|\cR).
\end{equation}
\end{theorem}

\begin{remark}
In a setting related to quantum illumination \cite{L08}, one considers the null hypothesis to be that the channel $\cN$ is a thermal bosonic channel and the other is a replacer channel $\cR$ that replaces with a thermal state (see the discussion in \cite[Section~2.3]{Cooney2016}). Recently, it was shown in \cite{PhysRevA.98.012101} that when the Hamiltonian is the photon number operator, the input that optimizes the energy-constrained, quantum relative entropy channel divergence $D_{H,E}(\cN\|\cR)$ is the two-mode squeezed vacuum state saturating the energy constraint. Combining this result with Theorem~\ref{thm:stein-energy-constrained}, we then have a complete characterization of the Stein exponent in this setting. For further developments along these lines, i.e., second-order characterizations of these exponents, see \cite{PhysRevLett.119.120501}.
\end{remark}

Finally, we recast one of the main results of \cite{Cooney2016} as a statement about the amortized sandwiched R\'enyi divergence of an arbitrary channel and a replacer channel.

\begin{prop}
Given an arbitrary channel $\mathcal{N}$ and a replacer channel $\mathcal{R}$,
the following amortization collapse holds for $\alpha>1$:
\begin{equation}
\widetilde{D}_{\alpha}^{\mathcal{A}}(\mathcal{N}\Vert\mathcal{R}
)=\widetilde{D}_{\alpha}(\mathcal{N}\Vert\mathcal{R}).
\end{equation}
\end{prop}

\begin{proof}
We always have that $\widetilde{D}_{\alpha}^{\mathcal{A}}(\mathcal{N}
\Vert\mathcal{R})\geq\widetilde{D}_{\alpha}(\mathcal{N}\Vert\mathcal{R})$, due
to Proposition~\ref{Lem:AmortizedIneq} and the fact that the sandwiched R\'enyi relative entropy is
faithful. So we now prove the opposite inequality. Let $\rho_{RA}$ and
$\sigma_{RA}$ be arbitrary states.  Consider that
\begin{align}
& \widetilde{D}_{\alpha}(\mathcal{N}_{A\rightarrow B}(\rho_{RA})\Vert
\sigma_{R}\otimes\tau_{B})-\widetilde{D}_{\alpha}(\rho_{RA}\Vert\sigma
_{RA})\notag \\
& \leq\widetilde{D}_{\alpha}(\mathcal{N}_{A\rightarrow B}(\rho_{RA}
)\Vert\sigma_{R}\otimes\tau_{B})-\widetilde{D}_{\alpha}(\rho_{R}\Vert
\sigma_{R})\\
& =\frac{1}{\alpha-1}\log\operatorname{Tr}\left[  \left(  \left(\sigma
_{R}^{\left(  1-\alpha\right)  /2\alpha}\otimes\tau_{B}^{\left(
1-\alpha\right)  /2\alpha}\right)  \mathcal{N}_{A\rightarrow B}(\rho
_{RA})\left(\sigma_{R}^{\left(  1-\alpha\right)  /2\alpha}\otimes\tau
_{B}^{\left(  1-\alpha\right)  /2\alpha}\right)  \right)  ^{\alpha}\right]
\notag \\
&\qquad -\frac{1}{\alpha-1}\log\operatorname{Tr}\left[  \left(  \sigma_{R}^{\left(
1-\alpha\right)  /2\alpha}\rho_{R}\sigma_{R}^{\left(  1-\alpha\right)
/2\alpha}\right)  ^{\alpha}\right]  \\
& =\frac{1}{\alpha-1}\log\frac{\operatorname{Tr}\left[  \left(
\tau_{B}^{\left(  1-\alpha\right)  /2\alpha}\mathcal{N}_{A\rightarrow
B}\left(  \sigma_{R}^{\left(  1-\alpha\right)  /2\alpha}\rho_{RA}\sigma
_{R}^{\left(  1-\alpha\right)  /2\alpha}\right)  \tau_{B}^{\left(
1-\alpha\right)  /2\alpha}\right)  ^{\alpha}\right]  }{\operatorname{Tr}
\left[  \left(  \sigma_{R}^{\left(  1-\alpha\right)  /2\alpha}\rho_{R}
\sigma_{R}^{\left(  1-\alpha\right)  /2\alpha}\right)  ^{\alpha}\right]
}.
\end{align}
Now defining
\begin{align}
X_{RA}   \coloneqq\frac{\sigma_{R}^{\left(  1-\alpha\right)  /2\alpha}\rho
_{RA}\sigma_{R}^{\left(  1-\alpha\right)  /2\alpha}}{\left\Vert \sigma
_{R}^{\left(  1-\alpha\right)  /2\alpha}\rho_{R}\sigma_{R}^{\left(
1-\alpha\right)  /2\alpha}\right\Vert _{\alpha}},
\qquad\quad
\Theta_{\omega}(\cdot)   \coloneqq\omega^{1/2}(\cdot)\omega^{1/2},
\end{align}
we see from above that
\begin{align}
& \widetilde{D}_{\alpha}(\mathcal{N}_{A\rightarrow B}(\rho_{RA})\Vert
\sigma_{R}\otimes\tau_{B})-\widetilde{D}_{\alpha}(\rho_{R}\Vert\sigma_{R})
\notag\\
& =\frac{1}{\alpha-1}\log\operatorname{Tr}\left[  \left(  \tau
_{B}^{\left(  1-\alpha\right)  /2\alpha}\mathcal{N}_{A\rightarrow B}\left(
X_{RA}\right)  \tau_{B}^{\left(  1-\alpha\right)  /2\alpha}\right)  ^{\alpha
}\right]\\
& =\frac{\alpha}{\alpha-1}\log\operatorname{Tr}\left\Vert (\Theta_{\tau
_{B}^{\left(  1-\alpha\right)  /\alpha}}\circ\mathcal{N}_{A\rightarrow
B})\left(  X_{RA}\right)  \right\Vert _{\alpha}\\
& \leq
\sup_{X_{RA}\geq 0 , \Vert X_R \Vert_\alpha \leq 1} 
\frac{\alpha}{\alpha-1}\log\operatorname{Tr}\left\Vert (\Theta_{\tau
_{B}^{\left(  1-\alpha\right)  /\alpha}}\circ\mathcal{N}_{A\rightarrow
B})\left(  X_{RA}\right)  \right\Vert _{\alpha}\\
&= \widetilde{D}_{\alpha}(\mathcal{N}\Vert\mathcal{R}),
\end{align}
with the last (non-trivial)\ step following from \cite{GW13}, which in turn built upon \cite[Theorem 10]{DJKR06} and \cite[Section~3]{J06} (see also \cite[Appendix~A]{Cooney2016} in this context).
\end{proof}

Clearly, we can use the amortization collapse above and Proposition \ref{thm:converseexponent} to conclude the $\geq$~inequality in \eqref{eq:sc-replacer}.


\section{Conclusion \& Outlook}\label{sec:conclusion}

In order to derive upper bounds on the power of adaptive quantum channel discrimination protocols, we introduced a framework based on the concept of amortized channel divergence. This led to various converse bounds for general quantum channel discrimination, and as our main result, we established the strong Stein's lemma for classical-quantum channels by showing that asymptotically the exponential error rate for classical-quantum channel discrimination is not improved by adaptive strategies.

We regard our work as an initial step towards a plethora of open questions surrounding quantum channel discrimination. For example, with regards to classical-quantum channels, we are still missing tight characterisations in the Chernoff and Hoeffding settings\,---\,which hold in the classical case~\cite{Hayashi09}. The same questions also remain open for the setting involving replacer channels, which is strongly connected to similar open questions about quantum-feedback-assisted communication \cite{Cooney2016}. Even more fundamentally, we left open the question of whether adaptive protocols improve the exponential error rate for quantum channel discrimination in the asymmetric Stein setting (as they do in the symmetric Chernoff setting). We suspect that this is the case. A first step in this direction would be to look at the intermediate (parallel) setting as discussed in Remark~\ref{rmk:block-coding}, in which a state $\gamma_{RA^n}$ is prepared, either the tensor-power channel $(\cN_{A \to B})^{\otimes n}$ or $(\cM_{A \to B})^{\otimes n}$ is applied, and then a joint measurement is performed on the systems~$RB^n$. We emphasise that it is not even known whether this setting  offers an asymptotic advantage compared to a tensor-power strategy with input $\gamma_{RA}^{\otimes n}$. The question might be thought of as determining if the following limit holds
\begin{align}
\frac{1}{n}D\left(\mathcal{N}^{\otimes n}\middle\|\mathcal{M}^{\otimes n}\right)\overset{?}{\to} D\left(\mathcal{N}\middle\|\mathcal{M}\right).
\end{align}
Now, note that if we restrict the quantum memory system $R$ to be one-dimensional (trivial), then the Hastings counterexamples to the minimal output entropy conjecture~\cite{Hastings:2009aa}, applied to the  setting involving a replacer channel, immediately give a separation to the tensor-product strategy. This suggests that for a non-trivial quantum memory $R$, there are some deep entropic additivity questions that remain to be explored. Finally, quantum channel discrimination is strongly connected to many other fundamental tasks in quantum information theory, and we expect plentiful applications of our framework to be found.


\section*{Acknowledgements}

We are grateful to Fernando Brand\~ao, Gilad Gour, Milan Mosonyi, Giacomo de Palma, and Andreas Winter for discussions related to the topic of this paper.
The authors would like to thank the Isaac Newton Institute for Mathematical Sciences for support and hospitality during the programme ``Beyond i.i.d.~in information theory,''  which was supported by EPSRC grant number EP/R014604/1.
CH acknowledges support from Spanish MINECO, project FIS2016-80681-P with the support of AEI/FEDER funds and FPI Grant No. BES-2014-068888, as well as by the Generalitat de Catalunya, project CIRIT 2017-SGR-1127. EK acknowledges support from the Office of Naval Research. MMW acknowledges support from the National Science Foundation under grant no.~1907615.
He is also grateful to MB for hosting him for research discussions at Imperial College London during May 2018.

\section*{Conflict of interest statement}

On behalf of all authors, the corresponding author states that there is no conflict of interest.


\appendix

\section{Amortization does not increase the Hilbert $\alpha$-channel divergences}\label{app:hilbert-alpha-div}

In this appendix, we prove that the Hilbert $\alpha$-divergence from \cite[Section~III]{BG17} obeys a data-processed triangle inequality, and as a consequence, channel divergences based on it do not increase under amortization. We also remark how other metrics based on quantum fidelity obey a data-processed triangle inequality, and so their corresponding channel divergences do not increase under amortization.

The Hilbert $\alpha$-divergence of states $\rho$ and $\sigma$ is defined for
$\alpha\geq1$ as \cite[Section III]{BG17}
\begin{align}
H_{\alpha}(\rho\Vert\sigma)  & \coloneqq\frac{\alpha}{\alpha-1}\log\sup{}_{\alpha
}(\rho/\sigma)\\
\sup{}_{\alpha}(\rho/\sigma)  & \coloneqq\sup_{\alpha^{-1}I\leq\Lambda\leq I}
\frac{\operatorname{Tr}[\Lambda\rho]}{\operatorname{Tr}[\Lambda\sigma]}.
\end{align}
It is known that \cite[Theorem 1]{BG17}
\begin{align}
\lim_{\alpha\rightarrow1}H_{\alpha}(\rho\Vert\sigma)  & =\frac{1}{2\ln
2}\left\Vert \rho-\sigma\right\Vert _{1},\\
\lim_{\alpha\rightarrow\infty}H_{\alpha}(\rho\Vert\sigma)  & =D_{\max}
(\rho\Vert\sigma).
\end{align}
Here we prove that this quantity obeys a data-processed triangle inequality
for all $\alpha\geq1$.

\begin{lemma}
[Data-processed triangle inequality]Let $\mathcal{P}_{A\rightarrow B}$ be a
positive trace-preserving map, and let $\rho_{A},\sigma_{A}\in S(A)$ and
$\omega_{B}\in S(B)$. Then the following inequality holds for all $\alpha
\geq1$:
\begin{equation}
H_{\alpha}(\mathcal{P}_{A\rightarrow B}(\rho_{A})\Vert\omega_{B})\leq
H_{\alpha}(\rho_{A}\Vert\sigma_{A})+H_{\alpha}(\mathcal{P}_{A\rightarrow
B}(\sigma_{A})\Vert\omega_{B}).
\end{equation}
\end{lemma}

\begin{proof}
For $\alpha=1$, we have that $\lim_{\alpha\rightarrow1}H_{\alpha}(\rho
\Vert\sigma)=\frac{1}{2\ln2}\left\Vert \rho-\sigma\right\Vert _{1}$, as
recalled above. The statement then follows from the usual triangle inequality:
\begin{align}
\left\Vert \mathcal{P}_{A\rightarrow B}(\rho_{A})-\omega_{B}\right\Vert _{1}
& \leq\left\Vert \mathcal{P}_{A\rightarrow B}(\rho_{A})-\mathcal{P}
_{A\rightarrow B}(\sigma_{A})\right\Vert _{1}+\left\Vert \mathcal{P}
_{A\rightarrow B}(\sigma_{A})-\omega_{B}\right\Vert _{1} \label{eq:dp-tri-trace-1}\\
& \leq\left\Vert \rho_{A}-\sigma_{A}\right\Vert _{1}+\left\Vert \mathcal{P}
_{A\rightarrow B}(\sigma_{A})-\omega_{B}\right\Vert _{1},
\label{eq:dp-tri-trace-2}
\end{align}
and the fact that trace distance is monotone with respect to positive,
trace-preserving maps.

To prove the inequality for $\alpha>1$, let $\Lambda_{B}$ be an arbitrary operator such that
$\alpha^{-1}I_{B}\leq\Lambda_{B}\leq I_{B}$. The map $\mathcal{P}
_{A\rightarrow B}^{\dag}$ is positive and unital because $\mathcal{P}
_{A\rightarrow B}$ is positive and trace preserving by assumption. Then $\alpha^{-1}
I_{A}\leq\mathcal{P}_{A\rightarrow B}^{\dag}(\Lambda_{B})\leq I_{A}$ and
\begin{align}
\frac{\operatorname{Tr}[\Lambda_{B}\mathcal{P}_{A\rightarrow B}(\rho_{A}
)]}{\operatorname{Tr}[\Lambda_{B}\omega_{B}]}  & =\frac{\operatorname{Tr}
[\mathcal{P}_{A\rightarrow B}^{\dag}(\Lambda_{B})(\rho_{A})]}
{\operatorname{Tr}[\Lambda_{B}\omega_{B}]}\\
& =\frac{\operatorname{Tr}[\mathcal{P}_{A\rightarrow B}^{\dag}(\Lambda
_{B})(\rho_{A})]}{\operatorname{Tr}[\mathcal{P}_{A\rightarrow B}^{\dag
}(\Lambda_{B})(\sigma_{A})]}\frac{\operatorname{Tr}[\mathcal{P}_{A\rightarrow
B}^{\dag}(\Lambda_{B})(\sigma_{A})]}{\operatorname{Tr}[\Lambda_{B}\omega_{B}
]}\\
& =\frac{\operatorname{Tr}[\mathcal{P}_{A\rightarrow B}^{\dag}(\Lambda
_{B})(\rho_{A})]}{\operatorname{Tr}[\mathcal{P}_{A\rightarrow B}^{\dag
}(\Lambda_{B})(\sigma_{A})]}\frac{\operatorname{Tr}[\Lambda_{B}\mathcal{P}
_{A\rightarrow B}(\sigma_{A})]}{\operatorname{Tr}[\Lambda_{B}\omega_{B}]}\\
& \leq\left(\sup_{\alpha^{-1}I_{A}\leq\Gamma_{A}\leq I_{A}}\frac
{\operatorname{Tr}[\Gamma_{A}(\rho_{A})]}{\operatorname{Tr}[\Gamma_{A}
(\sigma_{A})]}\right)\cdot\left(\sup_{\alpha^{-1}I_{B}\leq\Lambda_{B}\leq
I_{B}}\frac{\operatorname{Tr}[\Lambda_{B}\mathcal{P}_{A\rightarrow B}
(\sigma_{A})]}{\operatorname{Tr}[\Lambda_{B}\omega_{B}]}\right)\\
& =\sup{}_{\alpha}(\rho_{A}/\sigma_{A})\cdot\sup{}_{\alpha}(\mathcal{P}
_{A\rightarrow B}(\sigma_{A})/\omega_{B}).
\end{align}
Since the inequality holds for all $\Lambda_{B}$ such that $\alpha^{-1}I_{B}\leq\Lambda_{B}\leq I_{B}$, we conclude that
\begin{equation}
\sup{}_{\alpha}(\mathcal{P}_{A\rightarrow B}(\rho_{A})/\omega_{B})\leq\sup{}_{\alpha}(\rho_{A}/\sigma_{A})\cdot\sup{}_{\alpha}(\mathcal{P}_{A\rightarrow B}(\sigma_{A})/\omega_{B}).
\end{equation}
Finally, we take a logarithm and multiply by $\alpha/\left(  \alpha-1\right)
$ to conclude the statement of the lemma.
\end{proof}

By the same proof that we gave for the max-relative entropy in Proposition~\ref{lem:D_max} and using the fact that the Hilbert $\alpha$-divergence is strongly faithful \cite[Theorem 1(i)]{BG17}, we conclude that
there is an amortization collapse for the Hilbert $\alpha$-divergence of
quantum channels. As special case, we conclude that the diamond norm of the
difference of two channels does not increase under amortization.

\begin{prop}
Let $\mathcal{N}_{A\rightarrow B},\mathcal{M}_{A\rightarrow B}\in
\mathcal{Q}(A\rightarrow B)$. Then for all $\alpha\geq1$, we have the following amortization collapse:
\begin{equation}
H_{\alpha}^{\mathcal{A}}(\mathcal{N}\Vert\mathcal{M})=H_{\alpha}(\mathcal{N}\Vert\mathcal{M}).
\end{equation}
\end{prop}

We can establish related results for the $c$-distance  and the Bures distance of quantum states, both of which are based on the quantum fidelity. For states $\rho$ and $\sigma$, the $c$-distance \cite{R02,R03,GLN04,R06} and Bures distance \cite{Bures1969} are respectively defined as
\begin{equation}
c(\rho,\sigma) \coloneqq \sqrt{1-F(\rho,\sigma)}, \qquad B(\rho\|\sigma) \coloneqq \sqrt{2\left(1-\sqrt{F(\rho,\sigma)}\right)}.
\end{equation}
(In the above and what follows, we use the notation $c(\rho,\sigma)$ for $c$-distance in order to differentiate this quantity from the Chernoff divergence $C(\rho\|\sigma)$.)
The same proof as in \eqref{eq:dp-tri-trace-1}--\eqref{eq:dp-tri-trace-2}, along with the fact that the quantum fidelity is monotone with respect to positive, trace-preserving maps \cite[Corollary~A.5]{Mosonyi2015}, implies that the following data-processed triangle inequalities hold:
\begin{lemma}
[Data-processed triangle inequalities]
Let $\mathcal{P}_{A\rightarrow B}$ be a
positive trace-preserving map, and let $\rho_{A},\sigma_{A}\in S(A)$ and
$\omega_{B}\in S(B)$. Then the following inequalities hold:
\begin{align}
c(\mathcal{P}_{A\rightarrow B}(\rho_{A}),\omega_{B})& \leq
c(\rho_{A},\sigma_{A})+c(\mathcal{P}_{A\rightarrow
B}(\sigma_{A}),\omega_{B}),\\
B(\mathcal{P}_{A\rightarrow B}(\rho_{A})\|\omega_{B}) & \leq
B(\rho_{A}\|\sigma_{A})+B(\mathcal{P}_{A\rightarrow
B}(\sigma_{A})\|\omega_{B}).
\end{align}
\end{lemma}

By the same reasoning as above, we then conclude that the induced channel divergences do not increase under amortization:
\begin{prop}
Let $\mathcal{N}_{A\rightarrow B},\mathcal{M}_{A\rightarrow B}\in
\mathcal{Q}(A\rightarrow B)$. Then we have the following amortization collapses:
\begin{equation}
c^{\mathcal{A}}(\mathcal{N},\mathcal{M})=c(\mathcal{N},\mathcal{M}), \qquad B^{\mathcal{A}}(\mathcal{N}\|\mathcal{M})=B(\mathcal{N}\|\mathcal{M}).
\end{equation}
\end{prop}


\section{Generalized Fuchs-van-de-Graaf inequality}\label{app:gen-fuchs-vdgrf}

A well known inequality in quantum information theory is the following Fuchs-van-de-Graaf inequality \cite{FG99}:
\begin{equation}
\frac{1}{2}\left\Vert \rho-\sigma\right\Vert _{1}\leq\sqrt{1-F(\rho,\sigma)},
\end{equation}
which holds for density operators $\rho$ and $\sigma$, and $F(\rho,\sigma)\coloneqq\Vert\sqrt{\rho}\sqrt{\sigma}\Vert
_{1}^{2}$. The following lemma, proved in \cite[Supplementary~Lemma~3]{CKW14}, generalizes this relation to the case of
positive semi-definite operators $A$ and $B$, and it also represents a tighter
bound than that given in \cite[Theorem 7]{Audenaert2008}. (Note that 
\cite[Theorem 7]{Audenaert2008} generalizes one of the inequalities in \cite[Equation (1)]{Kholevo1972}
to positive semi-definite operators.) The proof of Lemma~\ref{lem:FvdG-PSD} that we give below is
very similar to the proof of Theorem 7 of \cite{Audenaert2008}, but it features a
minor change in the reasoning. The proof is also different from that given in \cite[Supplementary~Lemma~3]{CKW14}.

\begin{lemma}[\cite{CKW14}]\label{lem:FvdG-PSD}
For positive semi-definite, trace class operators $A$ and $B$ acting on a separable Hilbert space, we have that
\begin{equation}
\left\Vert A-B\right\Vert _{1}^{2}+4\left\Vert A^{1/2}B^{1/2}\right\Vert_{1}^{2}\leq\big(\operatorname{Tr}[A+B]\big)^2.
\end{equation}
\end{lemma}

\begin{proof}
For convenience, we give a complete proof and follow the proof of Theorem 7 of
\cite{Audenaert2008} quite closely. Consider two general operators $P$ and $Q$, and
define their sum and difference as $S=P+Q$ and $D=P-Q$. Then $P=\left(
S+D\right)  /2$ and $Q=\left(  S-D\right)  /2$. Consider that
\begin{align}
PP^{\dag}-QQ^{\dag}  & =\frac{1}{4}\left(\left(  S+D\right)  \left(
S+D\right)  ^{\dag}-\left(  S-D\right)  \left(  S-D\right)  ^{\dag}\right)
\\
& =\frac{1}{2}\left(SD^{\dag}+DS^{\dag}\right).
\end{align}
Then, we have
\begin{align}
\left\Vert PP^{\dag}-QQ^{\dag}\right\Vert _{1}  & =\frac{1}{2}\left\Vert
SD^{\dag}+DS^{\dag}\right\Vert _{1}\\
& \leq\frac{1}{2}\big(\left\Vert SD^{\dag}\right\Vert _{1}+\left\Vert
DS^{\dag}\right\Vert _{1}\big)\\
& =\left\Vert SD^{\dag}\right\Vert _{1}\\
& \leq\left\Vert S\right\Vert _{2}\left\Vert D\right\Vert _{2}.
\end{align}
Now pick $P=A^{1/2}U$ and $Q=B^{1/2}$, where $U$ is an arbitrary unitary. Then
$S,D=A^{1/2}U\pm B^{1/2}$, and we find that
\begin{equation}
\left\Vert A-B\right\Vert _{1}\leq\left\Vert A^{1/2}U+B^{1/2}\right\Vert
_{2}\left\Vert A^{1/2}U-B^{1/2}\right\Vert _{2}.
\end{equation}
Squaring this gives
\begin{align}
&\left\Vert A-B\right\Vert _{1}^{2} 
\notag\\
&  \leq\left\Vert A^{1/2}U+B^{1/2}\right\Vert _{2}^{2}\left\Vert A^{1/2}U-B^{1/2}\right\Vert _{2}^{2}\\
&  =\operatorname{Tr}\Big[(A^{1/2}U+B^{1/2})^{\dag}(A^{1/2}U+B^{1/2}
)\Big]\cdot\operatorname{Tr}\Big[(A^{1/2}U-B^{1/2})^{\dag}(A^{1/2}U-B^{1/2})\Big]\\
&  =\operatorname{Tr}\Big[A+B+B^{1/2}A^{1/2}U+U^{\dag}A^{1/2}B^{1/2}
\Big]\cdot\operatorname{Tr}\Big[A+B-B^{1/2}A^{1/2}U-U^{\dag}A^{1/2}B^{1/2}\Big]\\
&  =\left(\operatorname{Tr}[A+B]+2\operatorname{Re}\Big\{\operatorname{Tr}\Big[B^{1/2}A^{1/2}\Big]\Big\}\right)\Big(\operatorname{Tr}[A+B]-2\operatorname{Re}\Big\{\operatorname{Tr}\Big[B^{1/2}A^{1/2}U\big]\Big\}\Big)  \\
&  =\left(\operatorname{Tr}[A+B]\right)  ^{2}-4\left(\operatorname{Re}\Big\{\operatorname{Tr}\Big[B^{1/2}A^{1/2}U\Big]\Big\}\right)^{2}.
\end{align}
Note that the unitary $U$ in the above is arbitrary. So we can finally pick the unitary $U$ to be the operator from the polar decomposition of $B^{1/2}A^{1/2}$ as
\begin{equation}
B^{1/2}A^{1/2}U=\sqrt{B^{1/2}AB^{1/2}}.
\end{equation}
Then, we get
\begin{align}
\left\Vert A-B\right\Vert _{1}^{2} &  \leq\left(\operatorname{Tr}[A+B]\right)^{2}-4\left(\operatorname{Re}\Big\{\operatorname{Tr}\Big[\sqrt{B^{1/2}AB^{1/2}}\Big]\Big\}\right)^{2}\\
&  =\left(\operatorname{Tr}[A+B]\right)^{2}-4\left\Vert A^{1/2}B^{1/2}\right\Vert _{1}^{2}
\end{align}
and the proof is concluded.
\end{proof}


\bibliographystyle{unsrt}
\bibliography{Bib}

\end{document}